\documentclass[format=acmsmall, review=false]{acmart}
\usepackage{ec26-submission-style-files/acm-ec-26}
\usepackage{booktabs} 
\usepackage[ruled,noend]{algorithm2e} 

\SetAlFnt{\small}
\SetAlCapFnt{\small}
\SetAlCapNameFnt{\small}
\SetAlCapHSkip{0pt}
\IncMargin{-\parindent}
\usepackage{outlines}
\usepackage{placeins}
\usepackage{wrapfig}
\usepackage{makecell}
\usepackage{graphicx}
\usepackage{subcaption}
\usepackage{cleveref}
\usepackage{amsmath}
\DeclareMathOperator*{\argmax}{arg\,max}
\usepackage{multirow}
\usepackage{xspace}
\usepackage{tikz}
\usetikzlibrary{shapes.geometric} 
\usetikzlibrary{positioning}
\usetikzlibrary{calc}
\usepackage{thmtools}
\usepackage{rotating}
\usepackage{pdflscape}
\usepackage{booktabs}
\usetikzlibrary{positioning,fit,calc,backgrounds}

\newcommand{\Cov}{\mathrm{Cov}}

\usepackage{xcolor}

\newcommand{\students}{\mathcal{S}}
\newcommand{\programs}{\mathcal{P}}

\newcommand{\nudgecaps}{\mathcal{C}}
\newcommand{\programnudgecap}{\nudgecaps_j}

\newcommand{\applicationprofiles}{\succ}
\newcommand{\nudgeallocation}{X}
\newcommand{\matchsimulator}{M}
\newcommand{\matchoutcome}{\mu}
\newcommand{\matchoutcomeonnudges}{\mu_{\nudgeallocation}}

\newcommand{\programlossfxn}{\ell_j}
\newcommand{\adoptionfxn}{f}

\newcommand{\Ind}{\mathbf{1}}

\newcommand{\naive}{\textbf{Naive}\xspace}
\newcommand{\optimizelinear}{\textbf{Capacitated}\xspace}
\newcommand{\optimizesqrt}{\textbf{Balanced Capacitated}\xspace}
\newcommand{\high}{\ifmmode\mathrm{high}\else\textit{High}\xspace\fi}
\newcommand{\nudged}{\ifmmode\mathrm{nudged}\else\textup{nudged}\xspace\fi}
\newcommand{\medium}{\textit{Medium}\xspace}
\newcommand{\low}{\textit{Low}\xspace}
\newcommand{\methodabbr}{CNA\xspace}
\newcommand{\algoname}{NudgeCapComputation\xspace}

\newcommand{\blfootnote}[1]{
  \begingroup
  \renewcommand{\thefootnote}{}
  \footnotetext{#1}
  \addtocounter{footnote}{-1}
  \endgroup
}

\theoremstyle{acmdefinition}
\newtheorem{assumption}{Assumption}

\newtheorem{definition}{Definition}
\crefname{assumption}{assumption}{assumptions}
\Crefname{assumption}{Assumption}{Assumptions}

\setcitestyle{authoryear}
\AtBeginDocument{
  \pagestyle{fancy}
  \fancyhead{}
  \fancyhead[R]{\thepage}
  \fancyhead[L]{\small\sffamily Personalized Recommendations Without Inducing Congestion}
  \fancyfoot{}
}

\newcommand{\titletext}{Personalized Recommendations Without Inducing Congestion: Mitigating Disparities in the NYC High School Match}

\title[\titletext]{\titletext}

\author{Erica Chiang}
\affiliation{
  \institution{Cornell University}
  \city{New York}
  \country{USA}
  }

\author{Kenny Peng}
\affiliation{
  \institution{Cornell University}
  \city{New York}
  \country{USA}
}

\author{Rebecca Lichtenstein}
\affiliation{
  \institution{New York City Public Schools}
    \city{New York}
  \country{USA}
}

\author{Brielle McDaniel}
\affiliation{
  \institution{New York City Public Schools}
    \city{New York}
  \country{USA}
}

\author{Kristen O'Neil}
\affiliation{
  \institution{New York City Public Schools}
    \city{New York}
  \country{USA}
}

\author{Dejaunique Thomas}
\affiliation{
  \institution{New York City Public Schools}
    \city{New York}
  \country{USA}
}

\author{Lianna Wright}
\affiliation{
  \institution{New York City Public Schools}
    \city{New York}
  \country{USA}
}

\author{Jon Kleinberg}
\affiliation{
  \institution{Cornell University}
    \city{Ithaca}
  \country{USA}
}

\author{Eva Tardos}
\affiliation{
  \institution{Cornell University}
      \city{Ithaca}
  \country{USA}
}

\author{Nikhil Garg}
\affiliation{
  \institution{Cornell University}
    \city{New York}
  \country{USA}
}

\begin{abstract}

Algorithmic recommendations can help participants navigate large matching markets. For example, recommendations for school and college choices may reduce information frictions and disparities in access to high-performing programs. At scale, however, recommenders in capacity-constrained settings can be self-defeating: if they steer too many users toward the same items, then even users who were originally predicted to have a high chance of matching to an item may not, due to increased competition. In this paper, we formalize this phenomenon of \textit{recommendation-induced congestion}; motivated by the NYC high school match, we show that naive recommendations can cause sharp decreases in program acceptance rates, most affecting applicants with the fewest nearby options. Next, we propose and theoretically analyze a congestion-aware, bilevel optimize-and-simulate approach to allocate recommendations and improve match outcomes \textit{safely, in equilibrium}. Finally, we deploy this approach in the 2025-26 admissions cycle of the NYC high school match, aiming to reduce disparities by highlighting personalized lists of nearby, high-performing programs where an applicant has a high predicted offer likelihood. In a randomized controlled trial, we find that 16.4\% of treatment applicants ranked a recommended program, versus 10.5\% of control applicants who ranked a program they would have been recommended (57\% relative increase; $p$=0.011); 5.6\% of treatment applicants matched to such a program, versus 3.3\% of control applicants (71\% relative increase; $p$=0.071); further, no treatment applicant was rejected from a recommended program. Our findings suggest that recommenders should be analyzed and designed as market-shaping interventions. 

\end{abstract}

\begin{document}
\renewcommand{\shortauthors}{Chiang et al.}

\maketitle
\today
\section{Introduction}

\blfootnote{Corresponding authors: Erica Chiang (esc99@cornell.edu) and Nikhil Garg (ngarg@cornell.edu).}

People navigate large matching markets to find educational opportunities, jobs, and partners. At scale, these markets can increase access to opportunity, but it may be difficult to actually identify suitable matches. In turn, many platforms use \textit{recommender systems} to guide participants. While these recommenders are often designed to optimize individual-level outcomes, they have \textit{systemic} effects---one participant's recommendations impose externalities on others, either their potential matches or competitors. In our specific context, the New York City (NYC) high school match, eighth-grade students apply to high school by submitting a ranked list of programs, and a centralized algorithm matches each student to one program \cite{abdulkadiroglu_new_2005, gale1962college}. While the system was designed to expand access, forming preferences over programs is a complex, challenging process \cite{marian2023algorithmic}. Students must choose from over 900 high school programs, which vary widely in academic focus and performance, transportation requirements, and extracurricular activities. 

Forming informed preferences thus requires substantial time, information, and support. Given heterogeneous resources to navigate these informational burdens, the application complexities may perpetuate disparities \cite{broyard_nyc_2023, corradini2023information, cohodes2022informational, Corcoran18, idoux2024overcoming}. For example, \citet{pengUndermatchingDisparitiesPortfolio} find that Black, Hispanic, and low-income students in NYC disproportionately \textit{undermatch}: they could have matched to programs that measured higher on certain metrics (program selectivity, academic performance, impact, and other measures of quality), while being at least as close to home as their true match, had they applied; the largest gaps were among the students with the most competitive applications. 

Such disparities have motivated several informational interventions in NYC \cite{cohodes2022informational, Corcoran18}. One prominent approach used in these interventions is to provide applicants at target middle schools with \textit{non-personalized} lists of nearby, high-performing programs. Because every student within a given middle school receives the same list, such an intervention must either (a) be restricted to programs that most students would be admitted to (e.g., historically undersubscribed programs), or (b) risk steering many students to programs where they are ultimately rejected. The former limits the set of high schools that can be included, since many high-performing programs are competitive and require high grades and favorable lottery numbers. The latter can induce disappointment and may reduce trust in the intervention in future years.\footnote{\citet{Corcoran18} and \citet{cohodes2022informational} navigate this tradeoff by providing all students a list of about 25 programs at varying selectivity levels with guidance on implementing a ``reach, match, safety'' approach \cite{ali2021college}. \citet{cohodes2022informational} also provide interactive search and recommendation tools using preferences such as for extracurriculars.} 

Recent advances in personalized predictions of admissions outcomes enable moving beyond these approaches, toward personalized guidance regarding where a student is likely to be admitted. Beginning in the 2024-25 cycle, in work by \citet{corradini2025school}, NYC Public Schools (NYCPS) has provided all high school applicants with personalized, program-specific offer likelihood estimates during the application process. These likelihoods are bucketed into \{\high, \medium, \textit{Low}\}, with applicants told that \high corresponds to at least a 99\% chance of admission and \low to a less than 1\% chance \cite{nycdoeOfferChances2026}. Similar tools have recently been integrated into other school matching platforms \cite{arteaga2022smart, larroucau2024college}, including program-specific recommendations that depend on such predictions. More generally, predictions of downstream match outcomes are also increasingly incorporated into recommender systems in settings such as labor and dating platforms.

\paragraph{Recommendation-induced congestion.} The above discussion overlooks an important concern: recommender systems are not passive information tools. Match outcomes are \textit{endogenous} to all applicants' behavior, and an intervention's goal is to change behavior. Thus, individual-level interventions may be self-defeating: if many applicants are steered toward the same program, many---even those with \high estimated offer likelihood---may be rejected due to new competition that was not accounted for in the original likelihood estimation. We refer to this phenomenon as \textit{recommendation-induced congestion}. 
For example, \citet{pengUndermatchingDisparitiesPortfolio} identify individually-optimal application changes that would improve certain metrics of each student's matched program if adopted in isolation; however, they find in simulations that when many students adopt their ``optimal changes'' simultaneously, many students end up not matching at all. 
In this work, we further show in simulations that recommendation-induced congestion itself can be disparate: students with less access to high-performing programs are more likely to be affected by congestion at the programs they \textit{are} eligible for, due to geographic correlation.

\paragraph{This work.} We formalize the challenge of and propose an approach to allocating recommendations without inducing congestion. 
We apply our approach and leverage offer likelihood estimates to deploy a \textit{personalized informational intervention}. In partnership with NYCPS, we sent applicants personalized lists of high-performing, nearby programs where they have \high predicted offer likelihoods---focusing on applicants from middle schools that historically did not send many students to such schools. 
In addition to the overarching intervention goal of addressing disparities in applications to high-performing programs, a key design constraint in our deployment was to avoid recommendation-induced congestion. Our recommendations needed to preserve offer likelihood estimates even if many students adopted them; this was a non-negotiable constraint for our NYCPS partners, in order to prevent misleading applicants and to retain trust in the offer likelihood tool, future interventions, and the match process broadly. Our contributions are as follows. 

\textbf{Characterization of recommendation-induced congestion.} Through empirically-calibrated simulations, we show that naively recommending each program to all eligible students can induce congestion  at scale, failing to preserve offer likelihoods. The effects are heterogeneous, most affecting applicants with less \textit{ex ante} access to high-performing programs,  reinforcing disparities. 

\textbf{Method for safe recommendations in equilibrium.} To mitigate such effects, we propose a bi-level optimization method for allocating recommendations 
in capacity-constrained settings, safely. We manage congestion by capping the number of nudges sent to each program, using historical data to determine how high nudge capacities can be while meeting congestion constraints. Then, at deployment time, we optimize nudge allocations given the nudge capacities. Theoretically, we show that under natural assumptions, our approach approximately maximizes expected desirable matches, while preserving acceptance rates. Empirically, our method achieves target acceptance rates without reducing overall intervention effectiveness compared to the naive approach. 

\textbf{Deployment and RCT of personalized recommendations.} We deployed our method in an informational intervention for the 2025-26 NYC high school admissions cycle, providing applicants from target middle schools with \textit{personalized} recommendations to nearby, high-performing programs they are likely to be admitted to even if recommendations are adopted at scale. We evaluate our intervention through a pre-registered randomized controlled trial (RCT), finding substantial take-up of our recommendations among treatment applicants: 16.4\% of the 444 treatment applicants apply to a program they were nudged to, while 10.5\% of the 822 control applicants apply to a program they would have been recommended---a 57\% increase ($p$=0.011). We find suggestive evidence of an increase in \textit{matches} to nudged programs: 5.6\% of treatment applicants match to a nudged program, as opposed to 3.3\% for control applicants, a 71\% relative increase ($p$=0.071). Furthermore, no treatment applicants were rejected from a program they were nudged to,\footnote{The dropoff from 16.4\% to 5.6\% is fully explained by matches to programs listed higher than the nudged program.} and only one control applicant was rejected from a program they \textit{would have} been nudged to. 

We further apply our framework to study how big an effect such demand-side interventions can have. While our discussion in this paper is motivated by school placement due to our deployment, these considerations also arise in other capacity-constrained settings, such as labor and dating platforms: individual-level interventions in any of these settings should consider system-wide effects for safe, effective outcomes. 

\section{Related work}
This work is motivated by substantial documentation of information frictions in school choice \cite{jensen2010perceived, wiswall2015determinants, bleemer2018intended, ainsworth2023households, hastings2015effects, hastings2016informed, walters2018demand, jochim2014parents}, the resulting suboptimal application behavior \cite{marian2023algorithmic, tomkins2023showing, larroucau2024college, arteaga2022smart}, and the existence of \textit{disparities} in behavior and outcomes \cite{broyard_nyc_2023, Corradini2023HighSchoolChoices, pengUndermatchingDisparitiesPortfolio, singer2024detroitschoolchoice, hoxby2015college, hoxby2012missing, dynarski2006cost, corradini2023information, idoux2024overcoming, sattin2018surviving, condliffe2015stuck,garg2025heterogeneous}.
More broadly, a large literature considers the design of school assignment systems and empirically evaluates their performance, including a focus on disparities \cite{allman2022designing, jeloudar2025optimization, immorlica2020information, awadelkarim2023rank,angrist2024credible,dur2018reserve}. 

\textit{Informational interventions for school choice.} Many informational interventions have aimed to increase awareness (and reduce disparities) in school assignment, both in NYC \cite{Corcoran18, cohodes2022informational} and elsewhere \cite{arteaga2022smart,hastings2007preferences, hastings2008information, agte2024search, larroucau2024college, ye2024improving, campos2024preferenceslosangeles, valant2022neworleans, ainsworth2020valueadded,allende2019approximating, arteaga2022recommendations, neilson2026decade, elacqua2025promise}. They have provided students with information on school performance, value added, admissions chances, and/or outcomes. 
While some interventions deploy personalized search and recommendation tools, to our knowledge none of them have directly controlled recommendation-induced congestion. Most relevant to ours is the work of \citet{corradini2025school}, as discussed in the introduction; in NYC, their offer likelihood tool provides warnings to applicants in the online portal if they do not apply to any program where they have \high likelihood. This is intended to prevent students from failing to match at all, as opposed to preventing {undermatching}. \citet{arteaga2022recommendations} provide program-specific recommendations that depend on such predictions. \citet{larroucau2024college} provide personalized \textit{major} recommendations, but not \textit{program} recommendations, and further include warnings if an applicant is highly likely to be admitted to their first-choice option (suggesting that they may consider listing a more competitive program first). Further, some of these systems use offer likelihoods or other characteristics for their \textit{search} features. 

A benefit of using personalized offer likelihoods is that we can include more competitive programs: in our deployment, performance cutoffs defining high-performing schools were higher than those used in prior NYC interventions, and over half of the recommended programs were oversubscribed the previous year (had more applicants than seats) and thus would have been excluded in a non-personalized intervention limited to programs where most would be admitted.

Other past work has analyzed the prediction of admissions outcomes in similar settings \cite{dong2026capacity,pathak2021well}. In particular, \citet{dong2026capacity} show that machine learning-based predictions can become unstable when the applicant pool changes. 
Among other things, this implies the need to account for equilibrium effects when shifting application behavior in capacity-constrained allocation systems. 
Our work designs recommendations that explicitly account for these equilibrium effects when guiding applicants’ choices.

\textit{Recommendation in equilibrium.} While classical recommendation and decision support tools have historically been designed for uncapacitated environments, there is growing interest in accounting for scarcity, congestion, and fairness in recommendation settings. 
Recent work has developed methods for explicitly incorporating downstream matching and fairness considerations into recommendation design by balancing item exposure \cite{mashayekhi2023recon, christakopoulou2017recommendation, su2022optimizing, shi2025optimal,greenwood2024user,chen2025feature,NEURIPS2024_59d4e18a}; used assortment optimization to dynamically show a restricted set of items to each user, while achieving final outcomes that maximize matches or revenue and are aligned with user preferences \cite{ashlagi2022assortment, torrico2020multi, rios2023platform, housni2024two}; studied mechanisms for reducing congestion by changing incentives \cite{he2022applicationcosts}; and modeled congestion in capacity-constrained settings to characterize design tradeoffs \cite{bied2021congestion,bied2023toward,besbes2025impact,kanoria2021facilitating, arnosti2019managingcongestion, escobar2025congestion}. One setting with similar considerations of balancing allocations in a stage \textit{before} a matching is the interview stage for medical residencies, with proposed approaches to reduce congestion
\cite{ashlagi2017communication, ashlagi2025stable, allman_interviewing_2023, wapnir2021explaining}. 

While the above papers are largely theoretical, several recommender deployments have been congestion aware or estimated congestion ex post. \citet{rios2023improving} propose dynamic exposure policies and improve match rates in two dating markets, \citet{manshadi2025redesigning} develop and deploy a model for rankings in a large volunteer matching platform, \citet{huang2022managingcongestion} improve match efficiency in a dating market by controlling disclosure of demand information, {\citet{behaghel2024potential} deploy congestion-aware recommendations with French public employment services to maximize potential employment,} and \citet{chen2023reducing} develop an equilibrium-aware recommendation approach and apply it to a large online dating service. Other recommender deployments report congestion effects: \citet{neilson2026decade} documents congestion from information and recommendation tools in Chile's school choice system; \citet{altmann2026jobsearch} provide online job search advice at scale and find that employment and earnings gains diminish when more people receive treatment; and \citet{bobba2023perceived} provide applicants with feedback about their own academic performance to drive skill-based sorting into high school tracks, using the experimental results to model scale-up and finding that congestion can have negative downstream effects, including on equity.

Our work builds on this literature by developing, deploying, and evaluating a method for recommendation in equilibrium under hard capacity and offer likelihood constraints.

Recommendation-induced congestion is also closely related to concerns of algorithmic monoculture in hiring and other matching markets, in which employers using the same hiring algorithms may induce congestion on who they interview and make offers to \cite{kleinberg2021algorithmic,peng24monoculture,bommasani2022picking,creel2022algorithmic,bommasani2026algorithmic,peng2024wisdom,kim2025correlated}. 
Related work also studies the design of information displays in settings where agents respond strategically to shared information in the presence of congestion \cite{kayongo2021visualization, zhang2024designing}. Relatedly, recommendations based on offer likelihoods changing the offer likelihoods themselves is an example of performative prediction \cite{perdomo2020performative}. In \Cref{sec:relatedworkadditional}, we further discuss the relationship to work on experimentation in equilibrium.

\section{Safe Recommendation in capacity-constrained settings}
\label{sec:method}
The overarching goal of an intervention is to shift application behavior to increase matches among a specified set of ``desirable matches.'' We consider interventions that achieve these shifts by nudging students to programs they would not otherwise have applied to. However, because programs have limited capacity, changing demand can also change acceptance rates. In \Cref{subsec:consequences_of_congestion}, we characterize the consequences of failing to account for these shifts. In \Cref{subsec:problem_definition}, we formalize our goal of safely allocating recommendations (``nudges'') to achieve improved match outcomes without inducing congestion. Finally, in \Cref{subsec:allocation}, we develop a solution for this concern---a heuristic bilevel optimization approach to compute such an allocation---and theoretically analyze its performance. \Cref{fig:recommendation_induced_congestion} overviews the problem and technical approach.

\begin{figure}[t]
\centering
\begin{subfigure}{\textwidth}
\resizebox{\textwidth}{!}{
\begin{tikzpicture}[
    box/.style={
        draw,
        rounded corners,
        align=center,
        text width=3.0cm,
        minimum height=2.5cm,
        inner sep=8pt,
        font=\small
    },
    arrow/.style={
        ->,
        thick
    },
    title/.style={
        font=\bfseries\small,
        align=center
    }
]

\node[box, draw=blue!70, fill=blue!8] (predict) {
    \textbf{1. Prediction}\\[0.35em]
    Offer likelihood tool predicts \high chance at a program given historical data
};

\node[box, draw=green!60!black, fill=green!12, right=0.8cm of predict] (recommend) {
    \textbf{2. Recommendation}\\[0.35em]
    Nudge students to programs where they have \high offer likelihood\\[0.35em]
};

\node[box, draw=blue!70, fill=blue!8, right=0.8cm of recommend] (adopt) {
    \textbf{3. Adoption}\\[0.35em]
    Many students follow the same recommendation
};

\node[box, draw=blue!70, fill=blue!8, right=0.8cm of adopt] (competition) {
    \textbf{4. Congestion}\\[0.35em]
    Competition for recommended program increases
};

\node[box, draw=blue!70, fill=blue!8, right=0.8cm of competition] (invalid) {
    \textbf{5. Lower \\offer rates}\\[0.35em]
    Offer likelihood predictions no longer hold; decreased trust
};

\draw[arrow] (predict) -- (recommend);
\draw[arrow] (recommend) -- (adopt);
\draw[arrow] (adopt) -- (competition);
\draw[arrow] (competition) -- (invalid);

\draw[arrow, dashed, color=green!60!black, transform canvas={yshift=0.25cm}] (recommend) -- (adopt);
\draw[arrow, dashed, color=green!60!black, transform canvas={yshift=0.25cm}] (adopt) -- (competition);
\draw[arrow, dashed, color=green!60!black, transform canvas={yshift=0.25cm}] (competition) -- (invalid);

\draw[
    arrow,
    dashed,
    color=green!60!black,
    bend right=20
] (invalid.north west) to (recommend.north east);

\end{tikzpicture}
}
\caption{Our method simulates outcomes (green arrows) to generate recommendations that mitigate congestion.}
\label{fig:recommendation_induced_congestion}
\end{subfigure}

\begin{subfigure}{\textwidth}
    \centerline{\includegraphics[width=0.77\textwidth]{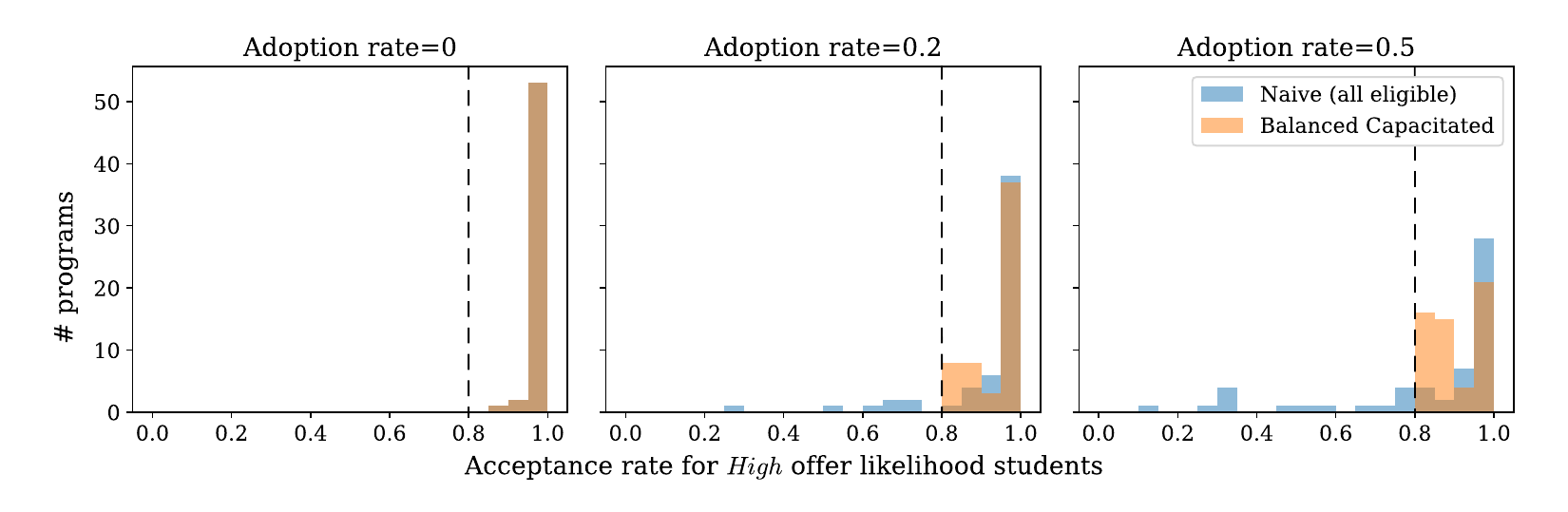}}
    \caption{Simulated acceptance rates at 3 different simulated adoption rates.}
    \label{fig:acceptance_rates_hist}
\end{subfigure}
\caption{\textbf{Recommendation-induced congestion.} (a) shows a conceptual diagram of recommendation-induced congestion and how our method mitigates congestion. (b) shows acceptance rates at nudge-eligible programs, at various rates of simulated nudge adoption. Under \naive nudge allocation (blue), simulated acceptance rates among \high offer likelihood students drop substantially with adoption. We compare these to our \optimizesqrt nudge allocation approach designed to keep these rates $\geq$ 80\% (orange).}
\end{figure}
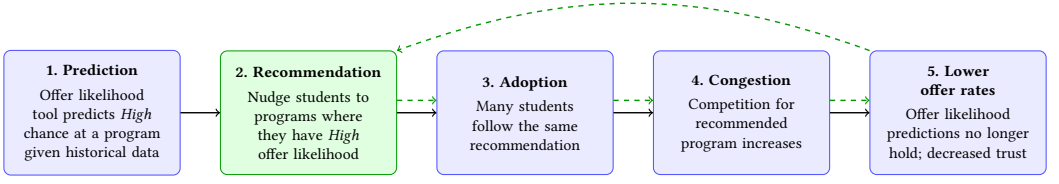
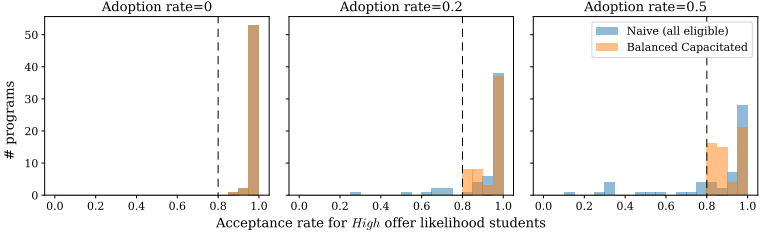

\subsection{Motivation: consequences of recommendation-induced congestion}
\label{subsec:consequences_of_congestion}
To illustrate recommendation-induced congestion, we simulate\footnote{In our setting, we can accurately \textit{simulate} match outcomes given rankings. Then, using historical data on rankings and assumptions about how applicants respond to nudges, we can evaluate match outcomes for any nudge allocation.} the potential consequences of a naive approach to recommendation using data from the 2024-25 NYC high school admissions cycle. (Detailed descriptions of the data and setting appear in \Cref{sec:datasetup}; here, we discuss results to motivate our method.) Recall that in our deployed intervention, our desirable matches are between students from target middle schools and nearby, high-performing high school programs where they have \high offer likelihood. The naive intervention approach nudges each student to \textit{all} programs that meet these nudge-eligible criteria (i.e., all pairs forming a desirable match).\footnote{\Cref{subsec:nonpersonalized_nudges} further evaluates the non-personalized approach used in prior interventions (every student at each middle school receives the same list of all nearby, high-performing programs), which induces even more congestion.} 

\Cref{fig:acceptance_rates_hist} shows histograms of simulated acceptance rates, among students with \high pre-nudge offer likelihoods across programs, under different assumed rates of students adopting nudges. In blue is the naive approach of nudging all students to all eligible programs (we defer discussion of orange to \Cref{sec:simulation_results}).
Adoption of nudges causes steep \textit{drop-offs in program acceptance rates}, even among students who ex ante were told they have \high (at least 99\%) likelihood of being admitted. 

Furthermore, congestion can induce \textit{disparate effects} when applicants have unequal access to high-performing programs. In NYC, there are large geographic and income-based disparities in the number of high-performing programs accessible to applicants (illustrated in \Cref{fig:access_disparities}). Under a naive recommendation approach, congestion is disproportionately high at programs in areas where students tend to have less access to other options. \Cref{fig:match_rates} illustrates this effect in simulation: applicants with less \textit{ex ante} access to these programs have the lowest match rates to nudged programs, even when nudged.

Thus, naive nudge allocation would fail to preserve offer likelihoods, with disparate effects. This would undermine both trust in the offer likelihood tool, which is designed to help students more effectively navigate program options, and would also diminish the effectiveness of the intervention at scale. Our method is designed to address both of these effects.

\begin{figure}[t]
  \centering
  \begin{subfigure}[t]{0.54\linewidth}
    \centering
    \includegraphics[width=0.95\linewidth]{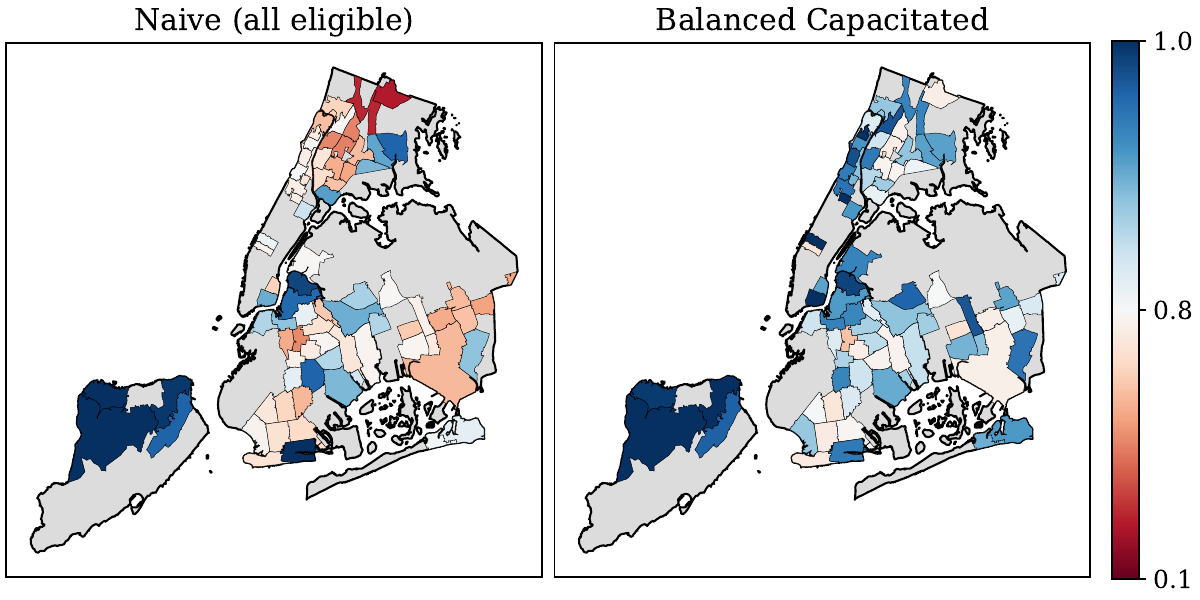}
    \caption{Simulated match rates by zip code. Gray zip codes have $\leq 2$ adopting students.}
    \label{fig:match_rate_geographic}
  \end{subfigure}
  \hfill
  \begin{subfigure}[t]{0.42\linewidth}
    \centering
    \includegraphics[width=0.8\linewidth]{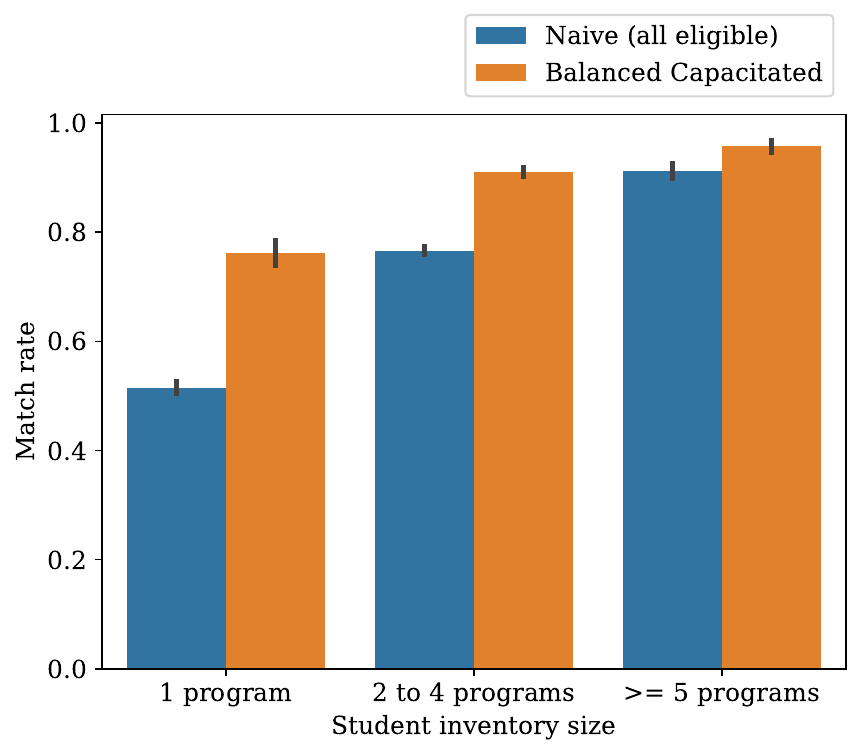}
    \caption{Simulated match rates by inventory size (access to high-performing programs).}
    \label{fig:match_rate_invsize}
  \end{subfigure}
  \caption{\textbf{Match rates to nudged programs at 20\% adoption.} In simulation, the \naive approach of nudging all eligible students to all eligible high school programs leads to \textit{steep, disparate} drop-offs in the match rates of nudged students to their nudged programs, despite the fact that they originally had \high offer likelihoods.}
  \label{fig:match_rates}
\end{figure}

\subsection{Model and Problem Definition}
\label{subsec:problem_definition}

Suppose we have some exogenously-specified set of ``desirable matches'' $Z$, where each binary $Z_{ij}$ indicates whether or not it would be desirable for student $i$ to match to program $j$. The \textit{intervention goal} is to increase the number of matches $(i,j)$ with $Z_{ij}=1$. In our setting, a student-program match is desirable ($Z_{ij}=1$) only if it meets certain ``nudge-eligibility'' criteria: an applicant $i$ attends a target middle school, has \high offer likelihood at a high-performing program $j$, and lives near $j$. (\Cref{sec:datasetup} contains a detailed description.) The \textit{recommendation task} is to choose an allocation of nudges $\nudgeallocation$, where each binary $\nudgeallocation_{ij}$ indicates whether student $i$ is nudged to program $j$. We only nudge individuals toward desirable matches ($X_{ij} \leq Z_{ij}$). 
Given a nudge allocation $\nudgeallocation$, students adopt nudges and apply to programs according to a function $\adoptionfxn(\nudgeallocation)$. The resulting applications induce a match outcome $\matchoutcomeonnudges = \matchsimulator(\adoptionfxn(\nudgeallocation))$, with $\matchoutcomeonnudges(i)$ denoting applicant $i$'s match.

\begin{definition}[Congestion constraint] We require nudges to preserve offer likelihoods. For the realized match process, let $\programlossfxn(\matchoutcome) = \min\{\programlossfxn^{\high}(\matchoutcome),\programlossfxn^{\nudged}(\matchoutcome)\}$, where $\programlossfxn^{\high}(\matchoutcome),\programlossfxn^{\nudged}(\matchoutcome)$ are the acceptance rates at program $j$ among all \high likelihood applicants and the subset nudged to $j$, respectively.\footnote{A student may not functionally apply to all programs on their list if they match to a higher-ranked one. Formally, $\programlossfxn^{\high} = $ (\# \high-likelihood students who functionally apply and match to $j$) / (\# \high-likelihood students who functionally apply to $j$), with $\programlossfxn^{\high} := 1$ if there are no such applicants, and analogously for $\programlossfxn^{\nudged}$. Note that other choices of $\ell$ are also compatible with our approach, e.g., considering acceptance rates among all applicants.} Then, we require that $\programlossfxn(\matchoutcomeonnudges) \geq L$ for all $j\in\programs$, where $\programs$ is the set of all programs $j$ that are ``nudge-eligible'' (meaning $Z_{ij}=1$ for at least one student $i$).
\end{definition}

When instantiating our method, we enforce the constraint among \high likelihood applicants only, since estimates for the adopting subset can be noisy when the subset is small; in our simulations (\Cref{sec:simulation_results}), we set $L = 0.8$, and in the deployment, we set $L = 0.95$. 
For a nudge allocation $\nudgeallocation$, let
$
\operatorname{Matches}_Z(\nudgeallocation)
:=
\sum_{i,j} Z_{ij}\mathbf{1}\{\matchoutcomeonnudges(i)=j\}
$
denote the number of desirable matches. 
Then, our above discussion yields the following optimization problem:
\begin{align}
    \nudgeallocation^* \in \argmax_{\nudgeallocation} \quad & \operatorname{Matches}_Z(\nudgeallocation) \label{opt:direct}\\
    \text{s.t.} \quad \programlossfxn(\matchoutcomeonnudges) \geq L \quad \forall j \in \programs; &
    \qquad
    \nudgeallocation_{ij} \in \{0, 1\},\ \ \nudgeallocation_{ij} \leq Z_{ij} \quad \forall i,j \nonumber
\end{align}
\noindent However, we cannot directly optimize \eqref{opt:direct} in deployment: predicting $\mathbf{1}\{\matchoutcomeonnudges(i) = j\}$ and $\programlossfxn(\matchoutcomeonnudges)$ for a potential nudge allocation $\nudgeallocation$ requires detailed counterfactual information on each applicant's behavior and the resulting match outcome. In other words, the relationships between a nudge allocation, the objective, and the congestion constraints are intermediated by unknown applicant behavior $\adoptionfxn(\cdot)$ and the match algorithm. Instead, we develop an approach to calculating a deployment-time allocation that approximately solves \eqref{opt:direct} in practice.

\subsection{Methodological Approach and Theoretical Results}
\label{subsec:allocation}

Our method, summarized in \Cref{fig:fullpipelinediagram}, separates recommendation into two stages. First, using historical data, we calculate \textit{nudge capacities} $\nudgecaps \in \mathbb{N}^{|\programs|}$, the maximum number of applicants we can nudge to each program such that we expect to satisfy the congestion constraints $\programlossfxn \geq L$. At deployment time, we then allocate nudges given these pre-computed capacities, using an integer program whose objective function is concave increasing in the number of nudges sent to each applicant so as to broadly distribute the available nudges. This decomposition allows us to leverage historical data to ensure we satisfy the congestion constraints, and then allocate recommendations at deployment-time without advance knowledge of applicants or their preferences. 

A central challenge our method solves is determining how high nudge capacities can be while satisfying the congestion constraints. We do so using the following bilevel optimization approach. Let CapacitatedNudgeAllocation, $\textbf{\methodabbr}(\nudgecaps)$, denote the nudge allocation returned by our integer program, given nudge capacities $\nudgecaps$. To evaluate candidate nudge capacities $\nudgecaps$, we make an \textit{adoption} assumption about how applicants who \textit{would have} submitted a ranking $\applicationprofiles$ without nudges instead submit a ranking $\applicationprofiles'$ with given a set of nudges. We further assume that, given any application behavior, we can evaluate whether the congestion constraints are satisfied (e.g., with a simulator). Our method for calculating nudge capacities, \algoname, then uses \textbf{\methodabbr} as a subroutine to evaluate candidate nudge capacities and search for the largest nudge capacity vector such that the nudge-induced applications satisfy the congestion constraints.

\pgfdeclarelayer{behind background}
\pgfsetlayers{behind background,background,main}

\begin{figure}[t]
\centering

\definecolor{historicalblue}{RGB}{217,226,243}
\definecolor{optgreen}{RGB}{220,245,210}
\definecolor{insetgray}{RGB}{159,159,159}

\resizebox{\textwidth}{!}{
\begin{tikzpicture}[remember picture,
    node distance=1.6cm and 2.5cm,
    box/.style={draw, rectangle, rounded corners, minimum width=0cm, minimum height=0.9cm, align=center},
    historicalinput/.style={box, fill=historicalblue},
    input/.style={box, fill=orange!12},
    opt/.style={box, fill=optgreen},
    def/.style={box, fill=none},
    sim/.style={box, fill=white},
    output/.style={
        draw=none,
        fill=none,
        align=center,
        inner sep=2pt
    },
    arrow/.style={->, thick},
    font=\large,
    block/.style={
        draw,
        rectangle,
        fill=black!4,
        inner sep=14pt
    },
    line/.style={thick}
]

\node[historicalinput] (historical_23_25) {\textbf{Historical data (2023-25)}};

\node[opt, below=1.8cm of historical_23_25] (cna) {\textbf{NudgeCapComputation}};
\node[inner sep=0pt, outer sep=0pt] (cnaSWmark) at (cna.south west) {};

\draw[arrow] (historical_23_25) -- (cna);

\node[opt, right=3cm of cna] (algo1) {\textbf{\methodabbr} \\\textbf{(Nudge allocation)}};

\node[output, right=2.2cm of algo1] (deploy) {\textbf{Deployment}};
\node[output, right=3.8cm of deploy] (eval) {\textbf{Analyze} \\\textbf{results}};

\draw[arrow] (algo1) -- node[above]{\textit{Nudges} $X$} (deploy);
\draw[arrow] (deploy) -- node[above, align=center]{\textit{Submitted} \\\textit{application profiles}} (eval);

\node[input, above=1.8cm of algo1] (match_simulator) {\textbf{Real-time data (2025-26)}};

\draw[arrow] (match_simulator) -- node[right, align=left]{\textit{Student attributes}\\ \textit{+ offer likelihoods}} (algo1);
\draw[arrow] (historical_23_25) -- (algo1);

\draw[arrow] ([yshift=0pt]cna.east) --
    node[above, align=center] [xshift=0.15cm]{\textit{Nudge caps $\nudgecaps$}}
    ([yshift=0pt]algo1.west);

\begin{scope}[on background layer]
\node[
    draw,
    color=insetgray,
    inner sep=8pt,
    line width=1.5pt,
    fit=(historical_23_25) (cna)
] {};
\end{scope}

\end{tikzpicture}
}

\vspace{1cm}

\resizebox{\textwidth}{!}{
\begin{tikzpicture}[remember picture, 
    node distance=1.6cm and 2.5cm,
    box/.style={draw, rectangle, rounded corners, minimum width=0cm, minimum height=0.9cm, align=center},
    historicalinput/.style={box, fill=historicalblue},
    input/.style={box, fill=orange!12},
    opt/.style={box, fill=optgreen},
    def/.style={box, fill=none},
    sim/.style={box, fill=white},
    output/.style={
        draw=none,
        fill=none,
        align=center,
        inner sep=2pt
    },
    arrow/.style={->, thick},
    font=\large,
    block/.style={
        draw,
        rectangle,
        fill=black!4,
        inner sep=14pt
    },
    line/.style={thick}
]

\node[opt, minimum width=2.4cm] (cna2) {\methodabbr \\ (Nudge \\ allocation)};
\node[sim, right=2cm of cna2, minimum width=2cm] (adoption)
    {Adoption \\ function $\adoptionfxn$};
\node[sim, right=2cm of adoption, minimum width=2cm] (match)
    {Match \\ simulator $\matchsimulator$};
\node[sim, right=1.8cm of match, minimum width=3.2cm] (congestion)
    {Evaluate \\ congestion \\ constraint $\ell_j \ge L$};

\draw[arrow] (cna2) -- node[above, output, align=left] {\textit{Nudges} $X$} (adoption);
\draw[arrow] (adoption) -- node[above, output, align=center] {\textit{Post-nudge} \\ \textit{profiles} $\succ'$} (match);
\draw[arrow] (match) -- node[above, output] {\textit{Match} $\mu$} (congestion);

\node[historicalinput, above=2.1cm of adoption, minimum width=11cm, font=\Large, minimum height=0.95cm]
    (historical) {\textbf{Match inputs (2024--25)}};

\node[output, below=0.3cm of historical, xshift=-2.7cm, align=left]
    (hist_student_label) {\textit{Student attributes} \\ \textit{+ offer likelihoods}};
\node[output, below=0.3cm of historical, xshift=1.65cm, align=left]
    (hist_app_label) {\textit{Student application} \\ \textit{profiles $\succ$}};
\node[output, below=0.3cm of historical, xshift=5.7cm, align=left]
    (hist_prog_label) {\textit{Program capacities} \\ \textit{+ rankings}};

\draw[arrow] ($(historical.south)+(-4.21cm,0)$) -- (cna2.north);
\draw[arrow] ($(historical.south)+(0,0)$) -- (adoption.north);
\draw[arrow] ($(historical.south)+(4.1cm,0)$) -- (match.north);

\node[historicalinput, left=2.4cm of cna2, yshift=0.75cm, minimum width=2.2cm]
    (hsmetrics) {\textbf{HS metrics} \\ \textbf{(2023--24)}};
\node[historicalinput, below=0.25cm of hsmetrics, minimum width=2.2cm]
    (feeder) {\textbf{Feeder} \\ \textbf{patterns} \\ \textbf{(2023--25)}};

\node[output, right=0.15cm of hsmetrics, yshift=0.15cm, align=center]
    (nudgeeligible) {\textit{Nudge-} \\ \textit{eligible HS}};
\node[output, right=0.15cm of feeder, yshift=0.05cm, align=center]
    (targetms) {\textit{Target MS}};

\draw[arrow] ([yshift=-10.8pt]hsmetrics.east) -- ([yshift=10.5pt]cna2.west);
\draw[arrow] ([yshift=12pt]feeder.east) -- ([yshift=-10.4pt]cna2.west);

\node[
    sim,
    below=0.8cm of congestion.south,
    minimum width=2.5cm,
    minimum height=1.2cm,
    align=center,
    aspect=2.2
] (decision_diamond) {Update $\mathcal C$ \\ or terminate};

\coordinate (feedback_down) at ($(congestion.south)+(0,-1.5cm)$);
\coordinate (feedback_left) at ($(cna2.south)+(0,-1.43cm)$);

\draw[line] (congestion.south) -- (decision_diamond.north);
\draw[line] (decision_diamond.west) -- node[above, output]
    {\textit{\quad \quad Candidate nudge capacities $\nudgecaps$}} (feedback_left);
\draw[arrow] (feedback_left) -- (cna2.south);

\coordinate (finalout) at ($(congestion.south)+(0cm, -3.6cm)$);
\draw[arrow] (decision_diamond.south) -- (finalout);
\node[output, below=0cm of finalout] {\textit{Nudge caps $\nudgecaps$}};

\node[output, below=1.8cm of cna2, xshift=1.65cm, align=left, font=\LARGE]
    (algo_label) {\textbf{NudgeCapComputation}};

\begin{pgfonlayer}{background}

\coordinate (capacity_top) at ([yshift=13pt]cna2.north);
\coordinate (capacity_bottom) at ([yshift=-8pt]algo_label.south);

\node[
    block,
    fill=optgreen,
    inner xsep=10pt,
    inner ysep=0pt,
    fit=(cna2)(adoption)(match)(congestion)
        (feedback_left)(capacity_top)(capacity_bottom)
] (capacityblock) {};

\coordinate (oracle_top) at ($(adoption.center |- cna2.north)+(0,8pt)$);
\coordinate (oracle_bottom) at ($(adoption.center |- cna2.south)+(0,-15pt)$);

\node[
    block,
    fill=gray!30,
    inner xsep=5pt,
    inner ysep=0pt,
    fit=(adoption)(match)(congestion)(oracle_top)(oracle_bottom)
] (oracleblock) {};
\end{pgfonlayer}

\node[output, anchor=south, font=\LARGE]
    at ($(oracleblock.center |- feedback_left)+(0,1.05cm)$)
    (oracle_label) {\textbf{Oracle $\mathcal O$}};

\begin{scope}[on background layer]
\node[
    draw,
    inner sep=14pt,
    color=insetgray,
    line width=1.5pt,
    fit=(cna2)(adoption)(match)(congestion)(feedback_down)(feedback_left)(hsmetrics)(capacityblock)(historical)
] (algo1expanded) {};
\end{scope}

\node[inner sep=0pt, outer sep=0pt] (algo1expandedNWmark)
    at ([xshift=4pt,yshift=-4pt]algo1expanded.north west) {};

\end{tikzpicture}
}

\begin{tikzpicture}[remember picture, overlay]
\draw[dashed, line width=0.5pt, dash pattern=on 2.3pt off 2pt]
    ([xshift=-0.34cm, yshift=-0.3cm]cna.south west) -- ([xshift=-0.34cm, yshift=-1.3cm]cna.south west);
\draw[dashed, line width=0.5pt, dash pattern=on 2.3pt off 2pt]
    ([xshift=0.9cm, yshift=-0.3cm]cna.south) -- ([xshift=11.3cm, yshift=-1.32cm]cna.south);
\end{tikzpicture}

\caption{\textbf{Overview of method and input data.} \algoname computes program-level nudge capacities $\nudgecaps$ using historical data (blue). Then at deployment time, we allocate nudges to new students (orange) subject to these nudge capacities. The inset shows more detail on how NudgeCapComputation searches over candidate nudge capacities $\nudgecaps$: for each candidate $\nudgecaps$, it uses \methodabbr as a subroutine to allocate nudges, simulates applicant behavior and the resulting match outcome given the nudges, and evaluates whether the congestion constraints are satisfied.}
\label{fig:fullpipelinediagram}

\end{figure}
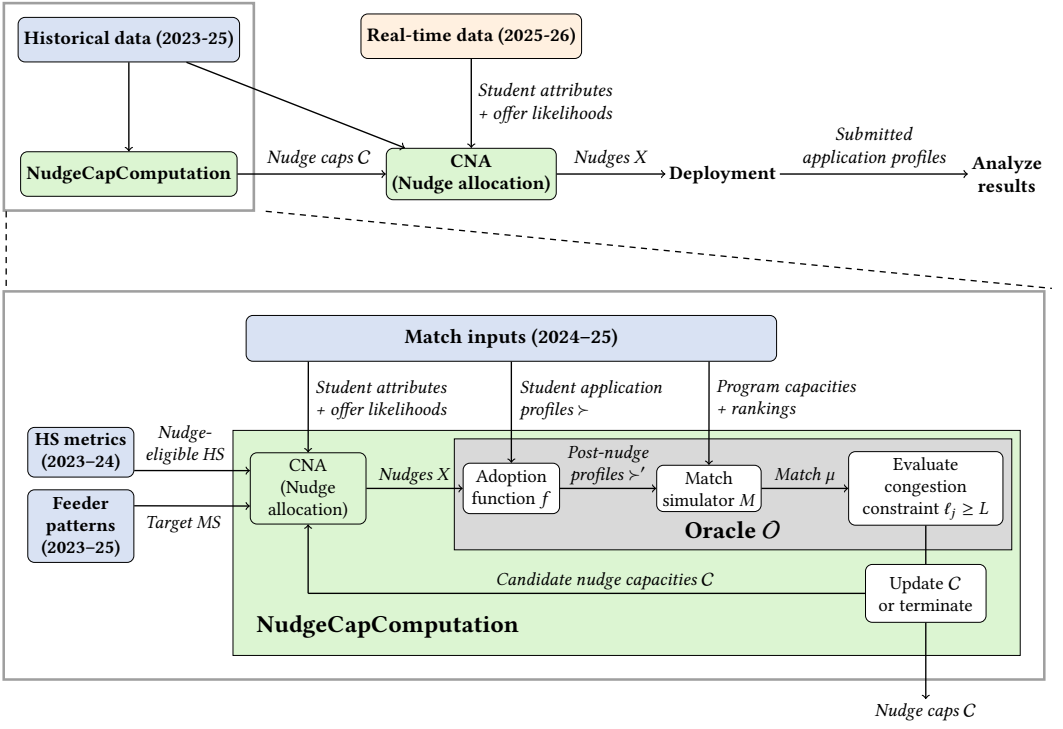

\paragraph{{Theoretical result}} We show that our algorithmic approach yields a set of nudge capacities and allocations that maximize a lower bound of the objective in \eqref{opt:direct} while satisfying the constraints. For theoretical simplicity, we assume the following adoption behavior.

\begin{assumption}[{Adoption function}]
{An applicant receiving $k$ nudges adopts nudges by adding one nudged program to the top of their list with probability $q(k)$, where $q(0)=0$ and $q$ is increasing and discretely concave.} Applicants do not apply to desirable matches without adopting a nudge.
\label{assm:adoption}
\end{assumption}

Our guarantees are stated in terms of expectations over random application and adoption behavior and the resulting match outcomes. (We also note that, in reality, applications often include desirable matches even without nudges.)

\begin{definition}[Constraint Oracle and Safe recommendations] Given a nudge allocation $\nudgeallocation$, let $\mathcal O_j(\nudgeallocation)\in\{0,1\}$ be a deterministic constraint \textit{oracle} that indicates whether the matching process induced by $\nudgeallocation$ satisfies the congestion constraint at program $j$ in expectation: i.e., whether the expected number of \high likelihood and nudged applicants who are admitted is at least a fraction $L$ of the expected number that effectively applies.
  
  Allocation $\nudgeallocation$ is safe if $\mathcal O_j(\nudgeallocation)=1$ for every $j\in\programs$, i.e., the constraint oracle is satisfied. Nudge capacities $\nudgecaps$ are safe if the induced allocation $\textbf{\methodabbr}(\nudgecaps)$ is safe.\footnote{We define safety for a nudge capacity vector $\nudgecaps$ with respect to a specific allocation it induces. This allows us to efficiently instantiate an oracle $\mathcal O$, as described in \Cref{subsubsec:instantiation_details}. We note that our theoretical results can extend beyond this definition (e.g., if the oracle is able to consider arbitrary allocations compatible with a single nudge capacity vector).} 
\end{definition}

The oracle captures how applicants adopt nudges and how their resulting applications are processed by deferred acceptance; in our setting, we instantiate $\mathcal O$ using historical data and a high-fidelity simulator of the match process. 

\begin{assumption}[Feasible and monotone congestion constraint]
The congestion constraints are satisfied without nudges: $\mathcal O_j(\mathbf{0}) = 1$ for all $j$. Furthermore, congestion is monotonic in nudge capacities: increasing a program's nudge capacity cannot decrease congestion at that program and cannot increase congestion at any other program. That is, holding $\nudgecaps_{-j}$ fixed, $\mathcal O_j\!\left(\textbf{\methodabbr}(\nudgecaps_{-j},c)\right)$ is non-increasing in $c$, and $\mathcal O_k\!\left(\textbf{\methodabbr}(\nudgecaps_{-j},c)\right)$ is non-decreasing in $c$ for all $k\neq j$.
\label{assm:deterministic_sim}
\end{assumption}

Then our approach provides an $L$-approximation to the objective function in \eqref{opt:direct}, among safe nudge capacities $\nudgecaps$ and feasible allocations $\nudgeallocation$ given those capacities.\footnote{By considering allocations feasible under some nudge capacity vector $\nudgecaps$ for which $\textbf{\methodabbr}(\nudgecaps)$ is safe, rather than over all allocations that are themselves safe, the benchmark on the right-hand side differs from our original optimization problem \eqref{opt:direct}. These two sets need not coincide, depending on the oracle used to determine safety. This gap is a consequence of our approach to first calculate nudge capacities, and then optimal allocations given those capacities.} 

\begin{theorem}[{{Approach approximately maximizes expected matches}}]
\label{thm:overall_guarantee}
{Suppose \Cref{assm:adoption,assm:deterministic_sim}. Let $\nudgecaps^*$ be returned by \algoname and $\nudgeallocation^*=\textbf{\methodabbr}(\nudgecaps^*)$. Then, $\nudgeallocation^*$ is safe and}
\[{\mathbb{E}\left[\operatorname{Matches}_Z(\nudgeallocation^*)\right]
\geq
L\left[\max_{\substack{
\nudgecaps'\text{ safe}\\
\nudgeallocation\text{ feasible for }Z,\nudgecaps'
}}
\mathbb{E}\left[\operatorname{Matches}_Z(\nudgeallocation)\right]\right].}
\]
\end{theorem}

We now detail each component and prove this result.

\subsubsection{Allocating nudges subject to nudge capacities}
\label{subsubsec:allocation_given_nudgecaps}
First, we describe how to compute an allocation $\nudgeallocation$ given program nudge capacities $\nudgecaps$, students $\students$, and desirable matches $Z$. We use this as a subroutine when calculating the nudge capacities themselves. Let $k_i(\nudgeallocation) = \sum_{j} \nudgeallocation_{ij}$ be the number of nudges allocated to student $i$ under allocation $\nudgeallocation$. For some utility function $u:\mathbb{N}\to\mathbb{R}_{\ge 0}$, we formulate the following integer program, denoted Capacitated Nudge Allocation (\textbf{\methodabbr}):
\begin{align*}
    \textbf{\methodabbr}(\nudgecaps, u)\textbf{} &:
    \argmax_{\nudgeallocation}\ \sum_{i \in \students} u(k_i(\nudgeallocation)) \\
    \text{s.t.} \quad
     \sum_{i\in \students} \nudgeallocation_{ij} \leq \programnudgecap \quad \forall j \in \programs; &
    \qquad
    \nudgeallocation_{ij} \in \{0, 1\},\ \ \nudgeallocation_{ij} \leq Z_{ij} \quad \forall i,j
\end{align*}
The first constraint ensures we allocate at most $\nudgecaps_j$ nudges to program $j$. The objective maximizes total utility of the allocated nudges, for some function $u$ that only depends on the \textit{number} of programs student $i$ is nudged to. 

When $u$ is linear ($u(x) = x$), \textbf{\methodabbr} is equivalent to maximizing the total number of nudges allocated. We choose $u$ to be strictly discretely concave and increasing, to select for a more balanced distribution of nudges.
Intuitively, a balanced distribution of nudges may be both more fair and yield more desirable matches. Concentrating nudges on students who are eligible for many programs (those with greater access) may leave others with no nudges; given the geographic disparities in access throughout the city (see \Cref{fig:access_disparities}), this could perpetuate existing disparities. Further, because each student ultimately matches (and enrolls) to at most one program, allocating many nudges to the same student may have diminishing returns. 
We now formalize these intuitions.

Given \Cref{assm:adoption} and nudge capacities $\nudgecaps$, \textbf{\methodabbr} approximately optimizes our match objective.

\begin{restatable}[CNA maximizes a lower bound on matches]{theorem}{cnamatchesthm}
\label{thm:cna_matches}
Suppose \Cref{assm:adoption}. If allocation $\nudgeallocation^*$ is safe and maximizes \textbf{\methodabbr}$(\nudgecaps,u)$ for any increasing and strictly discretely concave $u$, then
\[
\begin{aligned}
\mathbb{E}\left[\operatorname{Matches}_Z(\nudgeallocation^*)\right]
{\geq L \left[\max_{\nudgeallocation\text{ feasible for }Z,\nudgecaps}
\mathbb{E}\left[\operatorname{Matches}_Z(\nudgeallocation)\right]\right]},
\end{aligned}
\]
i.e., \textbf{\methodabbr} provides an $L$-approximation to our match objective function among feasible allocations.
\end{restatable}
The proof in \Cref{subsec:appendix_concavity_proof} shows that under any strictly discretely concave and increasing $u$, \textbf{\methodabbr} produces an allocation with \textit{prefix-maximal coverage} (it maximizes the number of applicants who receive at least one nudge ($\mathrm{Cov_1}$), and then at least two, and so on.) Then, given adoption \Cref{assm:adoption} and the preservation of offer likelihoods, nudge coverage implies match guarantees: notably, preserving offer likelihoods is key to intervention effectiveness. The first part of the proof proceeds by showing that the feasible allocations induce a matroid, so the greedy algorithm (which produces  {prefix-maximal} coverage) is optimal. For intuition, we also provide a direct construction in which any allocation that does not maximize $\mathrm{Cov_1}$ can be improved by reallocating nudges from those with multiple nudges to those with no nudges, increasing the concave objective and $\mathrm{Cov_1}$.\footnote{Note that the result is trivially true if $Z_{ij}=1$ for all $i,j$, since then we can arbitrarily transfer nudges to those without nudges; but for general $Z$ we need to show transfer feasibility, which we do so via an alternating path construction.}

The result further implies that the exact choice of $u$ is not important: any increasing and strictly discretely concave $u$ yields the same optimal allocation sets. Thus, we write $\textbf{\methodabbr}(\nudgecaps)$ for the allocation returned using such a $u$, breaking ties according to a fixed ordering over allocations.

In summary, we have described a method (\textbf{\methodabbr}) for computing a nudge allocation $\nudgeallocation$ given nudge capacities $\nudgecaps$ and motivated the use of a strictly concave objective. 

\subsubsection{Computing nudge capacities for congestion constraints}
\label{subsubsec:computing_nudgecaps}
\begin{algorithm}[t]
    \SetAlgoNoLine
    \KwIn{{Constraint oracle $\mathcal O$,} nudge capacity upper bounds $\overline{c}$}
    \KwOut{Program nudge capacities $\nudgecaps$}

    Initialize $\nudgecaps^{(0)} \leftarrow \overline{c}$\\

    \For{$t=1,2,\dots$}{
        $\nudgecaps^{(t)} \leftarrow \nudgecaps^{(t-1)}$\\

        \For{each program $j\in\programs$}{
            
            $\nudgecaps_j^{(t)} \leftarrow$ largest $c \in [0, {\nudgecaps_j^{(t-1)}}]$ s.t. $\mathcal O_j\!\left(\textbf{\methodabbr}\left(({\nudgecaps_{-j}^{(t-1)}},c)\right)\right)=1$,\\
            \qquad\qquad which can be done via binary search between $0$ and ${\nudgecaps_j^{(t-1)}}$.
        }
        \lIf{$\nudgecaps^{(t)} = \nudgecaps^{(t-1)}$}{\Return $\nudgecaps^{(t)}$}
    }
    \caption{\algoname (Diagram illustration in \Cref{fig:fullpipelinediagram} inset)}
    \label{alg:nudgecaps}
  \end{algorithm}
We now describe the computation of nudge capacities $\nudgecaps = (\programnudgecap)_{j\in\programs}$, the maximum number of students we can nudge towards each program while maintaining acceptance rates. 
Algorithm \ref{alg:nudgecaps} (\algoname) {initializes nudge capacities at their given upper bounds and iteratively decreases them to satisfy} the congestion constraints $\programlossfxn \geq L$ for all $j$. To determine this feasibility, it uses $\textbf{\methodabbr}$ as a subroutine to allocate nudges subject to the candidate nudge capacities and queries the oracle. For each program, we decrease its nudge capacity to the largest candidate value satisfying its congestion constraint. This process repeats until all constraints are simultaneously satisfied. 

We set the effective upper bounds $\overline c=\left(\sum_i Z_{ij}\right)_{j\in\programs}$, since at those nudge capacities, all applicants are nudged to all programs where they are nudge-eligible; this is the largest number of students we can simulate nudging using the observed data. Let $(\nudgecaps_{-j},\ c)$ denote the vector $\nudgecaps$ with program $j$'s nudge capacity replaced by $c$. Under \Cref{assm:deterministic_sim}, {\algoname} returns the greatest safe nudge capacities {bounded by $\overline c$}.

\begin{restatable}[Greatest safe nudge capacities]{theorem}{algothm} 
Suppose \Cref{assm:deterministic_sim}. Then \algoname{}\hspace{0pt}$(\mathcal O, \overline{c})$ terminates using $O\!\left(|\programs|\left(1+\sum_{i,j}Z_{ij}\right)\log\left(1+\sum_{i,j}Z_{ij}\right)\right)$ oracle calls, and the returned nudge capacities $\nudgecaps^*$ are safe. Moreover, for any safe $\nudgecaps \leq \overline{c}$, it holds that $\nudgecaps \leq \nudgecaps^*$ coordinatewise; thus, $\nudgecaps^*$ is the unique coordinatewise greatest safe vector bounded by $\overline c$.
\label{thm:algorithm_guarantees}
\end{restatable}

The proof is in \Cref{proof:algo_guarantees}. Putting our results together proves \Cref{thm:overall_guarantee}.

\begin{proof}[Proof of \Cref{thm:overall_guarantee}]
For any safe $\nudgecaps$, the bounded capacity vector $\widetilde{\nudgecaps}$ defined by $\widetilde{\nudgecaps}_j=\min\{\nudgecaps_j,\overline c_j\}$ has an identical feasible allocation set ($\textbf{\methodabbr}(\nudgecaps) = \textbf{\methodabbr}(\widetilde{\nudgecaps})$ under fixed tie-breaking), so $\widetilde{\nudgecaps}$ is safe. By \Cref{thm:algorithm_guarantees}, $\nudgecaps^*$ is safe and $\widetilde{\nudgecaps}\leq\nudgecaps^*$. Thus, every allocation feasible for $Z,\nudgecaps$ is also feasible for $Z,\nudgecaps^*$. The result follows from \Cref{thm:cna_matches} applied at $\nudgecaps^*$.
\end{proof}

\subsection{Implementing our method in practice}
\label{subsubsec:instantiation_details}

We now discuss the instantiation of our approach in practice. 
We use \textit{historical data} from the 2024-25 application cycle to compute nudge capacities: for completed application cycles, we observe a submitted \textit{ranked application profile} $\applicationprofiles^i$ for each applicant $i$, indicating strict preferences over some subset of programs in the absence of nudges.

This allows us to estimate counterfactual applications for any given nudge allocation $\nudgeallocation$---we assume an \textit{adoption function} $\adoptionfxn(\applicationprofiles, \nudgeallocation)$ that simulates how applicants would have \textit{updated} their profiles given a nudge allocation $\nudgeallocation$ and outputs a new collection of ``post-nudge profiles.'' We instantiate the adoption function as follows: each applicant adds one randomly selected nudged program to the top of their no-nudge profile with fixed probability $p$ (the \textit{adoption rate}). In general, $\adoptionfxn$ may either encode domain knowledge or be learned from data on prior interventions. 

We can then simulate {match outcomes} $\matchoutcomeonnudges = \matchsimulator(\adoptionfxn(\applicationprofiles, \nudgeallocation))$ for these post-nudge profiles, by running deferred acceptance using known program seat capacities and preferences.\footnote{Since program seat capacities and preferences remain fixed in all simulations, we omit them from the notation, even though the match outcome depends on them.} Using this data, we instantiate the constraint oracle as
\[
{\mathcal O_j(\nudgeallocation)
=
\mathbf 1\!\left\{
\ell^\high_j\!\left(\matchsimulator(\adoptionfxn(\applicationprofiles,\nudgeallocation))\right)\geq L
\right\}},
\]
i.e., given an allocation, we simulate adoption and the match, and then evaluate the constraint. (We just use $\ell^\high_j$ in our instantiation, since $\ell^\nudged_j$ becomes noisy with small nudged subsets at some programs; we verify that the constraint for the final nudge capacities generally holds among them as well.)
Then, in deployment, we generate nudge allocations for observed applicants using these pre-computed capacities $\nudgecaps$ (appropriately scaled and modified for robustness) and solve \textbf{\methodabbr}$(\nudgecaps, u(x)=\sqrt{x})$.\footnote{\textbf{\methodabbr} can be formulated as a lexicographic max-/min-cost flow on a bipartite graph via a student-copy expansion and is therefore solvable in polynomial time; we use Gurobi for implementation convenience and directly return its output.} Recall that the exact choice of $u$ does not matter, by \Cref{thm:concave_beats_linear_cov}. This full procedure is illustrated in \Cref{fig:fullpipelinediagram}. 

Finally, we note that our instantiated constraint oracle is stochastic (the adoption function encodes random behavior, and multiple optimal allocations may exist)---it calculates a single draw of the congestion constraint as opposed to its expectation, meaning \Cref{thm:algorithm_guarantees} does not apply as-is. (One approach is to turn \Cref{thm:algorithm_guarantees} into one with probabilistic guarantees and rerun the simulation multiple times per iteration, but this prohibitively increases runtime by requiring many deferred acceptance simulations for each iteration.) To accommodate this stochasticity without substantially increasing runtime, we modify \algoname in our implementation as follows: in each iteration, we loop through programs $\{j\}$; for each $j$, we simulate the match once to get an empirical estimate of $\programlossfxn$. Then we increase $j$'s nudge capacity if its empirical congestion metric is above $L + \epsilon$, and decrease it if it is below $L$. We continue until there is an iteration with no capacity changes. This modification reduces per-iteration runtime at the cost of ``backtracking;'' we may need to increase a program's nudge capacity, unlike in \Cref{alg:nudgecaps}. We find that this approach converges efficiently in practice; a complete run to compute nudge capacities for 56 nudge-eligible programs on the 2024-25 cycle data takes approximately 24 hours, and in practice this could be sped up further with warm-starts and further optimizations.

\subsection{Model and Methods Discussion}
In our setting, ``desirable matches'' encode our choices of ``high-performing'' high school programs, target middle schools, an applicant-program level offer likelihood constraint, and a transportation distance constraint. Although defining desirable matches is inherently normative, doing so is foundational to the idea of recommender systems, and our intervention goal is aligned with past interventions. Our method itself takes these as given and is agnostic to the specific student-program pairs $Z$ deemed desirable. 
Furthermore, while nudges outside of the desirable pairs ($X_{ij} > Z_{ij}$) could theoretically help to accomplish the intervention goal, we decided collectively with NYCPS stakeholders that we did not want to provide any information that could be deemed lower quality. We therefore only nudge students towards desirable matches. 

We also emphasize that accurately simulating counterfactual application behavior and match outcomes is crucial for our computation of nudge capacities. Following \citet{pengUndermatchingDisparitiesPortfolio} and the deployed offer likelihood work of \citet{corradini2025school}, we only nudge applicants to programs where we can reconstruct the program’s admissions ranking for \textit{counterfactual applicants}.\footnote{The nudges do not introduce any new applicants to any \textit{other} programs, so historical rankings over applying students are sufficient for simulating counterfactual match outcomes.} 
Relatedly, if our assumed adoption function $f$ is substantively incorrect, we would misestimate behavior and congestion. For instance, we assume that students not nudged in our intervention exhibit no behavioral change; in reality, non-nudged students may still change their behavior compared to comparable applicants in previous years, due to increased reliance on the offer likelihood estimates themselves, spillover effects, or other factors. As such, our deployment this year is conservative. In the future, these other behavioral changes could be directly incorporated into the model of $f$. 

Our theoretical results show that our method effectively optimizes for total number of desirable matches. Some natural extensions are that we could define $u_i$ to vary by student (e.g., heterogeneous adoption rates, or importance weights). While solutions may no longer be prefix-maximal, the greedy solution still optimizes this setup, and so related results hold. 

\section{Design decisions and empirical context on the NYC HS match } 
\label{sec:datasetup}
In this section, we overview our data and empirical choices relevant for both our simulations in \Cref{sec:simulation_results} and deployment in \Cref{sec:experiment}.
We provide further detail on the NYC high school match in \Cref{app:application_process_background} and on the data in \Cref{app:data}. We developed the system based on analysis and simulations using historical student data from the 2024-25 cycle (72,197 students), and we deploy our system for the 2025-26 cycle. We consider schools other than the 9 Specialized High Schools, which run a parallel process \cite{nycdoe_specialized_high_schools}.

We define criteria for designating ``high-performing high schools'' in collaboration with the NYCPS Office of School Performance, based on insights from past NYCPS surveys that families care most about school performance and safety metrics \cite{idoux2024overcoming}. We use public metrics reported on the 2023-24 school year, the most recent results available at deployment time. 

\begin{definition}[High-performing high school]A high-performing high school is a high school with 4-Year Graduation Rate $\geq$ 95\% (top 30\% of schools), 4-Year College and Career Readiness Rating $\geq$ 60 (top 22\%), and Safety and School Climate Rating $\geq$ 3.0 (top 74\%).
\end{definition}

High-performing high schools account for approximately 16\% of schools and 22\% of seats. Each high school has at least one \textit{program}; different programs at the same school may provide different concentrations (such as arts), and students apply to and are admitted at the \textit{program level}. Since not all programs have offer likelihood estimates available (e.g., auditions or essays are part of the admissions process), we further define a subset of programs at high-performing high schools that we can actually nudge select students towards, based on the availability of offer likelihoods. 

\begin{definition}[Potential nudge-eligible high school program]A potential nudge-eligible high school program is one that (1) is at a high-performing high school and (2) has offer likelihood estimates. 
\end{definition}

There are 56 potential nudge-eligible high school programs, across 40 schools and 17 (of 32) districts. This is the set of programs we use for simulations in \Cref{sec:simulation_results}. For the deployment, we requested approval from superintendents to include their schools in the intervention, and further reduced programs for a conservative deployment. Ultimately, 15 programs across 12 schools and 9 districts were included in the intervention as \textit{deployed nudge-eligible high school programs}. Notably, this set did not include any programs in the Bronx or Staten Island (so while we discuss effects on Bronx Latin in our simulations, it does not appear in our deployment set). 

We next define criteria for target middle schools: those that do not typically send students to high-performing high schools. We target students from these middle schools in our intervention.

\begin{definition}[Target middle school]A target middle school is a public NYC middle school (excluding charter schools) that sent less than 15\% of its students to high-performing high schools in the 2024-25 cycle, and averaged across 2023-24 and 2024-25 cycles. Note that we consider all high-performing high schools here, not just nudge-eligible ones. 
\end{definition}

This resulted in 240 target middle schools. This is the set of target middle schools used in all simulations. Demographically (all numbers based on the 2024-25 cycle), 83\% of students at these schools are eligible for Free/Reduced Lunch (FRL) (versus 73\% overall). 11.5\% of students at high-performing high schools came from target middle schools, even though target middle schools account for 33.6\% of incoming high school students and over 50\% of students at target middle schools have \high offer likelihood to at least one nudge-eligible program within 45 minutes of their home. For the deployment, due to changes in schools still open and other related rationales, we had a total of 218 target middle schools. Further, nudges were sent in two batches, and we dropped a set of middle schools from Batch 2 in response to stakeholder feedback after Batch 1. This criterion was applied to both treatment (25 dropped) and control (22 dropped). 

Finally, a student is eligible to receive emails encouraging them to apply to a nudge-eligible high school program if they meet the following criteria, determined in collaboration with many NYCPS offices and informed by parent focus groups held in Summer 2025.

\begin{definition}[Nudge-eligible student for a program]Any student $i$ is ``nudge-eligible for a program $j$'' if they attend a target middle school, their middle school sent less than 2\% of its students to program $j$ in the 2024-25 cycle, they have a \high chance of receiving an offer to program $j$, and their home census tract's centroid is within 45 minutes of program $j$'s high school. 
\end{definition}

We use OpenTripPlanner to estimate travel time using any combination of subway, bus, and walking. In data from the 2024-25 cycle, using zip code as an approximation for census tract, this led to about 11,500 total students who are nudge-eligible for at least one program. 88\% of these students are FRL-eligible, compared to 73\% in the general student population and 83\% at target middle schools.\footnote{If we filter for programs within 30 minutes of a student's home instead of 45, the number of students nudge-eligible for at least one program becomes about 6,500. Among these students, $87\%$ are FRL-eligible.} 
Geographic distributions of these schools and students are shown in \Cref{fig:counts_by_zipcode}.

\section{Results: Congestion in simulation}
\label{sec:simulation_results}

Here we demonstrate the implications of congestion and evaluate our approach using simulations on 2024-25 cycle data. This analysis is ``in distribution'': we compute nudge capacities (using our approach in \Cref{sec:method}) and then evaluate match outcomes \textit{as if we nudged applicants in that same cycle}, under various nudge allocation methods and adoption assumptions.\footnote{In contrast, the deployment is ``out of distribution,'' run on 2025-26 application data using capacities computed for 2024-25.} To account for stochasticity in nudge allocations and adoption, we report mean values over repeated simulations. 

\subsection{Empirical setup}
We use the definitions of nudge-eligible high schools and students specified in \Cref{sec:datasetup}. Notably, for all simulations in this section we use the ``full'' set of 56 potential nudge-eligible high school programs, including programs that we ultimately cut in the deployment (e.g., Bronx Latin). We allocate nudges among nudge-eligible students according to each of the methods defined below. At various adoption rates, we simulate modifications to their true submitted application profiles ($f$ in our above notation) and the resulting match outcomes. All non-nudge-eligible students retain their original application profiles. 
For each adoption rate $p$, we compute and enforce nudge capacities to target acceptance rates among \high offer likelihood students ($\programlossfxn$) above 80\% at $p$\% adoption.\footnote{Note that in the deployment, we ultimately targeted 95\% acceptance rates.} We compare three methods, calculating nudge capacities as defined in \Cref{sec:method}:
\begin{itemize}
    \item \textbf{\naive}: Nudge \textit{all} nudge-eligible student-program pairs, without using nudge capacities.
    \item \textbf{\optimizelinear}: Allocate nudges subject to nudge capacities, maximizing total number of nudges (linear objective), i.e., $\textbf{\methodabbr}(\nudgecaps, u(x) = x)$.
    \item \textbf{\optimizesqrt}: Allocate nudges subject to nudge capacities, maximizing a concave utility of nudges, i.e., $\textbf{\methodabbr}(\nudgecaps, u(x) = \sqrt{x})$.\footnote{As shown in \Cref{thm:concave_beats_linear_cov}, any increasing, strictly concave function has the same set of optimal allocations. We rely on Gurobi's returned solution to tiebreak among the optimal set. Note also that solutions maximizing a concave objective are also optimal for the linear objective. }
\end{itemize}

\optimizesqrt allocation is our proposed method, while the others disambiguate the effects of nudge capacities and of the concave objective function. In \Cref{subsec:nonpersonalized_nudges}, we further analyze \textit{non-personalized} recommendations in simulation, in which all students from target middle schools are nudged to all nearby nudge-eligible programs (reflecting the design of prior interventions).

\subsection{Results}
\paragraph{Nudge capacities preserve program-level acceptance rates.} \Cref{fig:acceptance_rates_hist} shows that the nudge capacities indeed decrease program-level congestion, retaining acceptance rates for \high offer likelihood students above the targeted 80\% threshold. This improves upon results from \naive allocation, where acceptance rates among \high offer likelihood students fall below 80\% for 7 of 56 programs at 20\% nudge adoption, and for 15 of 56 programs at 50\% nudge adoption. In \Cref{subsec:nonpersonalized_nudges}, we show that even \naive personalized recommendation substantially outperforms non-personalized recommendation, in which most programs reject more than half of applicants even at 20\% adoption.

\paragraph{Nudge capacities address disparities in match rates to nudged programs.} \Cref{fig:match_rate_geographic} depicts match rates to nudged programs by zip code. Nudge capacities lead to increased match rates among nudged students across the city, addressing geographic disparities that arose with the \naive approach. 

\Cref{fig:match_rate_invsize} shows that nudge capacities also address disparities in match rates by level of student access. With \naive nudge allocation, match rates among students with access to only one nudge-eligible program are much lower (despite all students adopting only \textit{one} nudged program). Enforcing nudge capacities with our \optimizesqrt method leads to increased match rates across students of all access levels, and particularly large gains among students with limited access to nudge-eligible programs (increasing match rates for students with 1 inventory program from 51.6\% to 76.2\%, and those for students with 5+ inventory programs from 91.3\% to 95.8\%). 

\paragraph{\optimizesqrt allocation increases match rates while approximately preserving the total number of students who match to nudged programs.} Notably, while the number of nudges and nudged students decreases when imposing capacities (leading to higher match rates), using our method retains the number of \textit{matches to nudged programs}---``desirable matches''---which our intervention goal is to maximize. In our simulations at 20\% (\Cref{table:20_adoption}), our method actually achieves \textit{more} matches to nudge-eligible programs than the \naive method of allocating nudges to all eligible students (which has substantially lower match rates). Thus, nudge capacities can help address congestion without sacrificing overall effectiveness of the intervention. Intuitively, the constrained nudges steer applicants toward the programs they are more likely to be admitted to in equilibrium. 

\paragraph{Allocation with concave objective leads to a geographic shift in prioritized students.} Tables \ref{table:20_adoption} and \ref{table:table_50adoption} also reveal that \optimizesqrt allocation improves upon \optimizelinear in number of unique students nudged ($\Cov_1$) and matched, empirically supporting the theoretical claims in \Cref{subsubsec:allocation_given_nudgecaps}. To illustrate the key difference between the allocation methods, \Cref{fig:Y35A_nudges_shift} depicts the distribution of students nudged to Bronx Latin in the simulation, the only nudge-eligible school in the Bronx. Due to geographic constraints, many students (1304) have Bronx Latin as their \textit{only} nudge-eligible program. Of these students, 630 reside in zip codes where \textit{no} student has access to \textit{any other} nudge-eligible programs.
Under \optimizelinear allocation, nudges to Bronx Latin are distributed arbitrarily among eligible students. In contrast, the \optimizesqrt allocation prioritizes students without access to other programs, shifting more nudges towards students in the Upper Bronx (especially those in zip codes with no access to other nudge-eligible programs). 

\begin{figure}[tb]
  \begin{center}
    \centerline{\includegraphics[width=0.6\textwidth]{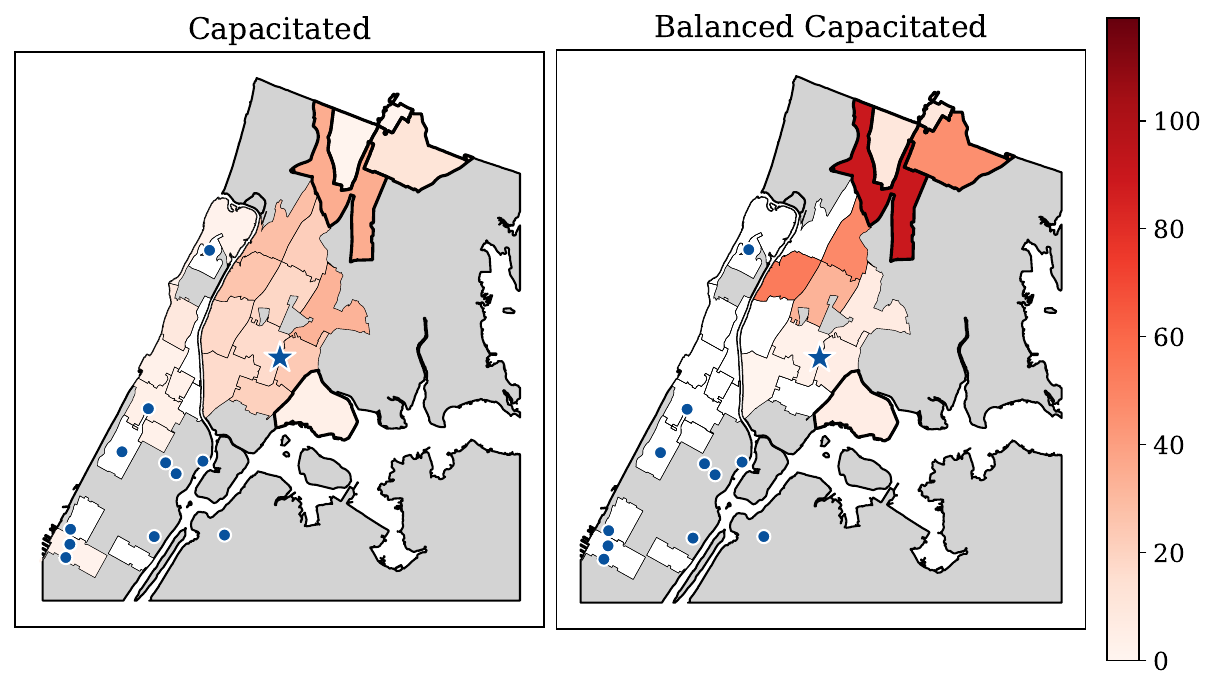}}
    \caption{\textbf{Simulated geographic shift in students nudged to Bronx Latin at 20\% adoption.} Geographic differences between recommendation methods in simulation. The shade of red reflects the number of students from each zip code nudged to Bronx Latin under each allocation method. Nudge-eligible programs are indicated by the blue star (Bronx Latin) and dots (other). Balanced allocation gives nudges to students in zip codes \textit{without} access to any other programs (indicated with thicker outline).}
    \label{fig:Y35A_nudges_shift}
    \end{center}
\end{figure}\renewcommand{\arraystretch}{1.2}

This shift leads to meaningful changes in match outcomes. As depicted in \Cref{fig:Y35A_matches_shift}, \optimizelinear allocation yields 9 matches from Bronx Latin-only eligible zip codes, whereas \optimizesqrt yields 25. Importantly, this does \textit{not} come at the expense of target students elsewhere, as they still receive nudges to other programs: among students from other zip codes with access to Bronx Latin, \optimizelinear allocation yields 517 matches to other nudged programs, whereas \optimizesqrt yields 611. Thus, balanced allocation increases desirable match outcomes for students with limited choice sets, \textit{without} reducing aggregate matches among other students.

~\\\noindent These simulations show that our method effectively achieves the intervention goal of increasing matches to high-performing programs while decreasing disparities in acceptance rates.  

\renewcommand{\arraystretch}{1.2}
\begin{table*}\centering
\small
\begin{tabular}{@{}lrrrrr@{}}\toprule
\textbf{Allocation Method} &
\makecell[r]{\textbf{Total \#}\\\textbf{Nudges}} &
\makecell[r]{\textbf{\# Students}\\\textbf{Nudged}} &
\makecell[r]{\textbf{Match}\\\textbf{Rate}} &
\makecell[r]{\textbf{\# Matched to}\\\textbf{Nudged}} &
\makecell[r]{\textbf{\# Matched to any}\\\textbf{Nudge-eligible}} \\
\midrule
\naive & 26,325 & 11,417 & 0.70 & 1,597 & 1,989 \\
\optimizelinear & 16,175 & 7,838 & 0.98 & 1,520 & 1,947 \\
\optimizesqrt & 16,175 & 9,804 & 0.89 & 1,744 & 2,153 \\
\bottomrule
\end{tabular}
\caption{\textbf{Simulation results at 20\% student adoption.} While the \naive method results in more total nudges and total students nudged, recommendation-induced congestion leads to low \textit{match rates} to nudged programs (among students who effectively apply to a nudged program, the fraction that is accepted). Enforcing nudge capacities results in higher match rates to nudged programs; on top of this, our \optimizesqrt method still results in \textit{more} students matched to nudged programs and nudge-eligible programs overall (\textit{desirable matches}).}
\label{table:20_adoption}
\end{table*}

\section{Deployment and Randomized Controlled Trial}
\label{sec:experiment}

In Fall 2025 (2025-26 admissions cycle), we deployed our method in a personalized intervention in collaboration with NYCPS. We sent emails to a subset of nudge-eligible students (our treatment group), highlighting high-performing programs that might be a good fit for them. We ran a pre-registered randomized controlled trial to evaluate the impact of our intervention.\footnote{Pre-registered at the AEA RCT Registry: https://www.socialscienceregistry.org/trials/17146}
Discussions with NYCPS stakeholders guided our congestion constraint policy decisions: in particular, we aimed for at least 95\% acceptance among \high likelihood applicants, at adoption rates of up to 40\%. These decisions were conservative, to account for uncertainty over student adoption (\Cref{fig:miscalibrated_nudgecaps} illustrates the tradeoff between intervention effectiveness and potential congestion effects from computing nudge capacities under incorrect assumptions about adoption; in this first deployment we conservatively err on the side of preserving \high offer likelihood acceptance rates). 

All students---including control students and those not eligible for the intervention---can access personalized offer likelihood estimates and program metrics through other sources (e.g., MySchools, commute time calculators). Thus, our intervention does not make available any \textit{new} information to treatment applicants, compared to others; rather, it highlights existing information in a way that encourages applicants to consider applying to high-performing programs that may be a good fit. Furthermore, applicants have ultimate agency in their application decisions and can ignore nudges if the programs do not align with their true preferences.

\subsection{Experiment design}
\paragraph{Nudge-eligible High School Programs.}
As mentioned in \Cref{sec:datasetup}, our deployment set of nudge-eligible programs was a subset of the potential nudge-eligible programs, resulting in 15 total programs. It did not include any programs in the Bronx or Staten Island. 

\paragraph{Experimental Population.} 
Nudge emails were sent in two batches during the application cycle. For each batch, we pulled data 1–2 days prior to email delivery and determined the experimental population for that batch as being nudge-eligible students (as defined in \Cref{sec:datasetup}) whose applications were \textit{in-progress} (they had opened their application, meaning we had offer likelihood estimates for them, but they had not yet submitted).  
Offer likelihood predictions are dependent on a student's lottery number, which is only generated when they begin their application in the online portal. Thus, this design allows us to balance reaching more students and sending them communications earlier in their application process (hopefully before they have finalized their decisions).

\paragraph{Randomization and Treatment Assignment.} 
We divided our experimental population into two groups: control students received all standard NYCPS communications, and treatment students received the standard communications plus \textit{nudge emails} with program recommendations. We randomized at the middle school level, assigning each target middle school to 0, 50, or 100\% treatment intensity prior to the first email batch. This design allows for the possibility of measuring spillovers across students who attend the same middle school. The number of schools assigned to each treatment intensity (keeping 50\% and 100\% treatment intensity groups at equal size) was chosen such that---based on last year's number of nudge-eligible students per middle school---the expected total number of treatment applicants would be approximately 500.  
We assigned middle schools to treatment intensities via stratified sampling, with strata defined by borough and expected number of nudge-eligible applicants. 

Individual student treatment assignment was implemented separately for each batch. For each batch, we assigned students to treatment or control within each school according to the school’s pre-assigned treatment intensity (0\%, 50\%, or 100\%), by randomly sampling the appropriate fraction of students from the experimental population at that middle school. Balance tests in \Cref{sec:balance_checks} support that the randomization was effective; we further prespecified and include stratum fixed effects in our primary analyses, to account for realized differences in the group samples.

\paragraph{Nudge Allocation for Treated Applicants.}
Among treatment students in each batch, we determined which nudge-eligible high school programs to recommend to which students using the integer program described in \Cref{subsubsec:allocation_given_nudgecaps}. Each program's nudge capacity was pre-computed as described in \Cref{subsubsec:computing_nudgecaps} using data from the 2024-25 cycle, targeting an acceptance rate of at least 95\% among high likelihood applicants at up to 40\% student adoption. Because nudges were sent in two batches, we solved the optimization problem separately for each batch. Nudge capacities were scaled according to the batch size: if Batch 1 contained $x$\% of the expected total of 500 treatment students, we would scale the Batch 1 nudge capacity for each program by $x$\% and Batch 2 nudge capacities by (100-$x$)\%. Ultimately, with 353 treatment students in Batch 1 and an overall target of 500, each program's nudge capacity in Batch 1 was approximately 70.6\% of its overall nudge capacity. 

\paragraph{Defining Counterfactual Nudges for Control Applicants.}
For the control group, the set of programs they \textit{would have been nudged to} is not clearly defined. There is no clear counterfactual for what would have happened if each student had instead been in the treatment group, since nudge allocations depend on the composition of the treatment group itself. For example, we cannot simply add each control group student (individually) to the treatment group and re-solve the optimization: due to the clustered sampling at the middle school level, adding each control student individually would not reflect the true counterfactual scenario where their middle school had been assigned to treatment; since congestion patterns are geographically correlated, this would under-estimate their expected congestion, and control group students would counterfactually receive larger nudge lists.

To address this challenge, we take advantage of the fact that, by design, the control and treatment groups are distributionally identical in expectation. For each batch, we therefore solve a separate, parallel optimization problem for the control group, as if the control and treatment groups had been swapped entirely. This produces, for each control applicant, a counterfactual set of programs to which they would have been nudged had they instead been assigned to treatment.

Because the control group was larger than the treatment group, we proportionally scaled program capacities in the control optimization. For each batch, if the control group had $x$\% as many students as the treatment group, we scaled each program’s nudge capacity by $x$\%. Using the resulting counterfactual nudge sets, we can directly compare behavior and match outcomes between groups.

\paragraph{Nudge Emails.} The 2025-26 high school application opened on October 7 and closed on December 4, 2025. We sent nudge emails from the NYCPS High School Enrollment email account\footnote{This NYCPS email account also sent other application-related emails to all students during the application period, such as reminders to submit applications before the deadline.} on October 30 (Batch 1) and November 19 (Batch 2). The email template is in \Cref{fig:email_text}. Our emails explained to applicants that they had a \high chance of offer at high-performing high school programs within 45 minutes of their home and listed these programs with tracked links to the school webpages (allowing us to observe clicks for each school, but not whether an applicant opened the email). On November 20, we sent targeted follow-up emails to students from Batch 1, personalized based on whether their saved application profile (if any) included no high-performing programs, high-performing programs but ranked below a program \textit{not} on the high-performing programs list where they had \high offer likelihood, or at least one high-performing program at the top of their list.

\paragraph{Operational details.} Deploying the intervention in collaboration with NYCPS involved many stakeholders and coordinated entities. The intervention and all family-facing text were approved by both University and NYCPS IRBs. Consistent with our data use and non-disclosure agreements, no individual-level identifiers were shared with University partners outside of secure NYCPS systems. 

Implementation choices were shaped by both technical and operational constraints, and key design decisions were made in collaboration with multiple NYCPS offices. Additionally, we conducted six virtual focus groups with 8th- and 9th-grade families in English and Spanish, helping to inform decisions such as what constituted a reasonable distance from a student's home. The email text was drafted and edited in coordination with multiple NYCPS offices and aligned with the High School Admissions Team’s existing communications language and schedule. Consistent with NYCPS norms, emails were sent to the \textit{parents or guardians} of applicants. Personalized emails were translated into 12 languages by the NYCPS translations team and then sent from the High School Enrollment email address, potentially in multiple languages as appropriate for the recipient. While we considered integrating nudges into the application portal, sending text messages, and mailing physical letters, these options were not feasible this year and we delivered nudges solely via email.

All major project decisions were discussed with and approved by NYCPS leadership, including the Chief Enrollment Officer, Deputy Chancellor of School Leadership, Acting First Deputy Chancellor, and Chancellor. We obtained approval from superintendents of all deployment high schools and informed staff at treated middle schools about the intervention and which students received nudges to which programs. Based on stakeholder feedback after the first batch, we removed students from several middle schools for the first batch's reminder emails and the second batch. 

After the application deadline but before match outcomes were released to applicants, we also administered an optional post-application survey to gather descriptive evidence on families’ perceptions of the emails and demand for recommendations. We sent a post-application survey to approximately 2,200 applicants with a \$20 incentive. Survey responses were not used as primary experimental outcomes.

\subsection{Experiment results}
\label{sec:experiment_results}
We ended up with 822 control and 444 treatment students. Here, we report results on treatment effects (application behavior and match outcomes), congestion, and general family engagement and perceptions, with supporting results in the appendix. 

\paragraph{Treatment effects on application rates and match outcomes.} 
\Cref{fig:student_behavior} contains the primary experimental results. Among treatment applicants, 16.4\% ranked at least one program they were nudged to at least somewhere on their list; among control applicants, 10.5\% ranked a program they would have been nudged to. In terms of match outcomes, 5.6\% of treatment applicants matched to a nudged program, and 3.3\% of control applicants matched to a program they would have been nudged to. \textit{No} treatment students were rejected from a program they were nudged to (out of 25 effective applicants, i.e., not counting applicants who listed a nudged program but were accepted to another program they ranked higher), and 1 control student was rejected from a program they would have been nudged to (out of 28 effective applicants), for an overall 98\% acceptance rate. 

We fit pre-registered logistic regression models with standard errors clustered at the middle school level.
Table \ref{table:treatment_effects_logit} reports estimates of the treatment effect on application behavior, both with and without stratum fixed effects. 
In the pre-registered model with stratum fixed effects, treatment significantly increased the likelihood of ranking a recommended program ($\hat\beta$=0.561, SE=0.219, $p$=0.011). For interpretability, we compute average marginal effects to translate this estimate to the probability scale, finding that treatment increased the probability of ranking a recommended program by 6.5 percentage points (SE=2.6pp, $p$=0.013), relative to a control mean of 10.5\%. There is also suggestive evidence that the treatment increased matches to recommended programs ($\hat\beta$=0.643, SE=0.356, $p$=0.071). On the probability scale, this corresponds to a 2.7 percentage-point increase (SE=1.6pp, $p$=0.103), relative to a control mean of 3.3\%. In terms of \textit{where} in a student's application profile these nudged programs were listed, control and treatment are distributionally similar: of students who ranked a nudged program (in control and treatment groups, respectively), 22.1\% and 23.3\% listed it first, and 47.7\% and 47.9\% in the top 3.

Next, we consider applications to all nudge-eligible programs, not just those the applicant was actually nudge-eligible for or nudged to. These outcomes are less dependent on our empirical choices, such as the calculation of hypothetical nudges for the control group. Generally, the additive effect sizes remain unchanged: 36.5\% of treatment applicants ranked at least one nudge-eligible program somewhere on their list, versus 29.6\% of control applicants. 8.3\% of treatment applicants matched to a nudge-eligible program, versus 6.3\% of control applicants.

We run robustness checks with linear probability models, alternative fixed effects specifications, and additional outcome specifications; these do not meaningfully change the results (\Cref{subsec:all_fixed_effects}, \ref{app:additional_experimental_results}). We also report results from additional analyses, including for spillover effects, local average treatment effects, and heterogeneous treatment effects (\Cref{app:additional_experimental_results}). 

\begin{figure}[tb]
  \centering

  \begin{subfigure}[t]{0.22\linewidth}
    \vspace{0pt}
    \centering
    \includegraphics[width=\linewidth]{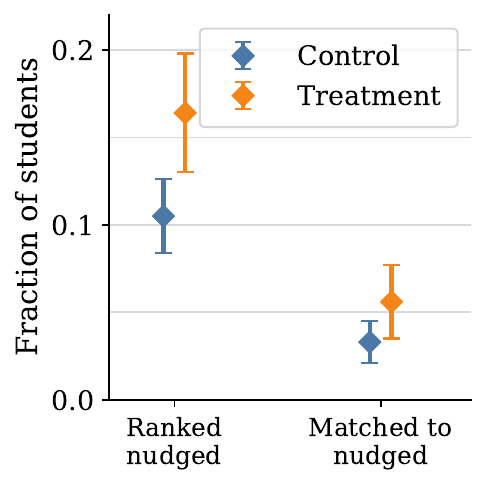}
    \caption{\textbf{Outcome rates.} }
    \label{fig:student_behavior_forest}
  \end{subfigure}
  \hfill
  \begin{subfigure}[t]{0.75\linewidth}
  \centering
  \begin{minipage}[t][2.95cm][t]{\linewidth}
    \vspace{5pt}
    \centering
    \small
    \resizebox{\linewidth}{!}{
      \begin{tabular}{lccccl}
       \toprule
        \textbf{Outcome} 
        & \textbf{Control rate} 
        & \textbf{Treatment rate} 
        & \textbf{Stratum FE} 
        & \textbf{Treatment coef.} 
        & \textbf{p-value} \\
        \midrule

        \multirow{2}{*}{Ranked nudged}
        & \multirow{2}{*}{0.105}
        & \multirow{2}{*}{0.164}
        & Yes & 0.561 (0.219) & 0.011** \\
        & & & No & 0.521 (0.232) & 0.024** \\

        \addlinespace[0.25em]

        \multirow{2}{*}{Matched to nudged}
        & \multirow{2}{*}{0.033}
        & \multirow{2}{*}{0.056}
        & Yes & 0.643 (0.356) & 0.071* \\
        & & & No & 0.564 (0.346) & 0.104 \\

        \bottomrule
      \end{tabular}
    }
  \end{minipage}
  \caption{\textbf{Regression-estimated treatment effects (Logit): primary outcomes.} 1,266 total observations. SEs in parentheses. * $p<0.10$, ** $p<0.05$.}
  \label{table:treatment_effects_logit}
\end{subfigure}

\caption{\textbf{Application behavior by treatment condition.} Regressions are fit with robust standard errors clustered at the school level, estimated with and without stratum fixed effects (borough and school size). CIs in the figure are 95\% CIs.}
\label{fig:student_behavior}
\end{figure}

\paragraph{Link clicks.} The above analyses follow an intent-to-treat design: we include \textit{all} applicants assigned to treatment or control. For insight on engagement, 88 treatment students (19.8\%) clicked on the trackable links, and among those, 25 (28.4\%) ultimately ranked a nudged program (compared to 16.4\% of all treatment applicants). We do not observe email open events.

\paragraph{Post-application survey responses.} We received 239 responses (42 treatment, 61 control, 136 from non-target middle schools). 25/42 treatment families reported \textit{reading} a nudge email; among those, 22/25 (88\%) reported that the email was helpful and 11/25 (44\%) reported that it affected their decisions. 173/195 (89\%) of non-treatment families reported that they would have liked to receive program recommendations (2 families did not answer the question and are omitted). Although these responses should be interpreted descriptively due to potential selection bias, they suggest substantial perceived value of personalized recommendations.

The survey also shed light on why some families did \textit{not} add nudged programs to their application; the most common reported reason was that the recommended programs' academic interest area was not aligned with the student's interests, and many families stated that extracurricular offerings would be an important factor for them in recommended programs. This suggests that while offer likelihood and commute distance are important dimensions for identifying potentially relevant programs, future recommendation systems may benefit from incorporating further personalization regarding students' academic and extracurricular preferences.

\subsection{Scale-up and robustness: What are the limits of student-side interventions?}
\label{sec:scale_up}
Our experiment design was intentionally conservative to account for uncertainty in the first deployment: we reduced the number of nudge-eligible programs, determined nudge capacities assuming an adoption rate of 40\%, and limited the expected treatment group size to 500. These precautions reflect NYCPS' real concerns about recommendation-induced congestion and its effect on offer likelihoods; they provided safeguards against unexpectedly high adoption, year-to-year shifts in application behavior, or other deviations from our modeling assumptions. 

Having deployed the intervention and observed its effects in practice, we may wish to relax these constraints and ``scale up'' such an intervention to increase its reach in the future. Along these lines, we explore two questions: (1) What is the maximum impact the modeled intervention could have, without increasing the supply of high-performing programs? (2) How robust are our congestion calculations to year-to-year distribution shifts in the market environment and applicant behavior? 

\paragraph{Scale-up}
Here, we ask: how impactful can demand-side interventions be? This question informs the necessity of supply-side interventions or other structural changes. We answer this question by calculating nudge capacities using our framework assuming 100\% adoption---since the capacities preserve acceptance rates, this calculation approximates how many applicants can simultaneously, safely, apply and match to high-performing options. Under our model, any larger effect, from demand-side behavior changes alone, must displace existing students with \high offer likelihoods.

Using the 2024-25 applicant population, we compute nudge capacities to maintain the targeted 95\% acceptance rate among \high offer likelihood students at 100\% adoption. We use the set of 15 \textit{deployed} high school programs. We run simulations akin to those in \Cref{sec:simulation_results} to assess the maximum number of students we could have nudged and that would have matched to nudge-eligible programs. 
With the set of 15 deployed programs, this procedure yields 873 desirable matches for the 4,982 nudge-eligible students (in reality, 42 such matches occurred that year). 

There are three further paths to scaling up the size of future interventions. First, our definition of high-performing schools was stricter than in previous interventions; relaxing this, potentially with multiple categories of schools, could allow including more schools. Relatedly, as the survey responses suggest, applicants also have (likely heterogeneous) preferences for other school characteristics, such as extracurriculars offered; incorporating such personal preferences in recommendations could further reduce congestion. Finally, we have so far considered the offer likelihood design as fixed; it may be possible to jointly design the intervention and likelihood calculations, such as by incorporating intervention effects into the estimators. These changes would expand the set of feasible intervention designs. 

\paragraph{Distribution shift} Each year, the distribution of applications across programs may change. For instance, offer likelihood estimates themselves (especially as they are more widely advertised and trusted), other policy decisions, demographic shifts, and program changes could change applicant behavior over time. Given our use of historical data in computing nudge capacities, it is important to understand robustness to these year-to-year distribution shifts. 

We re-compute nudge capacities using the (deployment year) 2025-26 applicant population, targeting the same 95\% acceptance rate among \high offer likelihood students at 40\% adoption,\footnote{While some 2025-26 applicants adopted nudges in response to our intervention, our treatment group was intentionally small. We simulate 40\% adoption on \textit{top} of the submitted profiles, making our computed capacities more conservative.} and compare these values with nudge capacities computed using the same parameters on the 2024-25 applicant population.\footnote{At some programs it was impossible to preserve \high offer likelihood student acceptance rates above 95\% because the rate dropped below 95\% in the true outcomes this year; we set these nudge capacities to 0.} We use our 2024-25 match simulator (using that cycle's seat group targets, rule sets, etc.\footnote{For programs with manual rankings, such as after auditions, we use their true rankings over the 2025-26 applicants.}) for both calculations, to focus on changes from applicant behavior, as well as to reflect the reality that an accurate match simulator for an admission year may not be available until after nudges must be computed.

This comparison reveals two forms of year-to-year shift. First, the market environment itself changed: while the 2024-25 simulator reproduces the 2024-25 match with 98.2\% accuracy, it achieves only 88.6\% accuracy for the 2025-26 match. This suggests that changes beyond student applications (e.g., program characteristics, policies, or data availability) can reduce the accuracy of offer likelihood and nudge capacity calculations across years.\footnote{While it is unclear, even in hindsight, what exact changes lead to this drop, two plausible factors are: (1) year-to-year changes in program capacities, and (2) our 2025-26 data does not perfectly identify FRL status applicants, instead combining that data with other status indicators, such as whether a student is an English Language Learner.} Thus, while using a previous year's simulator may be necessary given the data available at deployment time, these estimates necessarily inherit errors from shifts in the market environment. 
Second, we compare computed nudge capacities between years (\Cref{fig:nudgecap_comparison_40})---since we hold the match simulator fixed, these differences in nudge capacities are only due to shifts in the observed application data. Overall, nudge capacities are strongly correlated across years ($R$=0.945), although we do observe notable variation at individual programs. The mean absolute difference in nudge capacities between years was 48 (the average program nudge capacity in 2025 was 144), and the mean difference from 2024-25 to 2025-26 was -16. 

This aggregate decline is largely driven by 4 of the 15 programs, whose computed nudge capacities dropped to 0 when computed on 2025-26 data; collectively, the drop in nudge capacities at these programs was 246 (compared to a total net decline of 248). At these programs, there seems to have been a notable increase in applications and resulting congestion; our target acceptance rates could no longer be maintained in simulations, even without additional simulated nudges. 
Among the other programs, nudge capacities fluctuate in both directions but exhibit little net change.\footnote{Some of this program-level variation reflects limitations of the data available for our simulations, as opposed to purely changes in underlying congestion. In particular, we only received offer likelihood data for students captured in our two intervention batches, and therefore do not observe the full set of nudge-eligible pairs for the 2025-26 application cycle. {Based on past years' data, we expect that the number of nudge-eligible students at the time of the Batch 2 data pull is about \textit{half} the total number.} For 5 programs, this appears to have been the limiting factor on the maximum nudge capacity compared to the previous year.}

Together, these results imply that computing nudge capacities on historical data is informative even under the observed distribution shift, but that they do not induce exact guarantees. It is thus important to deploy conservative interventions, as we did---an open methodological question remains on how to best be robust to such shifts at scale. Overall, our results suggest promise in being able to effectively scale up the intervention while preserving offer likelihood constraints. 

\section{Conclusion}
In this paper, we study how to design personalized recommendations in capacity-constrained matching markets. We characterize recommendation-induced congestion by showing that simultaneous adoption can sharply and heterogeneously reduce acceptance rates, failing to preserve a key signal for the recommendations themselves. Our results suggest that recommendation systems in capacity-constrained environments (school matching, labor markets, or dating platforms) should be designed with congestion in mind---individual-level interventions can have substantial system-wide externalities, inducing both inefficiencies and disparities. Then, we propose a generalizable method for recommendation in equilibrium, preserving offer likelihoods while effectively achieving the intervention goal of increasing the number of matches to nudged programs. Finally, we implement our approach in a large, real-world market---the NYC High School Match---and evaluate it through a randomized controlled trial. Treated students were 57\% more likely to apply to a recommended program ($p$=0.011) and 71\% more likely to match to a recommended program ($p$=0.071), demonstrating that recommendations shift application behavior and providing suggestive evidence of changes in matches. To our knowledge, this is the first field deployment of a recommendation system to control congestion in a centralized school matching platform.

We foresee that similar analyses and interventions would be informative in other capacity-constrained matching contexts---including labor, dating, and healthcare markets. Our view casts recommender systems as interventions that affect both individual  recipients and the market as a whole. Under this view, safe, effective deployment of recommenders requires both measuring systemic effects and proactively designing for them. Our work highlights several potential design principles. While recommendations should incorporate match probabilities, they should do so in equilibrium, incorporating recommender effects themselves. This is especially important because, as we highlight, increased congestion can exacerbate existing disparities, disproportionately affecting applicants with the fewest existing opportunities. More broadly, recommendation design is policy design: the choices of ``desirable matches'' and the ``congestion constraint'' reflect decisions about how policymakers may want to change the status quo and preserve it, respectively. 

Finally, our intervention is ultimately constrained by the limits of demand-side interventions in capacity-constrained settings. It cannot be scaled up indefinitely while preserving offer likelihoods, and the supply of seats at high-performing schools limits the potential increases in matches. We quantify such limits---larger changes may require supply-side interventions, such as increasing the supply of seats in high-performing schools. These considerations have analogues in other settings, in which improving recommendations complements, rather than substitutes, structural changes. 

\subsection*{Discussion and Limitations}
Our experimental design and deployment make several policy choices, in consultation with our deployment partners. First, we prioritize preserving offer likelihoods for \high likelihood students, allowing for displacement of lower-likelihood applicants (quantified in \Cref{sec:displacement_effects}). Intuitively, this increases assortativity with respect to each program's admissions priorities, while retaining reliability of \high offer likelihood signals. If desired, the congestion constraint can be adapted to alternative policy goals. 

Additionally, showing offer likelihood predictions to students may change student behavior even without further nudge interventions, and this can shift the congestion we are trying to anticipate when designing our intervention (recall that all students see offer likelihood predictions \cite{corradini2025school}). This is one of the reasons we were conservative in our intervention. In future deployments,  explicitly modeling behavior changes in response to both offer likelihoods and recommendations could support more robust design. More broadly, \textit{jointly} designing offer likelihood estimates alongside recommendations may further improve outcomes in equilibrium. 

Our optimization framework also makes simplifying assumptions. These assumptions were made to retain model transparency and simplicity, which were especially important for our first deployment. First, we assume homogeneous nudge adoption; future versions of the system could capture heterogeneity along these lines and weight recommendations accordingly. Successfully implementing richer personalization has the potential to further improve match outcomes. Furthermore, our method currently produces batch recommendations and relies on accurately simulating match outcomes; extending it to settings with real-time recommendations and stochastic matches is an important direction for future work. Integrating recommendations directly into the application portal may enable dynamic recommendations and improved personalization based on student preferences. However, these extensions would make implementation and analysis more complex, and would require careful validation to ensure the model does not end up systematically de-prioritizing more disadvantaged students or amplifying prediction error. 

Researchers and NYCPS do not observe true applicant preferences. While our survey responses and past surveys \cite{idoux2024overcoming} indicate that safety, academic performance, and distance are most important to families, other factors not included in our study are also important. For example, our survey results also highlight the importance of program interest area (\Cref{sec:survey_outcomes}). Prior work has shown substantial preference heterogeneity, including by demographic group \cite{Corradini2023HighSchoolChoices,idoux2024overcoming,abdulkadiroglu_parents_2020,harris2015schools,burdick2024anywhere,hailey2022racial}. Relatedly, our school performance measures---though decided in conversation with NYCPS stakeholders---do not capture other measures of school quality, such as causal value-add \cite{angrist2024race}. However, applicants have ultimate decision-making authority; when our nudges do not match applicants' personal preferences, they may not be adopted, potentially mitigating negative effects. On the other hand, it is possible that some nudges were adopted without a full understanding of the program's offerings compared to other options. Thus, matching to a nudged program does not necessarily imply an individual welfare gain. 

\begin{acks}
This research is supported by the Alfred P. Sloan Foundation, a Cornell Center for Social Sciences seed grant and the center's Regulated Research Environment, and a William T. Grant Foundation Scholars Award to NG. EC is supported by NSF DGE \#2139899, a Cornell Tech Digital Life Initiative doctoral fellowship, and a Siegel PiTech PhD impact fellowship. NG is supported by NSF CAREER \#2339427, and Cornell Tech Urban Tech Hub, NASA, Google, Mastercard, and Amazon research awards. 
We are grateful to 
Samuel Daunt,
Katrina Fortune,
Allie Jackson Hing,
Amy Lee,
Brett Saffer,
Perry Stein,
Christian Barnes,
Matthew Broggini,
Nadiya Chadha,
Yolanda Fana,
Tamara Lindor,
Ahmad Shilleh,
Diying Wu,
Trevonda Kelly,
Eric Ashton,
Michael Cho,
Aurora Moraes,
Benjamin Cosman, and
Justin Matthews
at NYC Public Schools for their guidance, expertise, and support throughout this project.
We also thank
Gabriel Agostini,
Peter Bergman,
Samuel Breckenridge,
Christopher Campos,
Sarah Cohodes,
Viola Corradini,
Sophie Greenwood,
Clemence Idoux,
Russell Legate-Yang,
Irene Lo,
Anushka Murthy,
Parag Pathak,
Kevin Ren,
Ignacio Rios,
Emily Ryu,
and other members of the Garg Lab and Cornell AIPP 
for valuable conversations and feedback.
\end{acks}

\FloatBarrier
\bibliographystyle{ACM-Reference-Format}
\bibliography{library}

\newpage
\appendix
\section{Additional related work}
\label{sec:relatedworkadditional}

\paragraph{Algorithmic tools  to support the delivery of social and government services while accounting for disparities}
Recent work applies, and examines the performance of, predictive modeling and optimization methods in a range of domains, especially in New York City. These include urban planning and incident prediction \cite{mydlarz_floodnet_2024, zied_abozied_spatial_2025, balachandar2026urban,liu2024quantifying,agostini2024bayesian,franchi2025bayesian,liu2024redesigning,gosciak2026scrutinizing}; allocation of housing and social services \cite{azizi2018designing, chouldechova2018case, gosciak2026llms,koenecke2023popular,daw2026due}, library resources \cite{liu2025optimizing,Liu_Rankin_Garg_2024}, and food \cite{prendergast2017food, raman2026rescuelens, shi2021recommender}; healthcare \cite{bastani2021efficient, chang2021supporting, adams2022prospective, joshi2026reducing, bartle2025shopping, chiang2025learning,bartle2025faster}; education \cite{delarue2026algorithmic,bertsimas2019optimizing,tomkins2023showing,lee2023evaluating,lee2024algorithms}; and refugee resettlement \cite{ahani2024dynamic,bansak2018improving}. One consistent concern in such work---including in our setting---is how administrative burdens and resulting heterogeneous participation may have equity implications in sociotechnical systems \cite{garg2025heterogeneous}.

A line of work builds on these developments to explore the actual process of translating research into real-world impact. These papers examine challenges in deployments \cite{saxena2025ai, paleyes2022challenges, selten2024organizing, kedron_describing_2025}, propose principles and strategies for more effective and responsible implementation in practice \cite{bondi2021envisioning, perrault2020artificial, kshirsagar2021becoming, floridi2021design, chiangForestFromTheTrees2026}, and argue broadly for more research on implementation-specific issues or different academic evaluation norms \cite{mergel2024implementing, zuiderwijk2021implications, wilder2025fostering, arribas-bel_urban_2021}. Our work provides one case study on the deployment and evaluation of such an algorithmic system.

\paragraph{Experimentation under interference.} A growing literature also studies experimentation under interference in markets with scarce resources and develops methods for \textit{computing treatment effects} in these settings
\cite{holtz2025reducing,munroCausalInferenceInterference2025, karakolios2025estimation, arkhangelskyEvaluatingLocalPolicies2025, johariPriceExperimentationInterference2025, fariasCorrectingInterferenceExperiments2023, liExperimentingStochasticCongestion2024, boutilierOperationalDosageImplications2025, johariWhenDoesInterference2025, karnani2024essays, johari2022experimental, blake2014marketplace, eckles2016design}. For example, if the intervention increases demand for a program among the treatment group, then acceptance rates at that program may decrease for the control group, inflating a naively estimated treatment effect on match outcomes. Our work is complementary to this literature, which focuses on randomization schemes and debiasing estimators; we focus on the intervention design itself to \textit{minimize interference} by preserving offer likelihoods for applicants with ex ante \high likelihoods.

Our ability to simulate match outcomes given application behavior allows us to characterize interference effects and optimize our intervention accordingly. Using this capability, we further run scale-up analyses to explore how large our intervention could be while still bounding interference, analogous to analyses in \citet{alubaydli2017can, allende2019approximating, larroucau2024college}. 

\section{Theoretical proofs}
\label{sec:theoretical_proofs}
\FloatBarrier
\subsection{Coverage and proof of \Cref{thm:cna_matches}}

\label{subsec:appendix_concavity_proof}
\noindent We first define \textit{Coverage}, which captures how nudges are distributed.

\begin{definition}[Coverage]
\textit{Coverage at $k$} under $\nudgeallocation$ is the number of students who receive at least $k$ nudges. 
\[
\mathrm{Cov_k}(\nudgeallocation)=\sum_{i\in\students} \mathbf{1}\{k_i(\nudgeallocation)\ge k\}.
\]
\end{definition}
For example, $\mathrm{Cov_1}(\nudgeallocation)$ is the number of students who receive at least one nudge. We now define lexicographic maximum coverage, a property of an allocation that characterizes first maximizing the number of students who receive at least one nudge, then the number receiving two, and so on. 

\begin{definition}[Lexicographic maximum coverage]
An allocation $\nudgeallocation^*$ has \emph{lexicographic maximum coverage} among all feasible allocations given $Z, \nudgecaps$ if there does not exist any feasible $\nudgeallocation \text{and any }k$ such that $\mathrm{Cov_k}(\nudgeallocation) > \mathrm{Cov_k}(\nudgeallocation^*)$ and $ \mathrm{Cov_{k'}}(\nudgeallocation) = \mathrm{Cov_{k'}}(\nudgeallocation^*)$ for all $k' < k$.
\end{definition}

\begin{definition}[{Prefix-maximal coverage}]
An allocation $\nudgeallocation^*$ has \emph{prefix-maximal coverage} among all feasible allocations given $Z,\nudgecaps$ if, for every feasible allocation $\nudgeallocation$ and every $k\geq 1$,
\[
\sum_{r=1}^k \mathrm{Cov_r}(\nudgeallocation^*)
\geq
\sum_{r=1}^k \mathrm{Cov_r}(\nudgeallocation).
\]
\end{definition}
Note that prefix-maximal coverage implies lexicographic maximum coverage.

We show that, perhaps surprisingly, \textit{any} strictly concave $u$ yields prefix-maximal coverage, and therefore lexicographic maximum coverage. 

\begin{restatable}[Concavity maximizes coverage]{theorem}{concavitythm}
\label{thm:concave_beats_linear_cov}Let $u$ be increasing and strictly discretely concave. Then, $\nudgeallocation$ is an optimum of $\textbf{\methodabbr}(\nudgecaps, u)$ $\iff$ it has prefix-maximal coverage.
\end{restatable}

We first prove \Cref{thm:concave_beats_linear_cov} by showing that the feasible allocations induce a matroid, so the greedy algorithm (which produces prefix-maximal coverage) is optimal. For intuition, we also provide a direct construction in which any allocation that does not maximize $\mathrm{Cov_1}$ can be improved by reallocating nudges from students with multiple nudges to those with no nudges, increasing the concave objective and $\mathrm{Cov_1}$.\footnote{Note that the result is trivially true if $Z_{ij}=1$ for all $i,j$, since then we can arbitrarily transfer nudges to those without nudges; but for general $Z$ we need to show transfer feasibility, which we do so via an alternating path construction.} We then use the resulting coverage property to prove \Cref{thm:cna_matches}.

\FloatBarrier

\usetikzlibrary{arrows.meta, positioning, fit, decorations.pathreplacing}

\pgfdeclarelayer{bg}
\pgfsetlayers{bg,main}
\begin{figure}[t]
\centering
\begin{tikzpicture}[
    >=Stealth,
    vtx/.style={circle, draw, thick, minimum size=8mm, inner sep=0pt, fill=white, font=\scriptsize},
    term/.style={circle, draw, thick, fill=black, minimum size=6mm, inner sep=0pt},
    lab/.style={font=\scriptsize},
    background/.style={rounded corners, fill=gray!15, inner sep=8pt}
]

\node[term] (s) at (0,0) {};
\node[lab,left=2mm of s] {$s$};

\node[vtx] (i1) at (3,4.0) {$(i,1)$};
\node[vtx] (i2) at (3,3.0) {$(i,2)$};
\node[lab] (idots) at (3,2.4) {$\vdots$};
\node[vtx] (ir) at (3,1.6) {$\left(i,K_i\right)$};

\begin{pgfonlayer}{bg}
\node[background,fit=(i1)(i2)(ir)] (graybox1) {};
\end{pgfonlayer}
\node[lab,below=0.5mm of graybox1.south] {Copies of $i$};

\node[lab] (studvdots) at (3,0) {$\vdots$};

\node[vtx] (ip1) at (3,-1.2) {$(i',1)$};
\node[vtx] (ip2) at (3,-2.2) {$(i',2)$};
\node[lab] (ipdots) at (3,-2.8) {$\vdots$};
\node[vtx] (ipr) at (3,-3.6) {$\left(i',K_{i'}\right)$};

\node[lab] (ipdots2) at (3,-5.2) {$\vdots$};

\begin{pgfonlayer}{bg}
\node[background,fit=(ip1)(ip2)(ipr)] (graybox2) {};
\end{pgfonlayer}
\node[lab,below=0.5mm of graybox2.south] {Copies of $i'$};

\node[draw,rounded corners,inner sep=4mm,
      fit=(ip1)(ipr)(studvdots)(i1)(ipdots2)] (studbox) {};
\node[lab,below=1mm of studbox.south] {Student-copy nodes ($E$)};

\node[vtx] (vij) at (6.5,2.4) {$v_{ij}$};
\node[lab] (pairdots1) at (6.5,1.8) {$\vdots$};

\node[lab] (pairdots2) at (6.5,0.3) {$\vdots$};

\node[vtx] (vipj) at (6.5,-0.9) {$v_{i'j}$};
\node[vtx] (vijprime) at (6.5,-1.9) {$v_{i'j'}$};
\node[lab] (pairdots3) at (6.5,-2.5) {$\vdots$};

\node[lab] (pairdots4) at (6.5,-4.0) {$\vdots$};

\begin{pgfonlayer}{bg}
\node[background,fit=(vij)(pairdots1)] (graybox3) {};
\node[background,fit=(vipj)(pairdots3)] (graybox4) {};
\end{pgfonlayer}
\node[lab,below=0.5mm of graybox3.south] {Pairs for $i$};
\node[lab,below=0.5mm of graybox4.south] {Pairs for $i'$};

\node[draw,rounded corners,inner sep=4mm,
      fit=(vij)(pairdots1)(vij)(vijprime)(pairdots4)] (pairbox) {};

\node[lab,below=1mm of pairbox.south] {Pair nodes};

\node[vtx] (j) at (10,1.0) {$j$};
\node[lab] (jvdots) at (10,0) {$\vdots$};
\node[vtx] (jprime) at (10,-1.0) {$j'$};
\node[lab] (jvdots2) at (10,-2.0) {$\vdots$};

\node[draw,rounded corners,inner sep=4mm,
      fit=(j)(jvdots2)] (jbox) {};

\node[lab,below=1mm of jbox.south] {Program nodes};

\node[term] (t) at (12.5,0) {};
\node[lab,right=2mm of t] {$t$};

\draw[->] (s) -- node[lab,midway,above] {$1$} (i1);
\draw[->] (s) -- (i2);
\draw[->] (s) -- (ir);
\draw[->] (s) -- (ip1);
\draw[->] (s) -- (ip2);
\draw[->] (s) -- (ipr);

\draw[->] (i1) -- node[lab,midway,above] {$1$} (vij);
\draw[->] (i2) -- (vij);
\draw[->] (ir) -- (vij);

\draw[->] (ip1) -- (vijprime);
\draw[->] (ip2) -- (vijprime);
\draw[->] (ipr) -- (vijprime);

\draw[->] (ip1) -- (vipj);
\draw[->] (ip2) -- (vipj);
\draw[->] (ipr) -- (vipj);

\draw[->] (vipj) -- (j);
\draw[->] (vijprime) -- (jprime);
\draw[->] (vij) -- node[lab,midway,above] {$1$} (j);

\draw[->] (j) -- node[lab,midway,above] {$\nudgecaps_j$} (t);
\draw[->] (jprime) -- node[lab,midway,above] {$\nudgecaps_{j'}$} (t);

\end{tikzpicture}
\caption{\textbf{Diagram for proof of \Cref{thm:concave_beats_linear_cov}: flow network encoding feasible nudge allocations.} Each unit of $s\to t$ flow corresponds to one nudge. The first column contains $K_i=\sum_j Z_{ij}$ copies of each student $i$, where each copy $(i,r)$ represents a potential nudge. The second column contains pair nodes $v_{ij}$ for eligible pairs $(i,j)$ with $Z_{ij}=1$. The third column contains one node for each program $j$, with an edge of capacity $\nudgecaps_j$ to the sink. A feasible flow corresponds to a feasible nudge allocation: each student copy sends flow to at most one program, flow is allowed only between nudge-eligible pairs $(i,j)$, each student sends at most 1 unit of flow to any program, and each program $j$ receives at most $\nudgecaps_j$ nudges. Edges $s\to (i,r)$ enforce ``each copy used at most once.'' Edges $(i,r)\to v_{ij}$ correspond to eligible nudges $Z_{ij}$. Edges $v_{ij}\to j$ enforce that a student cannot be nudged to the same program twice (only one copy of student $i$ can send flow to $v_{ij}$). Edges $j\to t$ enforce program nudge capacities.}
\label{fig:flow_matroid}
\end{figure}
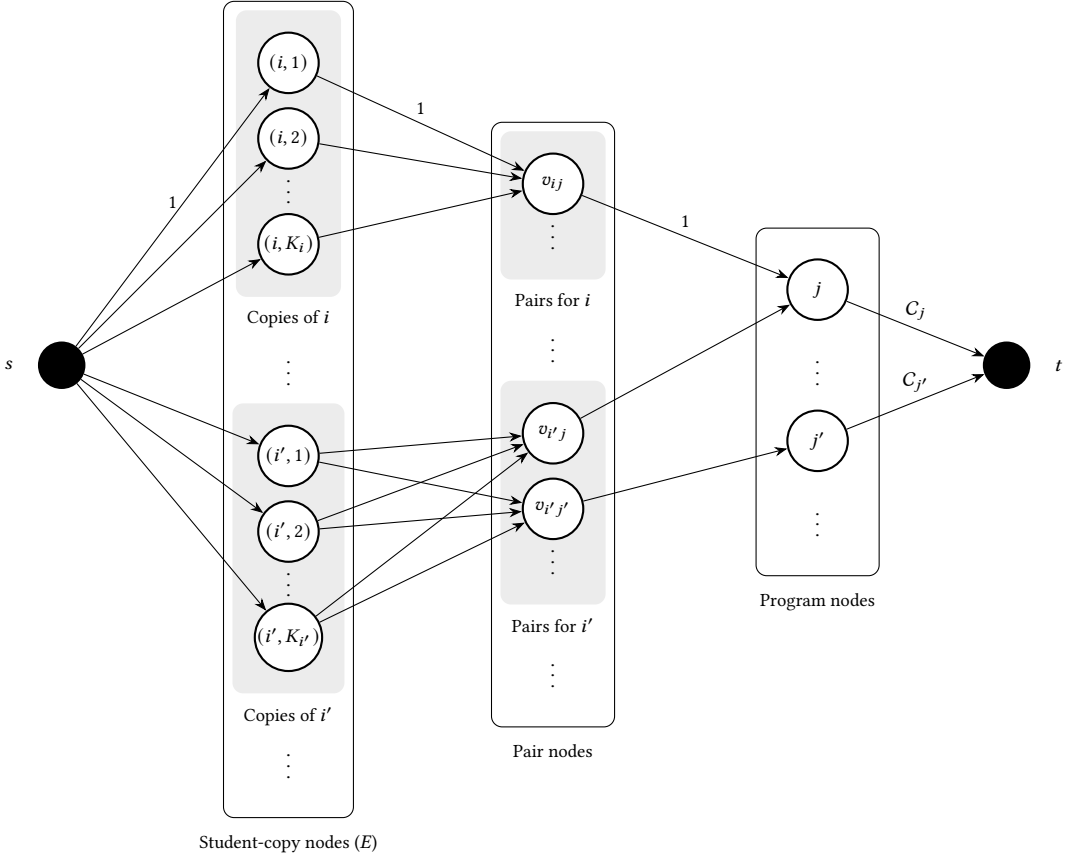

\paragraph{Setup.}
Suppose there are $n$ students and $m$ nudge-eligible programs. Recall the feasibility constraints
\[
\sum_{i=1}^n \nudgeallocation_{ij} \le \nudgecaps_j \quad \forall j,
\qquad
\nudgeallocation_{ij}\le Z_{ij}\quad \forall i,j,
\qquad
\nudgeallocation_{ij}\in\{0,1\}.
\]
Let $k_i(\nudgeallocation):=\sum_j \nudgeallocation_{ij}$, {$\widehat{\nudgecaps}_j:=\min\{\nudgecaps_j,\sum_i Z_{ij}\}$,} and $B:=\sum_{j=1}^m {\widehat{\nudgecaps}_j}$.
Define discrete marginals $\Delta(t):=u(t)-u(t-1)$ for $t\ge 1$. Without loss of generality, assume $u(0)=0$. 

\concavitythm*

\begin{proof}
We first construct a network flow problem in which (1) feasible integral flows encode exactly the feasible nudge allocations; and (2) the subsets of student nodes that can be simultaneously included in an integer flow form the independent sets of a matroid. Then, we use the fact that the greedy algorithm is optimal for maximizing linear weights over matroids to show that the optimizer of $\textbf{\methodabbr}(\nudgecaps,u)$ corresponds to a prefix-maximal nudge allocation when $u$ is strictly concave.

\paragraph{Step 1: Build a network whose feasible flows encode feasible nudge allocations.}

The network is illustrated in \Cref{fig:flow_matroid}. We construct a layered graph with a source, student-copy nodes, intermediate pair nodes, program nodes, and a sink.

A feasible nudge allocation $\nudgeallocation$ assigns each student a set of nudges to distinct programs where they are nudge-eligible, subject to nudge capacities. To model this, we create multiple copies of each student, where each copy represents one potential nudge for that student. Pair nodes enforce that no student is nudged to the same program twice (at most one copy of student $i$ can send flow to program $j$). 

Formally, for each student $i$, define the maximum possible number of distinct nudges they could receive as 
\[\
K_i := \sum_{j=1}^m Z_{ij}.
\]
Create a \emph{copy ground set}
\[
E := \{(i,r): i\in[n],\ r\in[K_i]\}.
\]
A unit of flow through node $(i,r)$ represents allocating one nudge to student $i$.

We build a directed graph $D=(V,A)$ with source $s$ and sink $t$ as follows:

\begin{itemize}
\item \textit{Student-copy nodes}: for each copy element $e=(i,r)\in E$, create a node $(i,r)$ and add an edge $s \rightarrow (i,r)$ with capacity 1. 
\item \textit{Pair nodes}: for each eligible pair $(i,j)$ with $Z_{ij}=1$, create a node $v_{ij}$. For every copy $(i,r)$ and program $j$ with $Z_{ij}=1$, add an edge $(i,r)\rightarrow v_{ij}$ with capacity 1.
\item \textit{Program nodes}: for each program $j$, create a node $j$. Add edges $v_{ij}\rightarrow j$ with capacity 1 for all $i$. Add an edge $j\rightarrow t$ with capacity $\nudgecaps_j$. 
\end{itemize}

Pair nodes enforce that each student can be nudged to a given program at most once: since the edge $v_{ij}\rightarrow j$ has capacity 1, at most one copy of student $i$ can send flow to program $j$. The edges $j\rightarrow t$ enforce program capacities $\nudgecaps_j$, and the edges $s\rightarrow(i,r)$ ensure that each copy of a student receives at most one nudge.

Feasible flows in this network correspond exactly to feasible nudge allocations $\nudgeallocation$ (see caption of \Cref{fig:flow_matroid} for a full mapping of flow constraints to nudge allocation constraints).

\paragraph{Step 2: The matchable sets of student-copy nodes form a matroid.}

To state the matroid property cleanly, we convert the program capacities to unit-capacity terminals.
For each program $j$, replace the $j\to t$ edge of capacity $\nudgecaps_j$ with $\nudgecaps_j$ distinct terminal
vertices $t_{j,1},\dots,t_{j,\nudgecaps_j}$, and add unit-capacity edges $j \to t_{j,r}$ for $r\in[\nudgecaps_j]$.
Let $T:=\{t_{j,r}: j\in[m],\, r\in[\nudgecaps_j]\}$ denote the set of terminals.

For a subset $S\subseteq E$ of student-copy nodes, call $S$ \emph{linkable} if there exist $|S|$
edge-disjoint directed paths in $D$ that start at the elements of $S$ and end at distinct terminals in $T$.
(Equivalently, there exists an integral flow that sends one unit from the source through each element of $S$
and routes these units to distinct terminals in $T$.)

Define
\[
\mathcal{I} \;:=\; \{S\subseteq E:\ S \text{ is linkable to } T \text{ in } D\}.
\]
A standard result in matroid theory asserts that $(E,\mathcal{I})$ is a matroid (a gammoid, specifically) \cite{schrijver2003combinatorial}. Thus the family of student-copy sets that can simultaneously be matched to terminals (which correspond to nudges) forms a matroid on ground set $E$. (In the directed line graph of $D$, these paths correspond to vertex-disjoint paths from vertices representing edges $s\to e$ ($e\in S$) to vertices representing edges entering $T$).

{For each program $j$, choose any $\widehat{\nudgecaps}_j$ eligible applicants. Because applicants have no nudge limit, these choices form a feasible allocation using all $B$ nudges. Hence this}
matroid has rank $B$ and {allocations using all $B$ nudges} correspond exactly to bases of the matroid.

\paragraph{Step 3: Linearize the concave objective via marginal weights on copies.}
Assign each element $e=(i,r)\in E$ the weight
\[
w(i,r) := \Delta(r)=u(r)-u(r-1).
\]\smallskip\noindent\emph{Normalization.}
We may assume $u(0)=0$ without loss of generality: replacing $u$ by $\tilde u(r):=u(r)-u(0)$ shifts the
objective by the constant $n\,u(0)$ and does not change monotonicity or discrete concavity.

\smallskip\noindent\emph{Canonical copies.}
Because all copies $(i,r)$ of a fixed student $i$ have identical outgoing edges in $D$, linkability of a set
$S\subseteq E$ depends only on the multiplicities $k_i(S):=\big|\{r:(i,r)\in S\}\big|$.
In particular, if $S\in\mathcal{I}$ and $S'$ is obtained by permuting which copies of student $i$ are selected
(holding $k_i$ fixed for all $i$), then $S'\in\mathcal{I}$ as well.

Therefore, among optimal solutions we may restrict attention to \emph{canonical} sets that, for each student $i$,
select the first $k_i(S)$ copies $\{(i,1),\dots,(i,k_i(S))\}$.
For such a canonical set,
\[
\sum_{(i,r)\in S} w(i,r)
= \sum_{i=1}^n \sum_{r=1}^{k_i(S)} \Delta(r)
= \sum_{i=1}^n u\!\left(k_i(S)\right),
\]
using the telescoping identity $u(k)=\sum_{r=1}^k (u(r)-u(r-1))$ and $u(0)=0$.

\paragraph{Step 4: Greedy optimality and prefix maximality.}
By the matroid greedy theorem (see, e.g., \citet{oxley2006matroid,schrijver2003combinatorial}), a maximum-weight base is obtained by sorting elements by nonincreasing weight and adding each element if it preserves independence.

If $u$ is strictly discretely concave, then $\Delta(1)>\Delta(2)>\cdots$, so weights strictly decrease in level index $r$.

Partition the ground set by level: $E^{(t)}:=\{(i,t): t\le K_i\}$.
Since $w(i,t)$ depends only on $t$, the greedy algorithm processes all of $E^{(1)}$ first, then $E^{(2)}$, etc.
Let $S^\star$ be any maximum-weight base. Under strictly decreasing weights, a standard consequence of the matroid greedy theorem is that
for every $t\ge 1$, $S^\star$ maximizes the prefix intersection size
\[
\left|S^\star \cap \bigcup_{r=1}^t E^{(r)}\right|
\]
among all bases. Because we may restrict attention to canonical bases, and for any canonical base $S$,
\[
\left|S\cap\bigcup_{r=1}^t E^{(r)}\right|
=
\sum_{r=1}^t \Cov_r(S),
\]
every optimum $S^\star$ therefore satisfies, for every base $S$ and every $t\geq 1$,
\[
\sum_{r=1}^t \Cov_r(S^\star)
\geq
\sum_{r=1}^t \Cov_r(S).
\]
Any feasible allocation can be extended to a base, which weakly increases each prefix sum, so this inequality also holds relative to every feasible allocation. 
(Intuitively: $\Delta(t)>\Delta(t+1)\ge \Delta(t+2)\ge\cdots$, and so the greedy algorithm first picks as many $t=1$ elements as possible, and only picks a $t=2$ element (a second nudge for a student) if it cannot pick any more $t=1$ elements, and so on.)

This prefix maximization implies lexicographic maximum coverage. A student contributes a copy to $S\cap E^{(r)}$ if and only if $k_i(S)\ge r$. Thus, for $t = 1$, $\left|S^\star \cap E^{(1)}\right| = \Cov_1(S^\star)$ is maximized. Holding $\Cov_1$ fixed, the above then maximizes $\left|S^\star \cap E^{(2)}\right| = \Cov_2(S^\star)$, the number of students with at least 2 nudges (2 copies with flow through them). Continuing, the above greedy solutions maximize the tail vector $(\Cov_1(S),\Cov_2(S),\dots)$ lexicographically, and are optimal by the matroid greedy theorem. Thus, every optimal solution has prefix-maximal (and therefore lexicographic maximum) coverage.\\

The reverse direction is straightforward: if $S$ is not optimal, then there exists an optimal solution $S^\star$ with strictly higher weight. Let $t$ be the smallest level at which $S$ and $S^\star$ differ. Then, $S^\star$ has more elements in $E^{(t)}$ than $S$, so $\Cov_t(S^\star)>\Cov_t(S)$, contradicting  prefix maximality of coverage.

\end{proof}

\subsubsection{Intuition for \Cref{thm:concave_beats_linear_cov}}
The proof of \Cref{thm:concave_beats_linear_cov} uses matroid-based arguments. For intuition on the result, next we prove a weaker result (that with strict concavity, every optimum maximizes $\Cov_1$) more directly, without invoking such machinery. 

A nudge allocation is a binary matrix $\nudgeallocation\in\{0,1\}^{n\times m}$, where $\nudgeallocation_{ij}=1$ indicates nudging student $i$ to program $j$, subject to the program-level nudge capacities and eligibility constraints
\begin{equation}
\sum_{i=1}^n \nudgeallocation_{ij} \le \nudgecaps_j \quad \forall j,
\qquad
\nudgeallocation_{ij} \le Z_{ij} \quad \forall i,j.
\label{eq:caps}
\end{equation}
Given $\nudgeallocation$, define each student's nudge list size $k_i(\nudgeallocation)=\sum_{j=1}^m \nudgeallocation_{ij}$ and the coverage
\[
\mathrm{Cov}(\nudgeallocation)=\sum_{i=1}^n \Ind\{k_i(\nudgeallocation)\ge 1\}.
\]
Let $N(\nudgeallocation)=\sum_{i,j} \nudgeallocation_{ij}=\sum_i k_i(\nudgeallocation)$ be the total number of nudges.

\begin{definition}[Maximal feasible coverage under full utilization]
\label{def:cov_max}
{For $\nudgecaps\leq\overline c$, }let $B=\sum_{j=1}^m \nudgecaps_j$.
Define the maximal feasible coverage among full-utilization (equivalently, linear-optimal) allocations as
\[
\mathrm{Cov}_{\max}(Z,\nudgecaps)
\;:=\;
\max\Bigl\{\mathrm{Cov}(\nudgeallocation): \nudgeallocation \text{ satisfies \eqref{eq:caps} and } N(\nudgeallocation)=B\Bigr\}.
\]
\end{definition}

As a baseline, one can maximize the total number of nudges (a linear objective):
\begin{equation}
\max_{\nudgeallocation} \ \sum_{i=1}^n k_i(\nudgeallocation) \quad \text{s.t.\ \eqref{eq:caps}.}
\label{eq:linear_problem}
\end{equation}

Now instead fix an increasing, discretely concave $u:\mathbb{Z}_{\ge 0}\to\mathbb{R}_{\ge 0}$. Let $u(0)=0$ without loss of generality.
The concave surrogate objective solves
\begin{equation}
\max_{\nudgeallocation} \ \sum_{i=1}^n u\bigl(k_i(\nudgeallocation)\bigr)\quad \text{s.t.\ \eqref{eq:caps}.}
\label{eq:concave_problem}
\end{equation}

\begin{theorem}[Concavity maximizes number of students with at least one nudge]
\label{thm:intuition}
{For $\nudgecaps\leq\overline c$, }fix any increasing, discretely concave $u$ with $u(0)=0$.
\begin{enumerate}
\item There exists an optimizer $\nudgeallocation^\star$ of \eqref{eq:concave_problem} satisfying
\[
\mathrm{Cov}(\nudgeallocation^\star)=\mathrm{Cov}_{\max}(Z,\nudgecaps).
\]
In particular, for any optimizer $\tilde{\nudgeallocation}$ of the linear problem \eqref{eq:linear_problem}, we have $\mathrm{Cov}(\nudgeallocation^\star)\ge \mathrm{Cov}(\tilde{\nudgeallocation})$.
\item If, in addition, $u$ has \emph{strictly diminishing returns} between the first and second nudge, i.e.,
\[
u(1)-u(0)\;>\;u(2)-u(1),
\]
then \emph{every} optimizer $\nudgeallocation^\star$ of \eqref{eq:concave_problem} satisfies $\mathrm{Cov}(\nudgeallocation^\star)=\mathrm{Cov}_{\max}(Z,\nudgecaps)$.
\end{enumerate}
\end{theorem}

\begin{proof} 
\noindent\emph{High-level intuition.} We argue by contradiction.
Suppose a full-utilization maximum-coverage allocation $Y$ covers more students than a concave-optimal allocation $\nudgeallocation$. Note that we can restrict attention to full-utilization concave-optimal allocations; {$\nudgecaps\leq\overline c$ implies} $|\{i: Z_{ij}=1\}|\ge \nudgecaps_j$ for every program $j$, so if some program $j$ has unused capacity (i.e., $\sum_i \nudgeallocation_{ij}<\nudgecaps_j$), then there exists an eligible student $i$ with $\nudgeallocation_{ij}=0$; setting $\nudgeallocation_{ij}=1$ increases $k_i(\nudgeallocation)$ by one and weakly increases $\sum_i u(k_i(\nudgeallocation))$. Repeating this operation yields a full-utilization allocation with objective weakly larger. 

Then, we show that, from the concave-optimal allocation, it is feasible to move a nudge from a student with at least 2 nudges to a student with 0 nudges, via an alternating-path swap. This results in increasing coverage and (by diminishing returns) not decreasing the concave objective (strictly increasing it under the strict condition).

\begin{figure}[tb]
    \centering
    \begin{tikzpicture}[
        font=\small,
        student/.style={circle, draw, fill=gray!10, minimum size=5.8mm, inner sep=1pt},
        program/.style={rectangle, rounded corners=1mm, draw, fill=gray!10, minimum width=8mm, minimum height=6mm, inner sep=1pt},
        xedge/.style={draw=blue!70!black, line width=1pt},
        yedge/.style={draw=orange!80!black, dashed, line width=1pt},
        removed/.style={draw=black!35, dashed, line width=0.8pt},
    ]
        \begin{scope}
            \node[student] (uA) at (0,2.0) {$u$};
            \node[student] (s1A) at (0,1.0) {$s_1$};
            \node[student] (s2A) at (0,0.0) {$s_2$};
            \node[student] (vA) at (0,-1.0) {$v$};

            \node[program] (p1A) at (2.2,1.5) {$p_1$};
            \node[program] (p2A) at (2.2,0.5) {$p_2$};
            \node[program] (p3A) at (2.2,-0.5) {$p_3$};

            \draw[yedge] (uA) -- (p1A);
            \draw[xedge] (s1A) -- (p1A);
            \draw[yedge] (s1A) -- (p2A);
            \draw[xedge] (s2A) -- (p2A);
            \draw[yedge] (s2A) -- (p3A);
            \draw[xedge] (vA) -- (p3A);

                \node[anchor=south] at (1.1,2.9) {\textbf{(a)} Alternating path $P$};
                \node[font=\scriptsize, anchor=west] at (0,2.3) {$k_u(\nudgeallocation)=0$};
                \node[font=\scriptsize, anchor=west] at (0,-1.35) {$k_v(\nudgeallocation)\ge 2$};
            \end{scope}

    \draw[->, line width=0.9pt] (3.6,0.5) -- (5.0,0.5) node[midway, above, font=\scriptsize] {flip along $P$};

        \begin{scope}[xshift=6.0cm]
            \node[student] (uB) at (0,2.0) {$u$};
            \node[student] (s1B) at (0,1.0) {$s_1$};
            \node[student] (s2B) at (0,0.0) {$s_2$};
        \node[student] (vB) at (0,-1.0) {$v$};

        \node[program] (p1B) at (2.2,1.5) {$p_1$};
        \node[program] (p2B) at (2.2,0.5) {$p_2$};
        \node[program] (p3B) at (2.2,-0.5) {$p_3$};

        \draw[xedge] (uB) -- (p1B);
        \draw[xedge] (s1B) -- (p2B);
        \draw[xedge] (s2B) -- (p3B);

        \draw[removed] (s1B) -- (p1B);
        \draw[removed] (s2B) -- (p2B);
            \draw[removed] (vB) -- (p3B);

                \node[anchor=south] at (1.1,2.9) {\textbf{(b)} After flipping ($\nudgeallocation'$)};
                \node[font=\scriptsize, anchor=west] at (0,2.3) {$k_u: 0\to 1$};
                \node[font=\scriptsize, anchor=west] at (0,-1.35) {$k_v: k\to k-1\ (\ge 1)$};
            \end{scope}

            \begin{scope}[yshift=-2.2cm, xshift=1.2cm]
                \draw[xedge] (0.0,0) -- (0.6,0);
                \node[font=\scriptsize, anchor=west] at (0.75,0) {$\nudgeallocation$-edge};
                \draw[yedge] (2.0,0) -- (2.6,0);
                \node[font=\scriptsize, anchor=west] at (2.75,0) {$Y$-edge};
                \draw[removed] (4.4,0) -- (5.0,0);
                \node[font=\scriptsize, anchor=west] at (5.15,0) {removed $\nudgeallocation$-edge};
            \end{scope}
    \end{tikzpicture}
    \caption{\textbf{Alternating-path exchange in \Cref{thm:intuition}} Panel (a) highlights an alternating path $P$ in the symmetric difference $E(\nudgeallocation)\triangle E(Y)$ between a concave-optimal allocation $\nudgeallocation$ and a full-utilization maximum-coverage allocation $Y$. Flipping $\nudgeallocation$ along $P$ (panel (b)) preserves program caps, increases coverage by one at the start vertex $u$, and does not decrease the concave objective by discrete concavity.}
    \label{fig:alt_path_flip}
\end{figure}
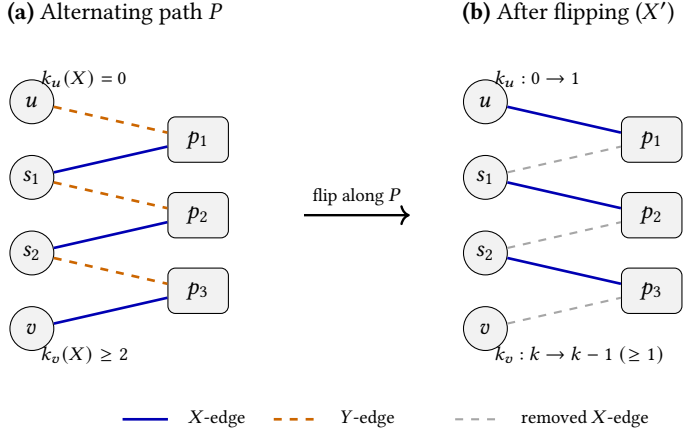

\paragraph{Step 1 (Pick benchmark allocations).}
Let $\mathcal{X}^\star$ be the set of optimizers of the concave problem \eqref{eq:concave_problem}, and choose $\nudgeallocation\in \mathcal{X}^\star$ that maximizes $\mathrm{Cov}(\nudgeallocation)$ over $\mathcal{X}^\star$. 
By monotonicity of $u$ and {the fact that $\nudgecaps\leq\overline c$ makes full utilization feasible}, we may assume $\nudgeallocation$ is full utilization, i.e., $N(\nudgeallocation)=B$.

Let $Y$ be the full-utilization allocation with $\mathrm{Cov}(Y)=\mathrm{Cov}_{\max}(Z,\nudgecaps)$ (such a $Y$ exists by Definition~\ref{def:cov_max}).
Assume for contradiction that $\mathrm{Cov}(Y)>\mathrm{Cov}(\nudgeallocation)$.

\paragraph{Step 2 (Build the symmetric-difference graph).}
Let $E(\nudgeallocation)=\{(i,j):\nudgeallocation_{ij}=1\}$ and define the symmetric difference subgraph
\[
H \;=\; (V,\ E(\nudgeallocation)\,\triangle\,E(Y)),
\]
where $\triangle$ denotes \emph{symmetric difference}: for sets $A,B$, $A\triangle B=(A\setminus B)\cup(B\setminus A)$.
Thus $E(\nudgeallocation)\triangle E(Y)$ consists of the edges present in exactly one of $\nudgeallocation$ and $Y$ (edges common to both are omitted).
We color edges in $E(\nudgeallocation)\setminus E(Y)$ as $\nudgeallocation$ and edges in $E(Y)\setminus E(\nudgeallocation)$ as $Y$.
Here $V$ consists of all students and all programs.

\paragraph{Step 3 (Decompose $H$ into alternating cycles/paths).}
Since both $\nudgeallocation$ and $Y$ are full utilization, every program vertex $j$ has equal $\nudgeallocation$-degree and $Y$-degree in $H$.
To see this, note that $\sum_i \nudgeallocation_{ij}=\sum_i Y_{ij}=\nudgecaps_j$ and edges common to $\nudgeallocation$ and $Y$ do not appear in $H$, so the number of $\nudgeallocation$-only edges incident to $j$ equals the number of $Y$-only edges incident to $j$.
Equivalently, writing $\deg_H^{(\nudgeallocation)}(v)$ and $\deg_H^{(Y)}(v)$ for the number of incident $\nudgeallocation$-colored and $Y$-colored edges at a vertex $v$ in $H$, we have $\deg_H^{(\nudgeallocation)}(j)=\deg_H^{(Y)}(j)$ for every program $j$.

This balance implies that $H$ can be decomposed into edge-disjoint alternating cycles and alternating paths.
Here an \emph{alternating} cycle/path is one whose edges alternate in color ($Y, \nudgeallocation, Y, \nudgeallocation,\dots$), and \emph{edge-disjoint} means that the components partition $E(H)$.
To construct such a decomposition, fix an arbitrary pairing at each program vertex $j$ that matches each incident $Y$-edge to a distinct incident $\nudgeallocation$-edge (possible since $\deg_H^{(Y)}(j)=\deg_H^{(\nudgeallocation)}(j)$).
Starting from any unused edge, follow a maximal alternating walk that (i) at programs always continues along the paired opposite-color edge and (ii) at students takes any unused opposite-color edge if one exists; this walk cannot get stuck at a program and therefore yields either an alternating cycle or an alternating path with student endpoints.
Delete its edges and repeat until no edges remain.

We will orient each alternating path to start with a $Y$-edge (equivalently, to end with an $\nudgeallocation$-edge); since $H$ is bipartite, every student-to-student path has even length and hence this orientation is well-defined.
If the end of one oriented path equals the start of another, concatenate them and repeat; then, no vertex is simultaneously the end of one path and the start of another.
In this resulting decomposition, any end vertex $v$ of an oriented alternating path satisfies $\deg_H^{(\nudgeallocation)}(v)>\deg_H^{(Y)}(v)$: cycles and internal vertices contribute equal numbers of incident $\nudgeallocation$- and $Y$-edges, while an oriented path contributes one extra incident $Y$-edge at its start and one extra incident $\nudgeallocation$-edge at its end.
See Figure~\ref{fig:alt_path_flip} for an illustration.

\paragraph{Step 4 (Locate a path from a student with 0 nudges in $x$ to a student with at least 2 nudges in $\nudgeallocation$).}
Recall that $k_i(\nudgeallocation)=\sum_{j=1}^m \nudgeallocation_{ij}$ denotes student $i$'s nudge list size under allocation $\nudgeallocation$.
Let $S_{\nudgeallocation}=\{i:k_i(\nudgeallocation)>0\}$ and $S_Y=\{i:k_i(Y)>0\}$ denote the sets of covered students under $\nudgeallocation$ and $Y$, so $\mathrm{Cov}(\nudgeallocation)=|S_{\nudgeallocation}|$ and $\mathrm{Cov}(Y)=|S_Y|$.
Since $\mathrm{Cov}(Y)>\mathrm{Cov}(\nudgeallocation)$, we have $|S_Y\setminus S_{\nudgeallocation}|>|S_{\nudgeallocation}\setminus S_Y|$.

Define the ``newly covered'' and ``fragile uncovered'' sets
\[
U=S_Y\setminus S_{\nudgeallocation}=\{i: k_i(\nudgeallocation)=0,\ k_i(Y)>0\},
\qquad
L=\{i: k_i(\nudgeallocation)=1,\ k_i(Y)=0\}.
\]
Here $U$ are students covered by $Y$ but not by $\nudgeallocation$, while $L\subseteq S_{\nudgeallocation}\setminus S_Y$ are students who have exactly one $\nudgeallocation$-edge and no $Y$-edge (so removing a single $\nudgeallocation$-edge would make them uncovered).
Since $|U|=|S_Y\setminus S_{\nudgeallocation}|>|S_{\nudgeallocation}\setminus S_Y|\ge |L|$, we have $|U|>|L|$.

Consider the edge-disjoint decomposition of $H$ into alternating cycles and alternating paths from Step~3, with every alternating path oriented to start with a $Y$-edge (and hence to end with an $\nudgeallocation$-edge).
For a student $i$, edges common to $\nudgeallocation$ and $Y$ do not appear in $H$, so
\[
k_i(\nudgeallocation)-k_i(Y)=\deg_H^{(\nudgeallocation)}(i)-\deg_H^{(Y)}(i).
\]
In particular, each $i\in U$ has $\deg_H^{(Y)}(i)=k_i(Y)\ge 1$ and $\deg_H^{(\nudgeallocation)}(i)=0$, so every $i\in U$ must be the start of at least one alternating path.
Each $i\in L$ has $\deg_H^{(\nudgeallocation)}(i)=1$ and $\deg_H^{(Y)}(i)=0$, so $i$ is the end of exactly one alternating path (since we need a unique incoming edge in $\nudgeallocation$ for an alternating path to end at $i$).
Since $|U|>|L|$, there exists an alternating path $P$ whose start vertex $u\in U$ and whose end vertex $v$ is not in $L$. 
By the endpoint property from Step~3, we have $\deg_H^{(\nudgeallocation)}(v)>\deg_H^{(Y)}(v)$, and hence $k_v(\nudgeallocation)>k_v(Y)$ by the displayed identity.
Finally, since $v\notin L$, this implies $k_v(\nudgeallocation)\ge 2$ (otherwise $k_v(\nudgeallocation)=1$ forces $k_v(Y)=0$ and hence $v\in L$).

\paragraph{Step 5 (Flip along the path).}
We have located an alternating path $P$ in $H$ starting at $u\in U$ with $k_u(\nudgeallocation)=0$ and ending at $v$ with $k_v(\nudgeallocation)\ge 2$. We can now reallocate nudges along this path to increase coverage without decreasing the concave objective (and strictly increasing it under the additional strict-diminishing-returns condition).

Define $\nudgeallocation'$ by flipping $\nudgeallocation$ along $P$: set $\nudgeallocation'_{ij}=1$ for edges $(i,j)\in E(Y)\cap P$, set $\nudgeallocation'_{ij}=0$ for edges $(i,j)\in E(\nudgeallocation)\cap P$, and keep all other entries the same.
By alternation, every program on $P$ loses exactly one incident $\nudgeallocation$-edge and gains exactly one incident $Y$-edge, so $\nudgeallocation'$ remains feasible under program capacities.
Eligibility is also preserved since every added edge lies in $E(Y)$ and $Y$ satisfies $Y_{ij}\le Z_{ij}$.
All internal students on $P$ lose and gain one nudge, so their list sizes are unchanged; the endpoints satisfy
\[
k_u(\nudgeallocation')=k_u(\nudgeallocation)+1=1,
\qquad
k_v(\nudgeallocation')=k_v(\nudgeallocation)-1\ge 1.
\]
Thus $\mathrm{Cov}(\nudgeallocation')=\mathrm{Cov}(\nudgeallocation)+1$.

Finally, only $k_u$ and $k_v$ change, so the concave objective changes by
\[
\bigl(u(1)-u(0)\bigr)-\bigl(u(k_v(\nudgeallocation))-u(k_v(\nudgeallocation)-1)\bigr)\ \ge\ 0
\]
by discrete concavity, with strict inequality under the additional condition $u(1)-u(0)>u(2)-u(1)$ since $k_v(\nudgeallocation)\ge 2$.

We now close the contradiction.
Since $\nudgeallocation$ is concave-optimal and the objective under $\nudgeallocation'$ is weakly larger, $\nudgeallocation'$ is also concave-optimal.
But $\mathrm{Cov}(\nudgeallocation')>\mathrm{Cov}(\nudgeallocation)$ contradicts the choice of $\nudgeallocation$ as the maximum-coverage optimizer in $\mathcal{X}^\star$.
Therefore $\mathrm{Cov}(\nudgeallocation)=\mathrm{Cov}_{\max}(Z,\nudgecaps)$, proving part~(1).

For part~(2), assume $u(1)-u(0)>u(2)-u(1)$.
Then discrete concavity implies $u(k)-u(k-1)\le u(2)-u(1)$ for all $k\ge 2$, so the objective change above is strictly positive since $k_v(\nudgeallocation)\ge 2$.
Thus no optimizer of \eqref{eq:concave_problem} can satisfy $\mathrm{Cov}(\nudgeallocation)<\mathrm{Cov}_{\max}(Z,\nudgecaps)$, because applying the same construction to such an optimizer would strictly improve the objective.
Hence every optimizer attains maximal coverage, proving part~(2).
\end{proof}

\subsubsection{Proof of \Cref{thm:cna_matches}}
\cnamatchesthm*
\begin{proof}
Recall that $k_i(\nudgeallocation)=\sum_j\nudgeallocation_{ij}$, the number of nudges that applicant $i$ receives. Let $K_i=\sum_j Z_{ij}$ be the maximum number of nudges applicant $i$ \textit{can} receive, and recall that $q(k)$ is the probability that an applicant receiving $k$ nudges adopts one.

Let $K=\max_i K_i$, let $\Delta_r=q(r)-q(r-1)$, and set $\Delta_{K+1}=0$. For any feasible allocation $\nudgeallocation$, summation by parts gives expected adoption as
\[
\sum_i q(k_i(\nudgeallocation))
=
\sum_{r=1}^K \Delta_r\Cov_r(\nudgeallocation)
=
\sum_{t=1}^K
(\Delta_t-\Delta_{t+1})
\sum_{r=1}^t\Cov_r(\nudgeallocation).
\]
Discrete concavity makes every coefficient $\Delta_t-\Delta_{t+1}$ nonnegative. The coverage inequality established in the proof of \Cref{thm:concave_beats_linear_cov} therefore implies that $\nudgeallocation^*$ maximizes expected adoption among feasible allocations.

Under \Cref{assm:adoption}, every desirable match under any feasible $\nudgeallocation$ is made by an applicant who adopted a nudge, so
\[
\mathbb{E}\left[\operatorname{Matches}_Z(\nudgeallocation)\right]
\leq
\sum_i q(k_i(\nudgeallocation)).
\]

Applying the Oracle guarantee on acceptance rates in expectation for each program and summing over programs therefore lower-bounds desirable matches by $L$ times expected adoption. Hence, for every feasible $\nudgeallocation$,
\[
\mathbb{E}\left[\operatorname{Matches}_Z(\nudgeallocation^*)\right]
\geq
L\sum_i q(k_i(\nudgeallocation^*))
\geq
L\sum_i q(k_i(\nudgeallocation))
\geq
L\mathbb{E}\left[\operatorname{Matches}_Z(\nudgeallocation)\right].
\]
Taking the maximum over feasible $\nudgeallocation$ completes the proof.
\end{proof}

\subsection{Proof of algorithm guarantees (\Cref{thm:algorithm_guarantees})}
\label{proof:algo_guarantees}

\algothm*

\begin{proof} We prove each of the three guarantees in turn. 
Note that the binary search procedure invoked by \algoname to find the largest feasible value no greater than the current value of a program's nudge capacity always terminates since the search space is integral and finite (between $0$ and $\nudgecaps_j$), the feasibility check is well-defined, and by \Cref{assm:deterministic_sim} {$\mathcal O_j$} is non-increasing in the candidate value. The feasible set is nonempty because {the no-nudge allocation is safe, and increasing the other programs' nudge capacities cannot increase congestion at program $j$}.

\paragraph{Finite termination.} At every iteration, the algorithm either terminates or strictly decreases at least one coordinate of $\nudgecaps$. Thus $\sum_{j\in \programs} \nudgecaps_j$ strictly decreases by at least $1$ at each nonterminal iteration. Since $\nudgecaps_j \geq 0$ for all $j\in\programs$, the total sum is bounded below by $0$. Therefore the algorithm terminates after at most $1+\sum_{j\in\programs}\overline c_j=1+\sum_{i,j}Z_{ij}$ iterations, including the final unchanged iteration. Each iteration uses $O(\log(1+\overline c_j))$ oracle calls for binary search at program $j$, giving $O\!\left(\left(1+\sum_{j\in\programs}\overline c_j\right)\sum_{j\in\programs}\log(1+\overline c_j)\right)$ oracle calls in total. Since $\overline c_j\leq\sum_{i,j}Z_{ij}$ for every $j$, this gives the stated bound.

\paragraph{Safe nudge capacities.} At termination, $\nudgecaps^{(t)}=\nudgecaps^{(t-1)}$. For each program $j$, the update rule therefore verifies {$\mathcal O_j(\textbf{\methodabbr}(\nudgecaps))=1$} at the final nudge capacity vector. Thus, the constraints are satisfied at termination.

\paragraph{Greatest safe nudge capacities.} Let $\nudgecaps$ be the final nudge capacities at termination. Let $\nudgecaps'\leq\overline{c}$ be any nudge capacity vector satisfying the congestion constraints. We show inductively that $\nudgecaps'\leq\nudgecaps^{(t)}$. This holds initially because $\nudgecaps^{(0)}=\overline{c}$. Suppose it holds at iteration $t-1$. For each program $j$, $\nudgecaps'_j$ is feasible in program $j$'s update: it is no greater than $\nudgecaps_j^{(t-1)}$, and cross-program monotonicity implies that raising the other nudge capacities from $\nudgecaps'_{-j}$ to $\nudgecaps_{-j}^{(t-1)}$ cannot {cause program $j$'s oracle outcome to change from satisfied to violated} (intuitively, increasing capacity elsewhere cannot increase congestion at program $j$, by assumption). The update therefore chooses $\nudgecaps_j^{(t)}\geq\nudgecaps'_j$. Hence $\nudgecaps'\leq\nudgecaps^{(t)}$, and thus $\nudgecaps'\leq\nudgecaps$ at termination.
\end{proof}

\subsection*{AI Disclosure}
An LLM (GPT 5.2--5.6) was used as a proof assistant, such as to help draw explanatory figures, explain ideas, and formalize proofs. The authors came up with all primary ideas and stand by all statements, and nothing was used without validation and editing. AI was not used for data analysis, and did not have any access to any non-public data. 

\FloatBarrier
\section{Supplementary simulation details}

\subsection{Non-personalized recommendations}
\label{subsec:nonpersonalized_nudges}
In all of our main text discussion, the \naive approach refers to nudging students to every high-performing program where they are nudge-eligible (they have \high offer likelihood and the program is within 45 minutes of their home). Here, we compare this to the \textit{non-personalized} approach used in prior interventions, where \textit{every} student at each middle school receives the same list of all nearby, high-performing programs (i.e., recommendations do not incorporate individual offer likelihood), showing that even \naive personalization improves substantively over the non-personalized approach. 

As we discuss in the introduction, non-personalized interventions must either (a) be restricted to programs that most students would be admitted to (e.g., historically undersubscribed programs), or (b) risk steering many students to programs where they are ultimately rejected. Under our definition of high-performing programs, option (a) would have severely limited the scale of the intervention. Only 28 of 56 high-performing programs had any open capacity in the 2025 admissions cycle (without any nudges). Further, the number of high-performing programs where at least half of all students at target middle schools have \high offer likelihood estimates is only 4 of 56. To quantify the consequences of retaining these programs (option (b)), we simulate the allocation method of nudging all students at target middle schools to all high-performing programs within 45 minutes of their middle school and evaluate outcomes under 20\% adoption (following the same procedure as in the main text); without personalization, 41 of 56 high-performing programs reject \textit{over 50\%} of nudged applicants, as opposed to 5 of 56 under our \naive allocation method (and, of course, 0 of 56 under our \optimizesqrt allocation method).

\subsection{Displacement effects}
\label{sec:displacement_effects}
Our main congestion criterion in the deployment is preserving offer likelihoods for \high offer likelihood students. In particular, if a student is told by NYCPS that their chance of admission at a program is over 99\%, these predictions must remain reliable even under large-scale adoption of nudges. Here, we also analyze calibration of \medium and \low offer likelihoods, as well as the displacement effects across \textit{all} students, based on simulations using historical 2024-25 application cycle data. We simulate 20\% adoption under our \optimizesqrt nudge allocation. 

For \medium offer likelihood students, acceptance rates in simulation decline under nudge adoption but largely remain calibrated. For \medium offer likelihood students, the mean acceptance rate across nudge-eligible programs decreases from 63.8\% to 49.0\% (median: 69.2\% to 52.1\%), still well within the (1, 99)\% stated interval. Of the 53 nudge-eligible programs that assign \medium offer likelihood to any students (of 56 total), 46 are calibrated in the absence of nudges. In our simulations at 20\% adoption of nudges, acceptance rates for \medium offer likelihood students drop below the 1\% threshold for 2 programs. Meanwhile, the mean nudge-eligible program acceptance rate of \low offer likelihood students changed from 8.36\% to 5.60\% (median 0.3\% to 0.0\%). 

We next examine which students were displaced from nudge-eligible programs or newly match to nudge-eligible programs in the simulations, and we compare demographics between the groups. Among displaced students, 13.0\% attend target middle schools. This indicates that the displacements increase representation of students from target middle schools at nudge-eligible high schools (and in particular of students with \high offer likelihoods at these programs, indicating concordance with the program- or city-level desiderata). Increasing the number of these matches is in line with our definition of ``desirable matches'' used throughout the paper. Further, 75.7\% of displaced students are FRL-eligible, compared to 86.1\% of students who newly match to one these programs due to nudges. 

Overall, while nudges naturally induce some displacement, calibration is largely preserved and the resulting reallocation aligns with our objectives.

\clearpage
\subsection{Other tables and figures}

\begin{figure}[htbp]
  \centering
  \begin{subfigure}[t]{0.32\linewidth}
    \centering
    \includegraphics[width=\linewidth]{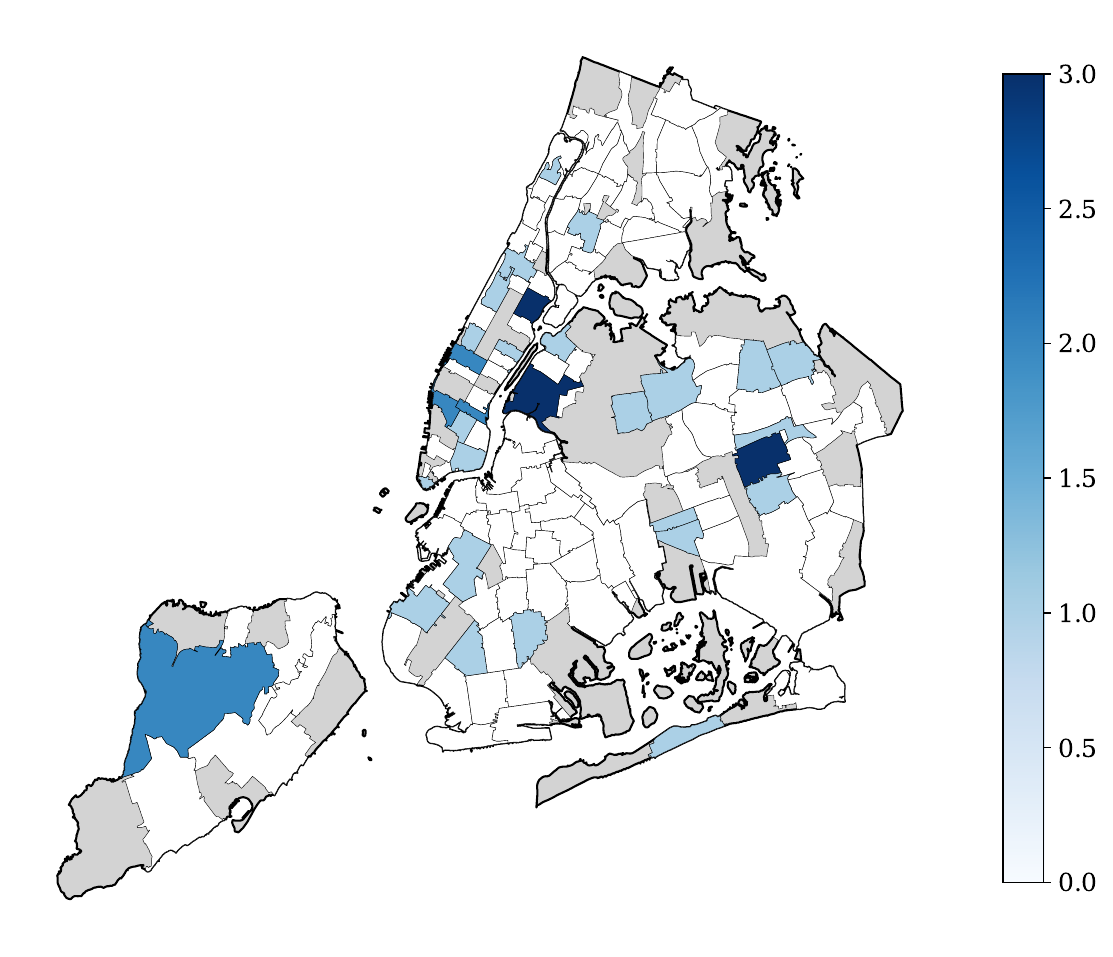}
    \caption{\textbf{Nudge-eligible high schools}}
    \label{fig:nudge_eligible_hs_zipcode_counts}
  \end{subfigure}
  \hfill
  \begin{subfigure}[t]{0.32\linewidth}
    \centering
    \includegraphics[width=\linewidth]{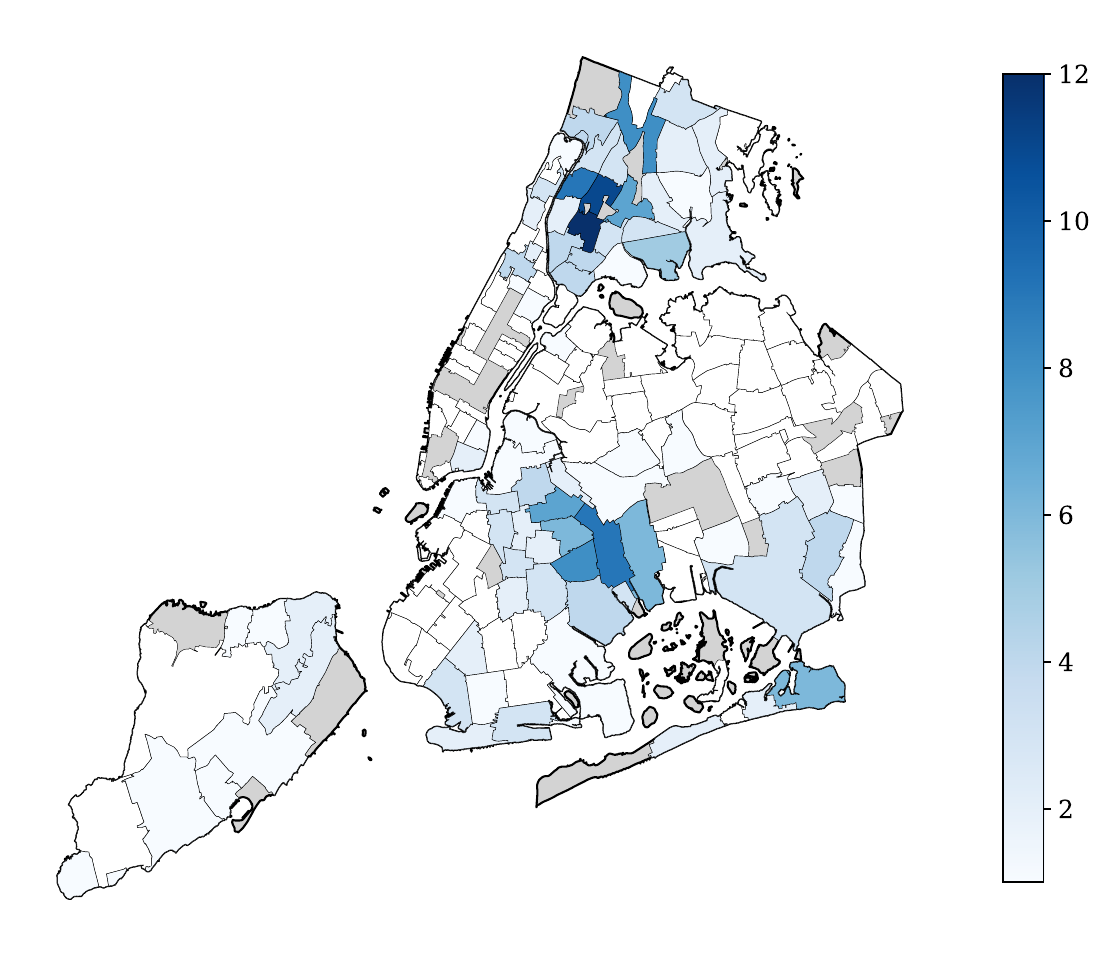}
    \caption{\textbf{Target middle schools}}
    \label{fig:target_ms_zipcode_counts}
  \end{subfigure}
  \hfill
  \begin{subfigure}[t]{0.32\linewidth}
    \centering
    \includegraphics[width=\linewidth]{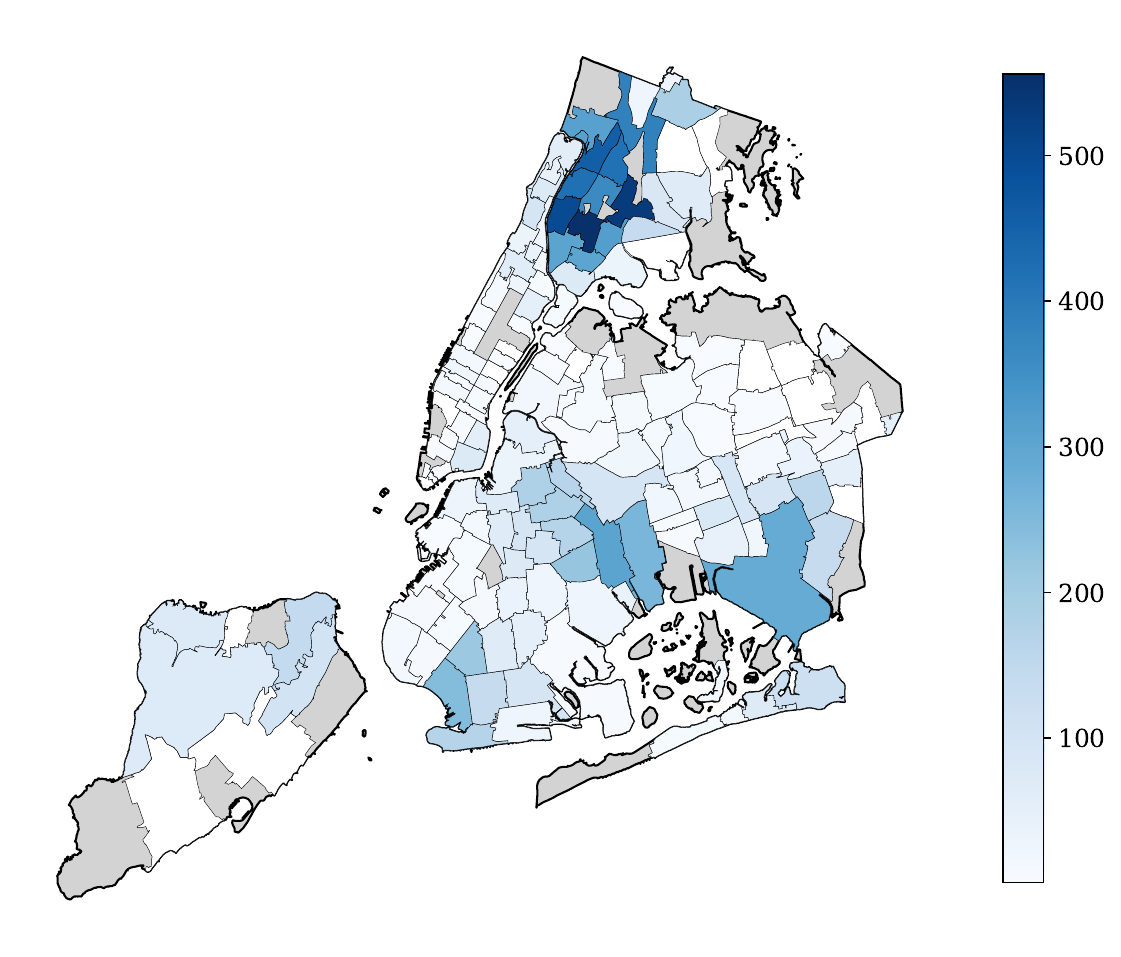}
    \caption{\textbf{Nudge-eligible students}}
    \label{fig:nudge_eligible_student_zipcode_counts}
  \end{subfigure}
  \caption{\textbf{Number of intervention schools and students by zip code.} Gray zip codes are locations where no high school, middle school, or student (respectively) is located at all.}
  \label{fig:counts_by_zipcode}
\end{figure}

\begin{figure}[htbp]
  \vspace{0.2cm}
  \centering
  \begin{subfigure}[t]{0.5\linewidth}
    \centering
    \includegraphics[width=0.75\linewidth]{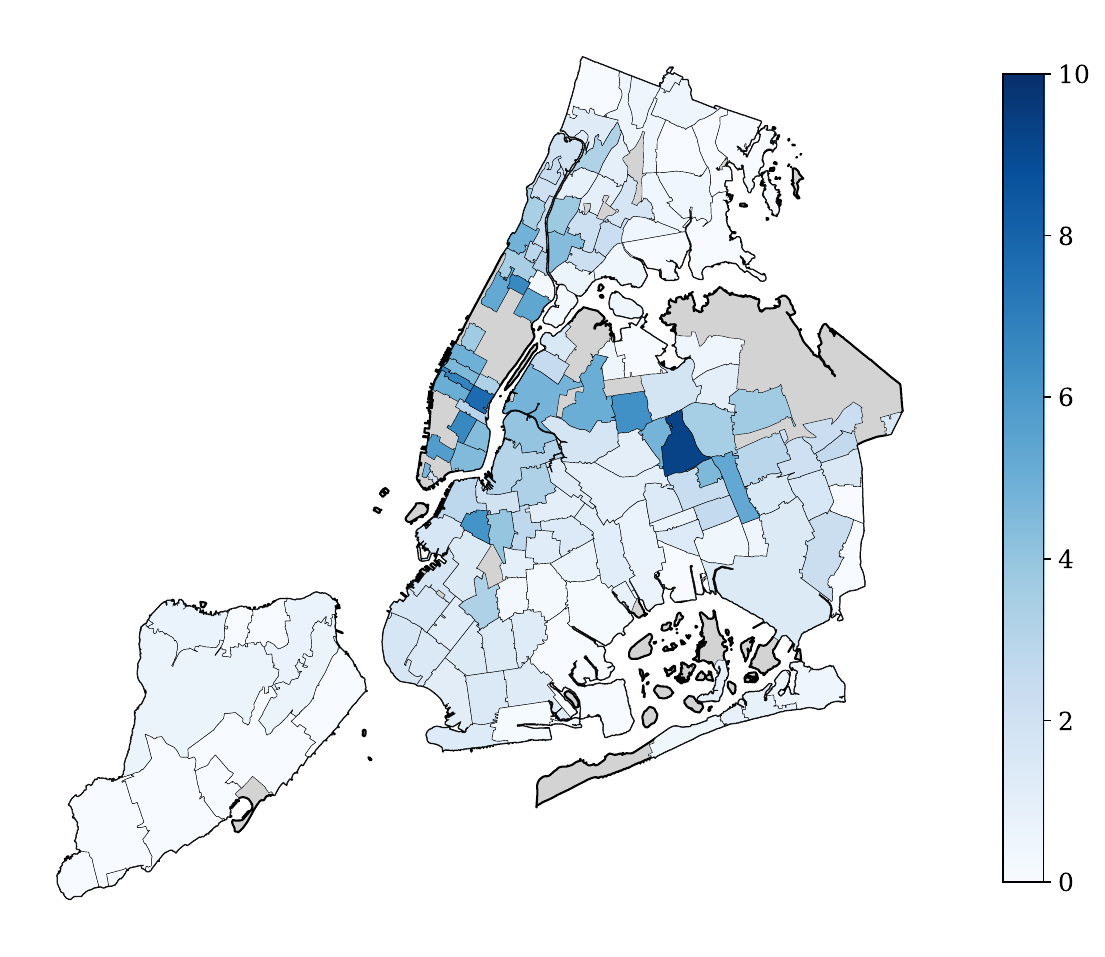}
    \caption{Mean inventory size by zip code. Gray zip codes have $\leq 1$ student who attend a target middle school.}
    \label{fig:access_disparities}
  \end{subfigure}
  \hfill
  \begin{subfigure}[t]{0.45\linewidth}
    \centering
    \includegraphics[width=0.7\linewidth]{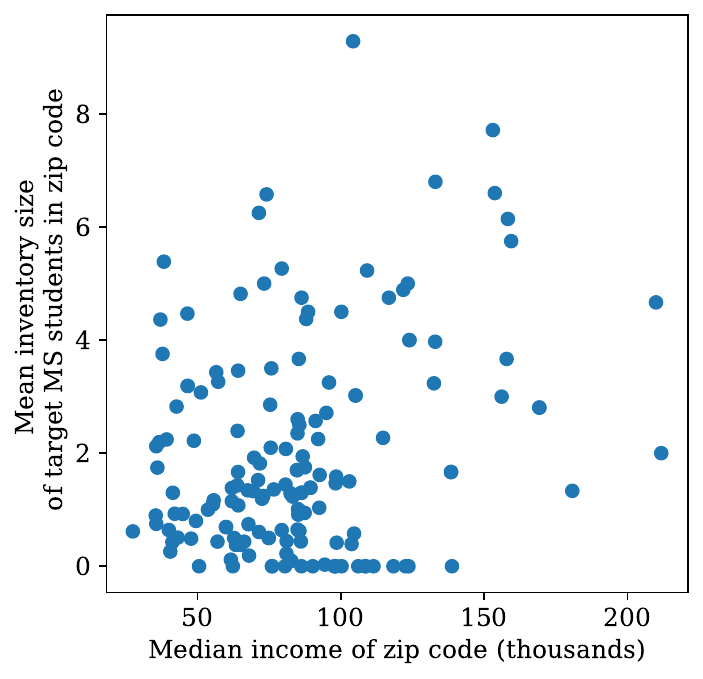}
    \caption{Zip code median income \cite{ipums_nhgis_2025} vs. mean inventory size. Pearson's r = 0.28.}
    \label{fig:income_inventory_scatter}
  \end{subfigure}
  \caption{\textbf{Disparities in access to nudge-eligible high schools.} A student's \textit{inventory size} is number of high-performing programs they are nudge-eligible for. We compute the mean inventory size of students living in each zip code (among students attending target middle schools).}
  \label{fig:congestion_and_disparities}
\end{figure}

\begin{table*}[htbp]
\vspace{0.35cm}
\centering
\small
\begin{tabular}{@{}lrrrrr@{}}\toprule
\textbf{Allocation Method} &
\makecell[r]{\textbf{Total \#}\\\textbf{Nudges}} &
\makecell[r]{\textbf{\# Students}\\\textbf{Nudged}} &
\makecell[r]{\textbf{Match}\\\textbf{Rate}} &
\makecell[r]{\textbf{\# Matched to}\\\textbf{Nudged}} &
\makecell[r]{\textbf{\# Matched to any}\\\textbf{Nudge-eligible}} \\
\midrule
\naive & 26,325 & 11,417 & 0.45 & 2,589 & 2,862 \\
\optimizelinear & 7,022 & 4,608 & 0.92 & 2,130 & 2,521 \\
\optimizesqrt & 7,021 & 6,227 & 0.80 & 2,515 & 2,853 \\
\bottomrule
\end{tabular}
\caption{\textbf{Simulation results at 50\% student adoption.} At 50\% adoption, enforcing nudge capacities continues to improve match rates with minimal effect on the total number of students matched to nudged and nudge-eligible programs.}
\label{table:table_50adoption}
\end{table*}

\begin{figure}[htbp]
  \begin{center}
    \centerline{\includegraphics[width=0.65\textwidth]{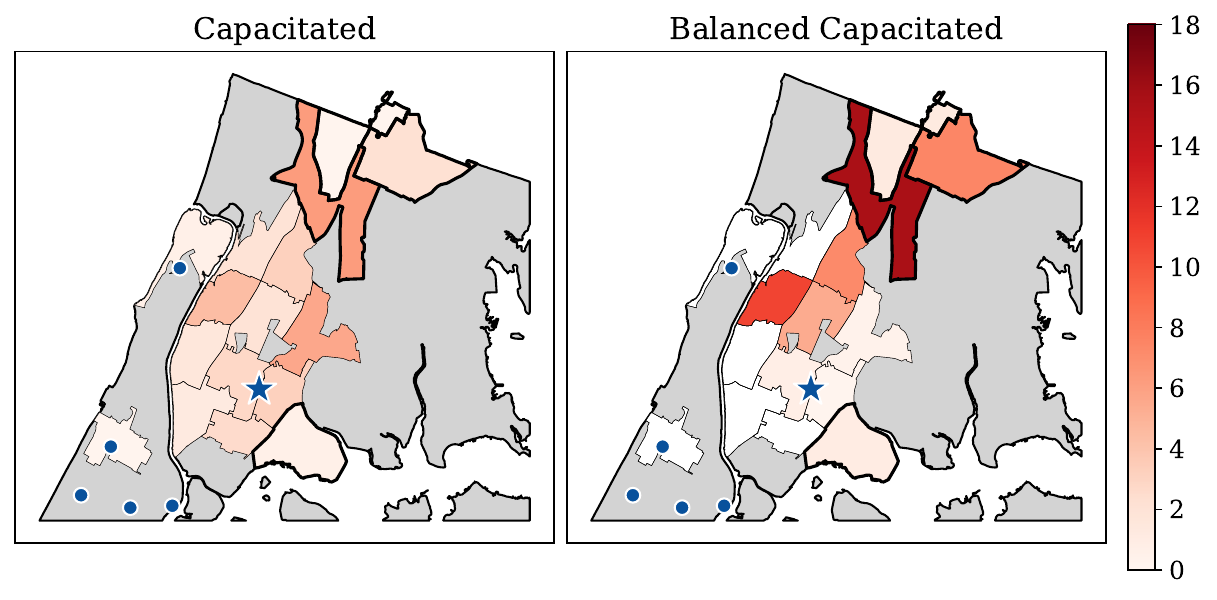}}
    \caption{\textbf{Geographic shift in nudged students.} The shade of red reflects the number of students from each zip code who are \textit{matched} to Bronx Latin after adopting a nudge for the program, under each allocation method. Nudge-eligible programs are indicated by the blue star (Bronx Latin) and dots (other). Balanced allocation results in more matches from students in zip codes \textit{without} access to any other programs (indicated with thicker outline).}
    \label{fig:Y35A_matches_shift}
    \end{center}
\end{figure}

\begin{figure}[ht]
  \centering
  \begin{subfigure}[t]{0.45\linewidth}
    \centering
    \includegraphics[width=\linewidth]{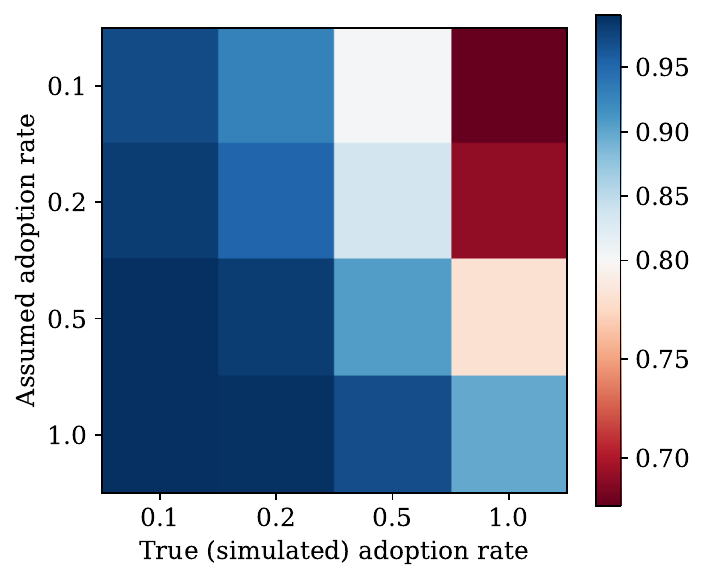}
    \caption{Program mean acceptance rate among \high offer likelihood students.}
    \label{fig:miscalibrated_nudgecaps_acceptancerate}
  \end{subfigure}
  \hfill
  \begin{subfigure}[t]{0.45\linewidth}
    \centering
    \includegraphics[width=\linewidth]{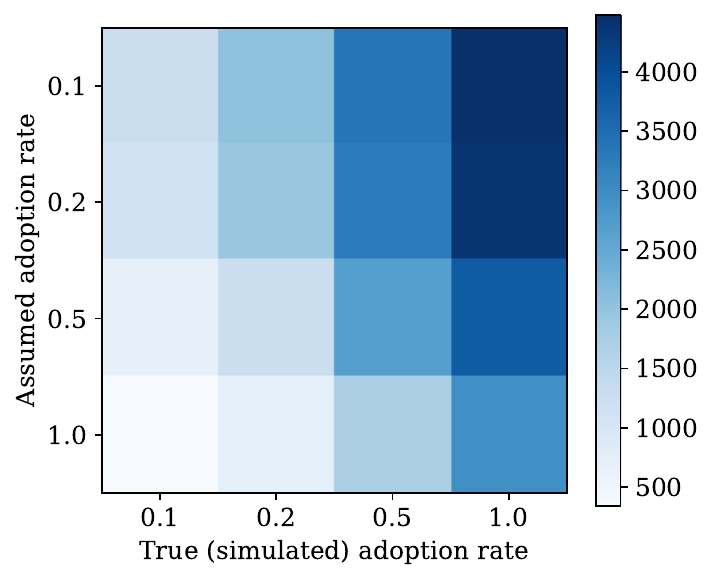}
    \caption{Total number of students matched to a nudged program.}
    \label{fig:miscalibrated_nudgecaps_nummatched}
  \end{subfigure}
  \hfill
  \caption{\textbf{Results of miscalibrated nudge capacities.} We sweep over pairs of assumed adoption rate (which informs nudge capacities) and true adoption rate, to characterize the behavior of miscalibrated nudge capacities. We find that setting nudge capacities too high or low navigates a tradeoff between \high offer likelihood acceptance rates and total intervention effectiveness (number of matches to nudged programs).}
  \label{fig:miscalibrated_nudgecaps}
\end{figure}

\FloatBarrier
\section{Supplementary deployment details}

\subsection{Details on NYC High School Match} 
\label{app:application_process_background}
New York City has about 400 high schools, attended by about 300,000 students. Each high school has at least one \textit{program}; different programs at the same school may provide different concentrations (such as arts). There are about 900 total programs, and students apply to and are admitted at the \textit{program level}. Before entering high school, students (often with the help of their families and school guidance counselors) create a ranked list of programs they wish to apply to, in an online system called MySchools. Since the 2024-25 admissions cycle, students have been able to list an unlimited number of programs on their application profile. Programs vary in admissions standards, and each program also has its own (publicly available) criteria for ranking students, with different seat availability, eligibility, and priority groups. 
We note that applications for General Education (GE) and Students with Disabilities (SWD) are run through the same process but have separate seat groups at each program. 
Once a student opens their application in MySchools, they are assigned and able to view their lottery number, which is used to break ties between students in the same priority group. After the admissions process closes in December, students are then matched to a single program on their list using student-proposing deferred acceptance, run by the NYC Public Schools (NYCPS). If a student fails to match to any program on their list, NYCPS assigns them a manual placement based on other programs they applied to and distance from home. 

The 2024-25 admissions cycle also saw the introduction of a new offer likelihood tool for students to use during the application process \cite{corradini2025school}. The tool uses a bootstrapping and simulation approach to calculate whether a student has a \high (>99\%), \medium, or \low (<1\%) chance of admission at a given program if they apply.\footnote{Not all programs have offer likelihoods available. For example, if a program uses qualitative/human graded information such as auditions/essays, then offer likelihoods are not available.} Since offer likelihood predictions are dependent on an individual student's lottery number, they are calculated and made visible to a student immediately after they begin an application in their MySchools portal.
We directly use these predictions in our work.

Note that all students have access to these offer likelihood estimates, including control students and those not eligible for the intervention. Thus, our intervention does not provide any new information (offer likelihoods, program information) to students that is not already available to them through other sources (MySchools, commute time calculators)---it simply highlights certain programs and their characteristics.  

\subsection{Data}
\label{app:data}

\paragraph{High school metric data} To determine which high school programs to involve in the intervention, we used public data characterizing schools through quantitative performance-based and survey-based metrics: graduation rates, college \& career readiness scores, and school safety ratings. We used metrics reported on the 2023-24 school year, the most recent results available at the time of the deployment. All of these metrics are reported at the \textit{school} level (even though applications are at the \textit{program} level), so the performance of all programs at a school are assumed to be the same. In practice, according to NYCPS collaborators, students within a school often have access to the same resources across programs, so considering school-level metrics are reasonable proxies for program-level metrics. We also use student-level graduation data from the 2017-18 and 2018-19 application cycles to compute program-level graduation rates and find that 91\% of programs at schools with school-level graduation rates above 95\% (as reported in the 2023-24 metric data) have a computed program-level graduation rate of at least 90\% across these two years. 

\paragraph{Middle school feeder pattern data} To determine which middle schools to involve in the intervention, we used ``feeder pattern'' data from the 2023-24 and 2024-25 admissions cycles, characterizing how many students from each middle school enroll at each high school. 

\paragraph{Student-level data}
We have access to each student's entire application, including current school, ranked application profile of program preferences, match outcome, home census tract, lottery number (used for breaking ties within an admissions priority group), and offer likelihoods using the above-described tool. We consider both General Education (GE) and Students with Disabilities (SWD), which use the same match process but have separate seat groups at each program. We developed the system based on analysis and simulations using historical data from the 2024-25 admissions cycle (72,197 students), i.e., we calculate program-level nudge capacities using 2024-25 data. Simulations on 2024-25 admissions cycle data use a student's zip code instead of census tract due to data availability constraints. We deploy our system for the 2025-26 admissions cycle, in which we had access to lottery numbers on a rolling basis, as students began applications. 

\paragraph{Match simulator} Lastly, we have access to details on each program's capacity and criteria for ranking students for the 2024-25 admissions cycle. We used this data to build a simulator that runs the deferred acceptance algorithm using these settings, meaning that we can accurately simulate match outcomes under any assumptions on student applications for this application cycle.

\subsection{Balance checks}
\label{sec:balance_checks}

We assess balance in student characteristics across treatment groups. While the realized groups appear to differ somewhat in residence borough and race composition, the distribution of p-values from corresponding balance tests is approximately uniform, consistent with valid randomization. This suggests that observed differences reflect random variation in the realized assignment (rather than systematic imbalance). To account for observed imbalances, we control for these characteristics in regressions as supplementary analyses, as shown in Tables \ref{table:all_fixed_effects_1} and \ref{table:all_fixed_effects_2}.

\begin{figure}[ht]
  \begin{center}
    \centerline{\includegraphics[width=0.6\textwidth]{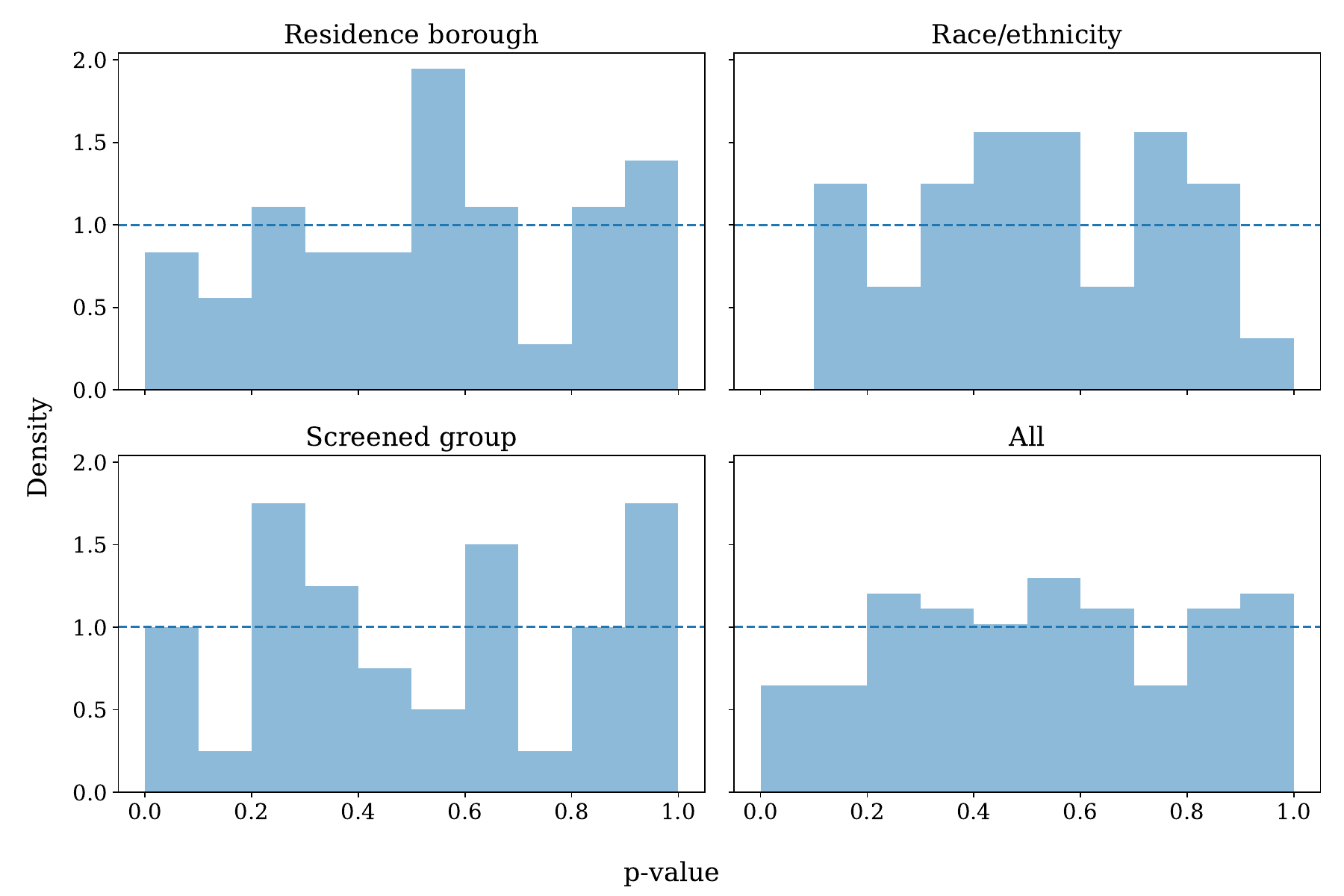}}
    \caption{\textbf{Distribution of p-values from balance tests.} 
    For each covariate category (e.g., residence borough = Queens), we use $Y$ as an indicator for the category and estimate linear probability models of the form 
    $Y \sim \text{treatment\_intensity\_50} + \text{treatment\_intensity\_100}$. 
    We estimate these models using (i) all students at target middle schools and (ii) only nudge-eligible students at those schools. 
    The figure plots the distribution of p-values from these tests. The approximately uniform distribution is consistent with differences arising from random variation under valid random assignment.}
    \label{fig:balance_tests_pvalues}
    \end{center}
\end{figure}


\begin{table}[htbp]
\centering
\begin{tabular}{lrrr}
\toprule
\textbf{Variable} & \textbf{Treatment} & \textbf{Control} & \textbf{Non-target} \\
\midrule
Total \# & 111 & 107 & 636 \\
Total \# students & 9,925 & 11,593 & 39,062 \\
\hline
Bronx & 0.41 & 0.41 & 0.19 \\
Brooklyn & 0.37 & 0.38 & 0.27 \\
Manhattan & 0.10 & 0.09 & 0.27 \\
Queens & 0.10 & 0.09 & 0.25 \\
Staten Island & 0.03 & 0.02 & 0.03 \\
\hline
Asian & 0.060 & 0.049 & 0.277 \\
Black & 0.310 & 0.324 & 0.112 \\
Hispanic/Latino & 0.525 & 0.521 & 0.382 \\
White & 0.075 & 0.079 & 0.187 \\
\hline
SWD & 0.226 & 0.228 & 0.171 \\
ELL & 0.223 & 0.208 & 0.185 \\
\hline
Mean \# nudge-eligible students & 6.9 & 7.5 & --- \\
Median \# nudge-eligible students & 4 & 5 & --- \\
\hline
Mean grade tier & 3.33 & 3.35 & 2.91 \\
\bottomrule
\end{tabular}
\caption{\textbf{Middle school characteristics.} Target middle schools (distributed within treatment and control) are those that historically did not send many students to high-performing high schools, so we expect them to systematically differ from middle schools not in the experimental group. Non-target middle schools reported here exclude charter schools (which were also filtered out for computing target middle schools); students from private middle schools and homeschool further contribute to the discrepancy between the number of students reported here and the total number of students in the match.}
\end{table}


\begin{table}[htbp]
\centering
\begin{tabular}{lrr}
\toprule
\textbf{Variable} & \textbf{Treatment} & \textbf{Control} \\
\midrule
Total \# & 444 & 822 \\
\hline
Bronx & 0.23 & 0.39 \\
Brooklyn & 0.48 & 0.41 \\
Manhattan & 0.07 & 0.09 \\
Queens & 0.22 & 0.11 \\
Staten Island & 0.00 & 0.00 \\
\hline
Asian & 0.072 & 0.063 \\
Black & 0.275 & 0.315 \\
Hispanic/Latino & 0.517 & 0.542 \\
White & 0.086 & 0.055 \\
\hline
Female & 0.538 & 0.491 \\
SWD & 0.187 & 0.223 \\
ELL & 0.169 & 0.181 \\
\hline
Screened Group 1 & 0.30 & 0.27 \\
Screened Group 2 & 0.18 & 0.20 \\
Screened Group 3 & 0.21 & 0.22 \\
Screened Group 4 & 0.16 & 0.14 \\
Screened Group 5 & 0.16 & 0.16 \\
Mean screened group & 2.7 & 2.7 \\
\bottomrule
\end{tabular}
\caption{\textbf{Nudge-eligible student characteristics.} Note that there are some imbalances between student demographic distributions between treatment and control, especially with respect to Borough. However, these do not explain our results: our primary regression specifications contain stratum fixed effects, which are defined by middle school Borough and number of nudge-eligible students in the 2024-25 admissions cycle.}
\end{table}


\begin{table}[htbp]
\centering
\begin{tabular}{lrr}
\toprule
\textbf{Variable} & \textbf{Treatment} & \textbf{Control} \\
\midrule

\multicolumn{3}{l}{\textit{Student-level measures (per student)}} \\
\hline
Mean \# in inventory & 1.65 & 1.42 \\
Median \# in inventory & 1 & 1 \\
Mean \# nudged & 1.60 & 1.35 \\
Median \# nudged& 1 & 1 \\
\hline

\multicolumn{3}{l}{\textit{Distances to nudged programs (minutes, per student--nudge pair)}} \\
\hline
Mean travel time to nudged program & 37.07 & 38.12 \\
Median travel time to nudged program & 39.22 & 39.89 \\
\hline

\multicolumn{3}{l}{\textit{Program performance among nudged programs (per student--nudge pair)}} \\
\hline
Mean graduation rate of program & 96.97 & 97.35 \\
Median graduation rate of program & 96.4 & 98.0 \\
Mean CCR of program & 70.5 & 70.2 \\
Median CCR of program & 70.4 & 70.4 \\
\hline

\multicolumn{3}{l}{\textit{Geographic distribution of nudged programs (share of student--nudge pairs)}} \\
\hline
Bronx & 0 & 0 \\
Brooklyn & 0.15 & 0.14 \\
Manhattan & 0.27 & 0.42 \\
Queens & 0.59 & 0.45 \\
Staten Island & 0 & 0 \\
\bottomrule

\end{tabular}
\caption{\textbf{Breakdown of characteristics on nudge-eligible students and nudged programs.}}
\end{table}

\FloatBarrier
\subsection{Treatment effects with fixed effects}
\label{subsec:all_fixed_effects}

\begin{table}[H]
\centering
\rotatebox{90}{
\begin{minipage}{1.35\textwidth}
\centering

\begin{minipage}{\linewidth}
\centering
\resizebox{\linewidth}{!}{%
\begin{tabular}{lcccccccccccccccc}
\toprule
 & (1) & (2) & (3) & (4) & (5) & (6) & (7) & (8) & (9) & (10) & (11) & (12) & (13) & (14) & (15) & (16) \\
\midrule
Treatment effect 
& 0.553 & 0.548 & 0.543 & 0.546 & 0.560 & --- & 0.552 & 0.557 & 0.497 & 0.493 & 0.469 & 0.472 & 0.523 & 0.518 & 0.517 & 0.521 \\
& (0.239) & (0.240) & (0.235) & (0.236) & (0.230) & --- & (0.225) & (0.227) & (0.249) & (0.250) & (0.245) & (0.246) & (0.237) & (0.238) & (0.230) & (0.232) \\
p-value & 0.021** & 0.022** & 0.021** & 0.021** & 0.015** & --- & 0.014** & 0.014** & 0.046** & 0.049** & 0.055* & 0.055* & 0.027** & 0.029** & 0.025** & 0.024** \\
\midrule
Residence borough FE & Yes & Yes & Yes & Yes & Yes & Yes & Yes & Yes \\
Ethnicity FE & Yes & Yes & Yes & Yes & & & & & Yes & Yes & Yes & Yes & & & &\\
Screened group FE & Yes & Yes & & & Yes & Yes & & & Yes & Yes & & & Yes & Yes & & \\
Gender FE & Yes & & Yes & & Yes & & Yes & & Yes & & Yes & & Yes & & Yes & \\
\bottomrule
\end{tabular}
}
\captionof{table}{\textbf{Treatment effects with alternative fixed effects specifications (Logit).} \textit{Note:} Maximum-likelihood estimation for specification (6) failed to converge. Standard errors in parentheses. *** $p<0.01$, ** $p<0.05$, * $p<0.1$.}
\label{table:all_fixed_effects_1}
\end{minipage}

\vspace{1.5cm}

\begin{minipage}{\linewidth}
\centering
\resizebox{\linewidth}{!}{%
\begin{tabular}{lcccccccccccccccc}
\toprule
 & (1) & (2) & (3) & (4) & (5) & (6) & (7) & (8) & (9) & (10) & (11) & (12) & (13) & (14) & (15) & (16) \\
\midrule
Treatment effect 
& 0.059 & 0.058 & 0.061 & 0.061 & 0.060 & 0.060 & 0.063 & 0.064 & 0.052 & 0.052 & 0.053 & 0.053 & 0.056 & 0.056 & 0.059 & 0.060 \\
& (0.026) & (0.026) & (0.027) & (0.027) & (0.026) & (0.026) & (0.027) & (0.027) & (0.028) & (0.028) & (0.029) & (0.029) & (0.027) & (0.027) & (0.028) & (0.028) \\
p-value & 0.024** & 0.026** & 0.025** & 0.025** & 0.019** & 0.021** & 0.018** & 0.018** & 0.058* & 0.061* & 0.068* & 0.068* & 0.036** & 0.038** & 0.034** & 0.034** \\
\midrule
Residence borough FE & Yes & Yes & Yes & Yes & Yes & Yes & Yes & Yes \\
Ethnicity FE & Yes & Yes & Yes & Yes & & & & & Yes & Yes & Yes & Yes & & & &\\
Screened group FE & Yes & Yes & & & Yes & Yes & & & Yes & Yes & & & Yes & Yes & & \\
Gender FE & Yes & & Yes & & Yes & & Yes & & Yes & & Yes & & Yes & & Yes & \\
\midrule
$R^2$ & 0.083 & 0.083 & 0.033 & 0.033 & 0.072 & 0.072 & 0.009 & 0.009 & 0.081 & 0.081 & 0.029 & 0.029 & 0.071 & 0.071 & 0.008 & 0.007 \\
\bottomrule
\end{tabular}
}
\captionof{table}{\textbf{Treatment effects with alternative fixed effects specifications (OLS).} Standard errors in parentheses. *** $p<0.01$, ** $p<0.05$, * $p<0.1$.}
\label{table:all_fixed_effects_2}
\end{minipage}

\end{minipage}
}
\end{table}

\FloatBarrier 

\subsection{Pre-registration compliance}
In this section, we present all procedures and analyses detailed in the pre-registration, and we describe the status of each. Nudge emails were implemented as described, with the exception that students from continuing middle schools were dropped from Batch 2, in response to stakeholder feedback. The same rule (on a feature independent from original randomization) was applied to control and treatment students. The pre-registration also stated that we would send follow-up emails to students a few weeks after their first intervention email. Due to the timing of original intervention emails (which faced a variety of unforeseen logistical constraints), follow-up emails were sent for Batch 1 students only. We also excluded students from receiving continuing middle schools from follow-up emails. Lastly, we pre-registered that we would send text messages as an additional method of communication to treatment families; this was not implemented, due to unforeseen issues with NYCPS text message platform. All other aspects of the intervention itself are in full compliance with the pre-registration. For analyses, we detail pre-registration compliance in \Cref{table:prereg_compliance}.

\begin{table}[tb]
\centering
\small
\begin{tabular}{p{0.4\linewidth} p{0.56\linewidth}}
\toprule
\textbf{Analysis} & \textbf{Status} \\
\midrule

\multicolumn{2}{l}{\textit{Primary experimental outcomes}} \\
\addlinespace
Primary outcome: applications 
& Logit results in \Cref{table:treatment_effects_logit}; OLS results in \Cref{table:treatment_effects_ols} \\
Primary outcome: matches 
& Logit results in \Cref{table:treatment_effects_logit}; OLS results in \Cref{table:treatment_effects_ols} \\
Alternative primary outcome specifications 
& \Cref{table:treatment_effects_ols} \\

\midrule
\multicolumn{2}{l}{\textit{Secondary experimental outcomes}} \\
\addlinespace
Scale-up analyses 
& \Cref{sec:scale_up} \\
Click outcomes 
& \Cref{sec:experiment_results,} \\

\midrule
\multicolumn{2}{l}{\textit{Survey outcomes}} \\
\addlinespace
Survey primary outcome 
& \Cref{table:primary_survey_outcomes} \\
Survey secondary outcomes 
& \Cref{table:descriptive_survey_results} \\
Survey exploratory outcomes 
& \Cref{table:compliers_analysis,table:info_source_regressions} \\
Survey qualitative outcomes & \Cref{sec:survey_outcomes} \\

\midrule
\multicolumn{2}{l}{\textit{Further exploratory analyses}} \\
\addlinespace
Local average treatment effects 
& \Cref{table:late_treatment_effects} \\
Heterogeneous treatment effects 
& \Cref{table:heterogeneous_treatment_effects} \\
Quantitative metrics of matched/enrolled school 
& Matched school only (enrollment data not yet available): \Cref{table:matched_prog_metrics} \\
Within-applicant comparisons of application profiles over time 
& \Cref{app:additional_experimental_results} \\
Outcomes at school level 
& Not reported; unable to compute full set of nudged students in deployment year due to data limitations \\
Long-term outcomes 
& Not yet realized \\

\bottomrule
\end{tabular}
\caption{\textbf{Pre-registration compliance for analyses.} This table summarizes the pre-registered analyses reported, along with deviations from the pre-registered plan.}
\label{table:prereg_compliance}
\end{table}

\subsection{Additional Experimental Results}
\label{app:additional_experimental_results}
In this section, we report the remaining pre-registered experimental outcomes. 
We also provide raw rates of applications and matches to nudged and nudge-eligible programs in \Cref{table:experiment_results}. ``Nudge-eligible programs'' refers to the 15 deployed nudge-eligible programs included in our experiment this year; this is distinct from nudge-eligible student-program pairs (desirable matches), which indicate that a student meets the \high offer likelihood, commute distance, and middle school feeder pattern requirements at that specific program.

\paragraph{Regression robustness checks} Our pre-registration specified that we would measure treatment effects for the primary outcome through logistic regression. For robustness, we further report results using linear probability models (\Cref{table:treatment_effects_ols}) to show that all significance results remain unchanged, with the exception of ``Ranked any nudged program in top 3'' without stratum fixed effects, which changes from being significant at the p < 0.10 level to p=0.111 with OLS. 
 
\paragraph{Additional outcome specifications.}
We report results for additional outcome specifications in \Cref{table:addtl_treatment_effects_logit} (logistic regression) and \Cref{table:treatment_effects_ols} (linear regression). First, we measure treatment effects on whether applicants ranked any \textit{nudge-eligible} program, rather than only the programs they were individually nudged to. While these treatment effects are not consistently significant across specifications (with vs. without stratum fixed effects), the behavioral changes in applying to nudged programs appears to be additive, translating to similar percentage-point increases in nudge-eligible programs.
We also report results for the two pre-registered alternative specifications of our primary outcome: whether an applicant ranked a nudged program first, and the rank position of the first nudged program on the applicant's list (conditional on ranking any nudged program).\footnote{Because rank positions are non-binary, we estimate that specification using linear models only.} We do not find a statistically significant treatment effect under either specification. 
 
\paragraph{Exploratory spillover analysis} We fit a linear regression with middle school treatment-intensity coefficients (\Cref{tab:spillovers}). Untreated students in 50\% treatment-intensity target middle schools exhibited higher application rates to nudged programs than students in 0\% schools. While this pattern is consistent with possible within-school spillovers, it may also reflect residual school-level imbalance or noise, since the spillover estimate is identified from between-school variation and a relatively small number of untreated students in partially treated schools.

\paragraph{Local Average Treatment Effects} As preregistered, we estimate local average treatment effects (LATE) using treatment assignment as an instrument for email engagement, defined as clicking at least one tracked link. By this definition, 19.8\% of treatment students complied with the intervention. However, our results indicate that many students likely read and engaged with the email without clicking on any link, and so the exclusion restriction is violated---these effects are likely overstated. While we report this analysis in \Cref{table:late_treatment_effects} to be consistent with the preregistration, we do not rely on them.

\paragraph{Heterogeneous treatment effects} For each subgroup variable (FRL/ELL/STH status, residence borough, screened group, race/ethnicity, gender), we estimate a regression with an interaction coefficient between treatment assignment and the feature: $Y_i \sim $ treatment*feature. We report results in \Cref{table:heterogeneous_treatment_effects}. We find no evidence of heterogeneous treatment effects, though our RCT was not powered to detect such effects. 
 
\paragraph{Additional analysis of link-clicking behavior} We find that the mean number of link clicks per treated applicant is 0.502, and on average applicants click on links for 18.2\% of their nudged programs. We also examine whether link click behaviors vary across subgroups, regressing an indicator for clicking any link on demographic characteristics among treated applicants. Overall, we find limited evidence of heterogeneous engagement, although we were underpowered to do so (\Cref{tab:click_heterogeneity}). 

\paragraph{Comparisons between pre-intervention and submitted application lists} We also examine saved student application lists from before the intervention, to study how and whether applicants revise their existing lists in response to nudges. We find that only 10 of 444 treatment applicants had any saved programs as of the data pulls one day before their nudge email was sent (this number is 18 of 822 for control applicants). None of these 10 applicants had a nudged program on their pre-intervention saved list, and none added a nudged program to their final application. This low rate of adding schools to the portal is consistent with our intervention only targeting applicants with in-progress, un-submitted applications---but given the small and potentially biased sample of applicants who had already added schools to their application upon receiving the email, means we have limited ability to study student behavior in adding schools to an existing list.

\FloatBarrier 

\begin{figure}[tb]
  \centering

  \begin{subfigure}[t]{0.24\linewidth}
    \vspace{0pt}
    \centering
    \includegraphics[width=\linewidth]{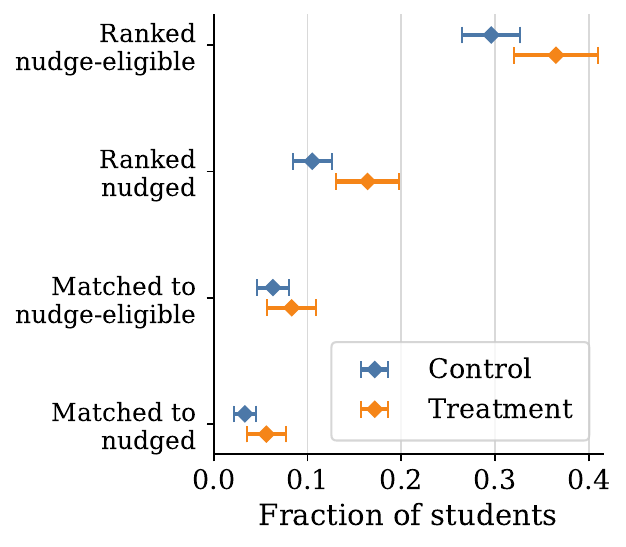}
    \label{fig:student_behavior_forest_full}
  \end{subfigure}
  \hfill
  \begin{subfigure}[t]{0.75\linewidth}
  \centering
  \begin{minipage}[t][2.95cm][t]{\linewidth}
    \vspace{2pt}
    \centering
    \small
    \resizebox{\linewidth}{!}{%
      \begin{tabular}{@{}llrrrrrrr@{}}
        \toprule
        &&&& \multicolumn{4}{c}{\textbf{Ranked program\ldots}} & \\
        \cmidrule(l){5-8}
        \textbf{Group} & \textbf{Total} & \textbf{Submitted} & \textbf{Program type} & 
        \emph{Anywhere} & \emph{First} & \emph{Second} & \emph{Third} & \textbf{Matched}\\
        \midrule
        
        \multirow{2}{*}{Control}
        & \multirow{2}{*}{822}
        & \multirow{2}{*}{803 (97.7\%)}
        & \textit{Nudged}
        & 86 (10.5\%) & 19 (2.3\%) & 11 (1.3\%) & 11 (1.3\%) & 27 (3.3\%) \\
        
        & 
        &
        & \textit{Nudge-eligible}
        & 243 (29.6\%) & 53 (6.4\%) & 36 (4.4\%) & 29 (3.5\%) & 52 (6.3\%) \\
        
        \addlinespace
        
        \multirow{2}{*}{Treatment}
        & \multirow{2}{*}{444}
        & \multirow{2}{*}{439 (98.9\%)}
        & \textit{Nudged}
        & 73 (16.4\%) & 17 (3.8\%) & 9 (2.0\%) & 9 (2.0\%) & 25 (5.6\%) \\
        
        &
        &
        & \textit{Nudge-eligible}
        & 162 (36.5\%) & 36 (8.1\%) & 21 (4.7\%) & 19 (4.3\%) & 37 (8.3\%) \\
        
        \bottomrule
        \end{tabular}
    }
  \end{minipage}
\end{subfigure}
\vspace{-0.5cm}
\caption{\textbf{Summary of student behavior in response to nudges.} For both control and treatment groups, we summarize the number of students total, who ultimately submitted an application, and who ranked (anywhere on the list or in any of the first 3 spots) or matched to a nudged or nudge-eligible program in their application.}
\label{table:experiment_results}
\label{fig:behavior_summary}
\end{figure}


\begin{table}[tb]\centering
\vspace{1cm}
\small
\caption{\textbf{Regression-estimated treatment effects (Logit): additional outcome specifications}}
\label{table:addtl_treatment_effects_logit}
\begin{tabular}{lccccl}
\toprule

\multirow{2}{*}{\textbf{Outcome}} & 
\textbf{Control} & \textbf{Treatment} & \multirow{2}{*}{\textbf{Stratum FE}} & \textbf{Treatment} & \multirow{2}{*}{\textbf{p-value}} \\
& \textbf{mean} & \textbf{mean} & & \textbf{effect} & \\
\midrule

\multirow{2}{*}{Ranked any nudge-eligible program}
& \multirow{2}{*}{0.296}
& \multirow{2}{*}{0.365}
& Yes & 0.226 (0.139) & 0.104 \\
& & & No & 0.314 (0.155) & 0.042** \\

\multirow{2}{*}{Ranked any nudged program first}
& \multirow{2}{*}{0.023}
& \multirow{2}{*}{0.038}
& Yes & 0.638 (0.456) & 0.162 \\
& & & No & 0.520 (0.418) & 0.213 \\

\multirow{2}{*}{Ranked any nudged program in top 3}
& \multirow{2}{*}{0.050}
& \multirow{2}{*}{0.079}
& Yes & 0.545 (0.298) & 0.068* \\
& & & No & 0.489 (0.292) & 0.094* \\

\multirow{2}{*}{Matched to any nudge-eligible program}
& \multirow{2}{*}{0.063}
& \multirow{2}{*}{0.083}
& Yes & 0.279 (0.294) & 0.344 \\
& & & No & 0.297 (0.276) & 0.281 \\

\bottomrule
\end{tabular}

\begin{flushleft}
\footnotesize
Logistic regressions. Robust standard errors clustered at the school level in parentheses.
Treatment is randomly assigned within strata; regressions are estimated with and without stratum fixed effects (borough and school size).
Control and treatment means are the average outcomes among control and treatment students in the experimental population, respectively.
Coefficients are log-odds ratios for each outcome. There are 1,266 total observations.
* $p<0.10$, ** $p<0.05$
\end{flushleft}
\end{table}


\begin{table}[htbp]\centering
\small
\caption{\textbf{Regression-estimated treatment effects (OLS): primary outcomes and additional outcome specifications}}
\label{table:treatment_effects_ols}
\begin{tabular}{lccccl}
\toprule

\multirow{2}{*}{\textbf{Outcome}} & 
\textbf{Control} & \textbf{Treatment} & \multirow{2}{*}{\textbf{Stratum FE}} & \textbf{Treatment} & \multirow{2}{*}{\textbf{p-value}} \\
& \textbf{mean} & \textbf{mean} & & \textbf{effect} & \\
\midrule

\multirow{2}{*}{Ranked any nudge-eligible program}
& \multirow{2}{*}{0.296}
& \multirow{2}{*}{0.365}
& Yes & 0.049 (0.030) & 0.105 \\
& & & No & 0.069 (0.034) & 0.041** \\

\multirow{2}{*}{Ranked any nudged program}
& \multirow{2}{*}{0.105}
& \multirow{2}{*}{0.164}
& Yes & 0.064 (0.026) & 0.013** \\
& & & No & 0.060 (0.028) & 0.034** \\

\multirow{2}{*}{Ranked any nudged program first}
& \multirow{2}{*}{0.023}
& \multirow{2}{*}{0.038}
& Yes & 0.0185 (0.014) & 0.188 \\
& & & No & 0.0152 (0.013) & 0.244 \\

\multirow{2}{*}{Ranked any nudged program in top 3}
& \multirow{2}{*}{0.050}
& \multirow{2}{*}{0.079}
& Yes & 0.032 (0.018) & 0.079* \\
& & & No & 0.029 (0.018) & 0.111 \\

\multirow{2}{*}{Matched to any nudge-eligible program}
& \multirow{2}{*}{0.063}
& \multirow{2}{*}{0.083}
& Yes & 0.019 (0.020) & 0.353 \\
& & & No & 0.020 (0.019) & 0.292 \\

\multirow{2}{*}{Matched to any nudged program}
& \multirow{2}{*}{0.033}
& \multirow{2}{*}{0.056}
& Yes & 0.027 (0.016) & 0.095* \\
& & & No & 0.023 (0.016) & 0.136 \\

\multirow{2}{*}{Rank position of first nudged program}
& \multirow{2}{*}{3.91}
& \multirow{2}{*}{4.16}
& Yes & -0.008 (0.604) & 0.990 \\
& & & No & 0.257 (0.625) & 0.681 \\

\bottomrule
\end{tabular}

\begin{flushleft}
\footnotesize
Linear probability models. Robust standard errors clustered at the school level in parentheses.
Treatment is randomly assigned within strata; regressions are estimated with and without stratum fixed effects (Borough and school size).
Control and treatment means are the average outcomes among control and treatment students in the experimental population, respectively.
When coefficients represent changes in the probability, multiplying by 100 gives the percentage-point effect. There are 1,266 total observations for all regressions except ``Rank position of first nudged program'', for which there are 159 total observations. * $p<0.10$, ** $p<0.05$
\end{flushleft}
\end{table}


\begin{table}[tb]\centering
\vspace{0.5cm}
\small
\caption{\textbf{Spillover effects}}
\label{tab:spillovers}
\begin{tabular}{lccl}
\toprule

 & \textbf{Stratum FE} & \textbf{Coefficient} & \textbf{p-value} \\
\midrule

\multirow{2}{*}{Treatment intensity 50}
& Yes & 0.0837 & 0.044** \\
& No & 0.0778 & 0.079* \\

\multirow{2}{*}{Treatment intensity 100}
& Yes & 0.0502 & 0.276 \\
& No & 0.0436 & 0.374 \\

\multirow{2}{*}{Treatment}
& Yes & 0.0213 & 0.530 \\
& No & 0.0198 & 0.556 \\

\bottomrule
\end{tabular}
\begin{centering}
\\
\footnotesize
\vspace{0.1cm}
Linear probability model with robust standard errors clustered at the school level. * $p<0.10$, ** $p<0.05$
\end{centering}
\end{table}


\begin{table}[htbp]\centering
\small
\caption{\textbf{Local Average Treatment Effects for primary outcomes and alternative outcome specifications.} We define an observed complier as a treatment applicant who clicks on at least one tracked link in their intervention email. We note that these effects are likely overstated---more people read the email than clicked on a tracked link, so not all compliers were observed. We do not track email open rates directly.}
\label{table:late_treatment_effects}
\begin{tabular}{lccccl}
\toprule

\multirow{2}{*}{\textbf{Outcome}} & 
\textbf{Control} & \textbf{Observed} & \textbf{Stratum} & \multirow{2}{*}{\textbf{LATE}} & \multirow{2}{*}{\textbf{p-value}} \\
& \textbf{mean} & \textbf{complier mean} & \textbf{FE} & & \\
\midrule

\multirow{2}{*}{Ranked any nudge-eligible program}
& \multirow{2}{*}{0.296}
& \multirow{2}{*}{0.500}
& Yes & 0.253 (0.144) & 0.078* \\
& & & No & 0.349 (0.138) & 0.012** \\

\multirow{2}{*}{Ranked any nudged program}
& \multirow{2}{*}{0.105}
& \multirow{2}{*}{0.284}
& Yes & 0.333 (0.104) & 0.001*** \\
& & & No & 0.302 (0.098) & 0.002*** \\

\multirow{2}{*}{Ranked any nudged program first}
& \multirow{2}{*}{0.023}
& \multirow{2}{*}{0.091}
& Yes & 0.096 (0.052) & 0.065* \\
& & & No & 0.077 (0.049) & 0.120 \\

\multirow{2}{*}{Ranked any nudged program in top 3}
& \multirow{2}{*}{0.050}
& \multirow{2}{*}{0.159}
& Yes & 0.167 (0.074) & 0.024** \\
& & & No & 0.146 (0.070) & 0.038** \\

\multirow{2}{*}{Matched to any nudge-eligible}
& \multirow{2}{*}{0.063}
& \multirow{2}{*}{0.136}
& Yes & 0.098 (0.080) & 0.220 \\
& & & No & 0.101 (0.076) & 0.182 \\

\multirow{2}{*}{Matched to any nudged program}
& \multirow{2}{*}{0.033}
& \multirow{2}{*}{0.114}
& Yes & 0.138 (0.062) & 0.026** \\
& & & No & 0.118 (0.059) & 0.044** \\

\bottomrule
\end{tabular}

\begin{flushleft}
\footnotesize
Local average treatment effects estimated using two-stage least squares (2SLS) linear probability models, with randomized treatment assignment as an instrument for clicking at least one tracked link in the nudge email. Robust standard errors clustered at the school level in parentheses.
Treatment is randomly assigned within strata; regressions are estimated with and without stratum fixed effects (Borough and school size).
Control and observed complier means are the average outcomes among control students and observed compliers in the experimental population, respectively.
Coefficients represent probability changes; multiplying by 100 gives the percentage-point effect. There are 1,266 total observations. * $p<0.10$, ** $p<0.05$, *** $p< 0.01$
\end{flushleft}
\end{table}


\begin{table}[htbp]
\centering
\small
\begin{tabular}{lcl}
\toprule
\textbf{Feature} & \textbf{Coef.} & \textbf{p-value} \\
\midrule
\multicolumn{3}{l}{\textit{FRL/ELL/STH status}} \\
\addlinespace
Treatment (Non-FRL/ELL/STH) & 0.0176 & 0.745 \\
Treatment $\times$ FRL/ELL/STH & 0.0499 & 0.336\\

\midrule
\multicolumn{3}{l}{\textit{Residence borough}} \\
\addlinespace
Treatment (Brooklyn) & 0.053 & 0.186 \\
Treatment $\times$ Manhattan & -0.021 & 0.794 \\
Treatment $\times$ Queens & -0.055 & 0.430 \\
Treatment $\times$ Bronx & 0.076 & 0.259 \\

\midrule
\multicolumn{3}{l}{\textit{Screened group}} \\
\addlinespace
Treatment (Group 1) & 0.091 & 0.121 \\
Treatment $\times$ Group 2 & 0.002 & 0.979 \\
Treatment $\times$ Group 3 & -0.092 & 0.119 \\
Treatment $\times$ Group 4 & -0.047 & 0.444 \\
Treatment $\times$ Group 5 & -0.046 & 0.492 \\

\midrule
\multicolumn{3}{l}{\textit{Race/ethnicity}} \\
\addlinespace
Treatment (Hispanic) & 0.039 & 0.266 \\
Treatment $\times$ Asian & 0.038 & 0.767 \\
Treatment $\times$ Black & 0.040 & 0.395 \\
Treatment $\times$ White & -0.038 & 0.689 \\
Treatment $\times$ Multi-racial & 0.161 & 0.138 \\

\midrule
\multicolumn{3}{l}{\textit{Gender}} \\
\addlinespace
Treatment (Female) & 0.076 & 0.033** \\
Treatment $\times$ Male & -0.034 & 0.345 \\

\bottomrule
\end{tabular}
\caption{\textbf{Heterogeneous treatment effects.}
Outcome is whether the applicant applied to a nudged program.}
\label{table:heterogeneous_treatment_effects}
\end{table}


\begin{table}[htbp]
\centering
\caption{\textbf{Heterogeneity in link-clicking behavior among treated applicants}}
\label{tab:click_heterogeneity}
\begin{tabular}{lcccccc}
\toprule
& \multicolumn{2}{c}{\textbf{Any link click}} 
& \multicolumn{2}{c}{\textbf{\# link clicks}} 
& \multicolumn{2}{c}{\textbf{Frac. links clicked}} \\
\cmidrule(lr){2-3} \cmidrule(lr){4-5} \cmidrule(lr){6-7}
\textbf{Subgroup variable} 
& \emph{Coef.} & \emph{$p$-value} 
& \emph{Coef.} & \emph{$p$-value} 
& \emph{Coef.} & \emph{$p$-value} \\

\midrule
\multicolumn{3}{l}{\textit{FRL/ELL/STH status}} \\
\addlinespace
FRL/ELL/STH & -0.0840 & 0.099 & -0.143 & 0.424 & -0.0739 & 0.147\\
\midrule
\multicolumn{3}{l}{\textit{Screened group}} \\
\addlinespace
Group 2 & 0.0525 & 0.335 & -0.0423 & 0.802 & 0.0463 & 0.397 \\
Group 3 & 0.0085 & 0.890 & -0.0094 & 0.966 & -0.0039 & 0.948 \\
Group 4 & 0.0605 & 0.357 & 0.0782 & 0.750 & 0.0411 & 0.522 \\
Group 5 & 0.0246 & 0.715 & 0.0468 & 0.840 & 0.0278 & 0.664 \\
\midrule
\multicolumn{3}{l}{\textit{Race/ethnicity}} \\
\addlinespace
Black & -0.0559 & 0.428 & -0.1034 & 0.497 & -0.0684 & 0.333 \\
Hispanic & -0.0085 & 0.905 & 0.1628 & 0.307 & -0.0282 & 0.684 \\
Multi-racial & 0.2227 & 0.033 & 0.2549 & 0.171 & 0.2227 & 0.033 \\
White & 0.0310 & 0.722 & 0.2988 & 0.230 & -0.0019 & 0.983 \\
\bottomrule
\end{tabular}

\vspace{0.5em}
\begin{flushleft}
\footnotesize
Each row reports the coefficient from a separate linear regression among treated applicants. FRL/ELL/STH indicates whether a student is classified as free/reduced-price lunch eligible, an English language learner, or a Student in Temporary Housing. The screened group (based on student grades) regressions use Group 1 as the reference category. The race/ethnicity regressions use Asian students as the reference category. 
\end{flushleft}
\end{table}


\begin{table}[htbp]\centering
\vspace{1cm}
\small
\caption{\textbf{Treatment effect on matched school metrics.} Reported metrics are those used to define high-performing schools.}
\label{table:matched_prog_metrics}
\begin{tabular}{lcccr}
\toprule

\textbf{Metric} & \textbf{Control mean} & \textbf{Stratum FE} & \textbf{Coef.} & \textbf{p-value} \\
\midrule

\multirow{2}{*}{Graduation Rate}
& \multirow{2}{*}{90.67}
& Yes & 0.27 (0.7) & 0.675 \\
& & No & 0.17 (0.8) & 0.820 \\

\multirow{2}{*}{College and Career Readiness}
& \multirow{2}{*}{49.78}
& Yes & 1.79 (1.16) & 0.122 \\
& & No & 2.19 (1.38) & 0.111 \\

\multirow{2}{*}{Safety}
& \multirow{2}{*}{3.62}
& Yes & 0.025 (0.047) & 0.589 \\
& & No & -0.031 (0.070) & 0.662 \\

\bottomrule
\end{tabular}

\begin{flushleft}
\footnotesize
Linear models. Robust standard errors clustered at the school level in parentheses.
Treatment is randomly assigned within strata; regressions are estimated with and without stratum fixed effects (Borough and school size).
Control mean is the average outcome among control students in the experimental population. Metrics are not available for all schools, and matches to those schools are omitted; this results in 1,228 total observations for Graduation Rate and College and Career Readiness regressions, and 1,191 for Safety regressions. * $p<0.10$, ** $p<0.05$, *** $p< 0.01$
\end{flushleft}
\end{table}


\FloatBarrier
\subsection{Additional survey results}
\label{sec:survey_outcomes}

We analyze post-application survey responses to measure whether the nudge emails increased applicants' familiarity with high-performing high school programs near their home, as pre-registered. We sent the survey to all treatment and control students from middle schools retained over both batches (370 and 707, respectively), and an equal number of non-treatment students, for a total of 2,154 total students. 

Our survey had 42 treatment respondents, 61 control respondents, and 136 respondents from non-target middle schools. For each student, we sample one nearby, high-performing high school program to ask more detailed questions about; for treatment and control students this program is randomly sampled from nudged programs they \textit{applied to} if any exist, and from all nudged programs otherwise. For non-target students, the program is randomly sampled from all nearby, high-performing high school programs. 

\Cref{table:primary_survey_outcomes} summarizes the primary survey outcomes (distinct from primary experimental outcomes): on the survey sample, we do not find evidence of treatment effect on self-reported familiarity or rates of recognition on sampled programs. \Cref{table:descriptive_survey_results} details all other survey responses, reporting raw survey rates for other questions. Due to small sample size, these results should be interpreted as descriptive evidence, potentially providing insight for future interventions. 

We also report outcomes specifically among students who report reading the email (\Cref{table:compliers_analysis}) and run exploratory regressions by applicants' self-reported primary information sources for high school options (\Cref{table:info_source_regressions}).

Finally, we conducted a thematic analysis of open-ended survey responses. Among families who answered why they did not apply to the sampled high-performing program, most responses referred either to the program's specific academic interest area or the fact that it was too far from home.
Only 2 of 25 families provided a free-response answer regarding their perception of the email, so we do not draw conclusive insights from this question. In response to ``what advice would have helped you navigate the application process better'', the most common themes were ``nothing'', more school/program information, insight on admissions chances, personalized recommendations, clarity on the process and admissions requirements, and better support for different languages or students with individualized education programs. 

\begin{table}[htbp]
\vspace{0.5cm}
\centering
\small
\begin{tabular}{p{0.54\linewidth}lcc}
\toprule
\textbf{Outcome} & \textbf{Model} & \textbf{Coefficient} & \textbf{$p$-value} \\
\midrule

Self-reported familiarity with high-performing programs near home (very unfamiliar (1) -- very familiar (4) scale)
& Ordered logit
& -0.25
& 0.517 \\
& OLS
& -0.11
& 0.590 \\
\addlinespace
Familiar with sampled high-performing program
& Logit
& 0.368
& 0.416 \\
& OLS
& 0.089
& 0.409 \\
\bottomrule
\end{tabular} \\
\caption{\textbf{Primary survey outcomes.} Familiarity with high-performing programs (self-reported familiarity and rates of recognition on sampled programs), comparing treatment and control respondents. Ordered logit and logit coefficients are reported in log-odds units. For self-reported familiarity, the OLS coefficient is measured in points on the 1–4 scale. For the binary familiarity outcome, the OLS coefficient can be interpreted as a change in probability; multiplying by 100 gives the percentage-point effect. 103 total observations.}
\label{table:primary_survey_outcomes}
\end{table}


\definecolor{veryimp}{RGB}{76, 114, 176}   
\definecolor{imp}{RGB}{85, 168, 104}       
\definecolor{notimp}{RGB}{129, 114, 179}   
\definecolor{noresp}{RGB}{200, 200, 200}   
\def\labelpad{0.1}

\newcommand{\importancebar}[4]{
\begin{tikzpicture}[baseline=-0.6ex]
    \def\barwidth{6.5}
    \def\barheight{0.28}

    \pgfmathsetmacro{\wA}{\barwidth * #1 / 239}
    \pgfmathsetmacro{\wB}{\barwidth * #2 / 239}
    \pgfmathsetmacro{\wC}{\barwidth * #3 / 239}
    \pgfmathsetmacro{\wD}{\barwidth * #4 / 239}

    \fill[veryimp] (0,-\barheight/2) rectangle (\wA,\barheight/2);
    \fill[imp] (\wA,-\barheight/2) rectangle (\wA+\wB,\barheight/2);
    \fill[notimp] (\wA+\wB,-\barheight/2) rectangle (\wA+\wB+\wC,\barheight/2);
    \fill[noresp] (\wA+\wB+\wC,-\barheight/2) rectangle (\barwidth,\barheight/2);

    \ifnum#1>14
        \node[anchor=west, font=\tiny\bfseries, text=white, inner sep=0pt]
        at (\labelpad, 0) {#1};
    \fi

    \ifnum#2>14
        \node[anchor=west, font=\tiny\bfseries, text=white, inner sep=0pt]
        at (\wA+\labelpad, 0) {#2};
    \fi

    \draw[black!20] (0,-\barheight/2) rectangle (\barwidth,\barheight/2);
    
\end{tikzpicture}%
}


\begin{table}[htbp]
\centering
\caption{\textbf{Descriptive survey outcomes.} We report raw survey rates for other outcomes. These are exploratory and could inform future interventions.}
\label{table:descriptive_survey_results}
\small
\begin{tabular}{
p{0.46\linewidth}
>{\centering\arraybackslash}p{0.12\linewidth}
>{\centering\arraybackslash}p{0.13\linewidth}
>{\centering\arraybackslash}p{0.14\linewidth}
}
\toprule
\textbf{Outcome} & \textbf{Treatment} & \textbf{Control} & \textbf{Non-target} \\
\midrule

\multicolumn{4}{l}{\textit{Panel A. Reaction to email (Treatment group only)}} \\
\addlinespace
Read the email
& 59.5\% (25/42)
& --
& --\\
Email was helpful (of applicants who read email)
& 88.0\% (22/25)
& --
& --\\
Email affected application decisions (of applicants who read email)
& 44.0\% (11/25)
& --
& --\\
Perception of sampled program improved after reading (of applicants who read email and recognized sampled prog)
& 66.7\% (10/15)
& --
& --\\

Would \textit{not} have applied \textit{without} the email (of applicants who applied to sampled prog)
& 72.7\% (8/11)
& --
& --\\

\midrule

\multicolumn{4}{l}{\textit{Panel B. Knowledge about programs (Treatment \& control groups)}} \\
\addlinespace
Recognized sampled program
& 52.4\% (22/42)
& 45.9\% (28/61)
& --\\
Email readers: knowledge \textit{before} email \\
\quad Sampled prog is at high-performing school 
& 33.3\% (5/15)
& --
& --\\
\quad \high likelihood at sampled prog 
& 13.3\% (2/15)
& --
& --\\
\quad Sampled prog is within 45 min of home
& 20.0\% (3/15)
& --
& --\\

Non-email readers: knowledge at end of process\\
\quad Sampled prog is at high-performing school 
& 85.7\% (6/7)
& 89.3\% (25/28)
& --\\
\quad \high likelihood at sampled prog 
& 71.4\% (5/7)
& 57.1\% (16/28)
& --\\

\midrule
\multicolumn{4}{l}{\textit{Panel C. Demand for recommendations (Control \& non-target groups)}} \\
\addlinespace
Would have considered applying to sampled prog
& --
& 64.8\% (35/54)
& 56.9\% (62/109)\\
Interested in receiving recommendations
& --
& 90.2\% (55/61)
& 88.1\% (118/134)\\

\midrule
\multicolumn{4}{l}{\textit{Panel D. Reasons for not applying to sampled program (All groups)}} \\
\multicolumn{4}{l}{\quad \textit{Multiple choice among applicants who recognized sampled prog but did not apply}}\\
\addlinespace
Did not know enough about it
& 22.2\% (2/9)
& 31.8\% (7/22)
& 21.4 \% (9/42) \\
Did not match child's interests
& 44.4\% (4/9)
& 31.8\% (7/22)
& 57.1\% (24/42) \\
Did not think child would get in
& 22.2\% (2/9)
& 9.1\% (2/22)
& 16.7\% (7/42)\\
Too far from home
& 22.2\% (2/9)
& 9.1\% (2/22)
& 4.8\% (2/42)\\

\midrule
\addlinespace
\multicolumn{4}{l}{\textit{Panel E. Feature importance for recommendations (Aggregated across all groups)}} \\
\addlinespace
\multicolumn{4}{l}{
\scriptsize
\textcolor{veryimp}{\rule{0.8em}{0.8em}} Very important \quad
\textcolor{imp}{\rule{0.8em}{0.8em}} Important \quad
\textcolor{notimp}{\rule{0.8em}{0.8em}} Not important \quad
\textcolor{noresp}{\rule{0.8em}{0.8em}} No response
}\\
\addlinespace
Safety
& \multicolumn{3}{l}{\importancebar{178}{44}{4}{13}}\\
Academic performance
& \multicolumn{3}{l}{\importancebar{137}{88}{4}{10}}\\
Distance from home
& \multicolumn{3}{l}{\importancebar{108}{104}{15}{12}}\\
Program theme or interest area
& \multicolumn{3}{l}{\importancebar{105}{117}{9}{8}}\\
High chance of offer
& \multicolumn{3}{l}{\importancebar{98}{120}{7}{14}}\\
Extracurricular offerings
& \multicolumn{3}{l}{\importancebar{81}{136}{12}{10}}\\

\bottomrule
\end{tabular}
\begin{flushleft}
\footnotesize The treatment, control, and non-target columns report raw survey outcomes. The non-target group is included for descriptive context only and is not part of the pre-registered treatment-control comparison. Some exploratory outcomes are conditional on responding to follow-up questions, so denominators vary across rows.
\end{flushleft}
\end{table}


\begin{table}[htbp]
\centering
\caption{\textbf{Outcomes by applicants' self-reported email readership (exploratory)}}
\label{table:compliers_analysis}
\small
\begin{tabular}{p{0.52\linewidth}cc}
\toprule
\textbf{Outcome} & \textbf{Read email} & \textbf{Did not read email} \\
\midrule
Self-reported familiarity with high-performing programs near home (very unfamiliar (1) -- very familiar (4) scale)
& 2.72 $\pm$ 0.37
& 3.06 $\pm$ 0.52 \\
Familiar with sampled high-performing program
& 60.0\% (15/25)
& 41.2\% (7/17) \\
Applied to at least one nudged program
& 44.0\% (11/25)
& 29.4\% (5/17) \\
Applied to at least one nudge-eligible program
& 64.0\% (16/25)
& 58.8\% (10/17) \\
\bottomrule
\end{tabular}
\begin{flushleft}
\footnotesize This table conditions on treatment assignment and splits respondents by whether they self-reported reading the email. Results for self-reported familiarity include 95\% confidence intervals.
\end{flushleft}
\end{table}


\begin{table}[htbp]
\vspace{0.5cm}
\centering
\caption{\textbf{Regressions for outcomes by applicants' primary information sources (exploratory)}}
\label{table:info_source_regressions}
\small
\begin{tabular}{p{0.2\linewidth}rrrr}
\toprule
\textbf{Info source} 
& \textbf{Familiarity}
& \textbf{Read email}
& \textbf{Applied to nudged prog}
& \textbf{Interested in recs} \\
\midrule
Intercept
& 2.930 (0.119)
& 0.449 (0.135)
& 0.156 (0.045)
& 0.853 (0.056) \\
\addlinespace
Other parents
& -0.047 (0.124)
& -0.246 (0.170)
& -0.071 (0.035)
& -0.033 (0.052) \\
Paid consultants
& -0.018 (0.472)
& -0.764 (0.204)
& -0.051 (0.033)
& 0.126 (0.076) \\
Counselors
& -0.016 (0.104)
& 0.162 (0.137)
& -0.006 (0.037)
& 0.026 (0.050) \\
MySchools
& 0.032 (0.104)
& -0.066 (0.169)
& -0.035 (0.037)
& -0.005 (0.054) \\
Other websites
& 0.062 (0.167)
& -0.180 (0.218)
& -0.038 (0.040)
& -0.054 (0.076) \\
Information sessions
& 0.077 (0.122)
& 0.399 (0.114)
& -0.017 (0.041)
& 0.082 (0.043) \\

\midrule
Observations & 239 & 42 & 239 & 197 \\
\bottomrule
\end{tabular}
\begin{flushleft}
\footnotesize This table reports coefficients from exploratory regressions with indicators for applicants' self-reported primary information sources. Outcomes are: self-reported familiarity with high-performing programs near home on a 1--4 scale (all respondents); whether the applicant reported reading the email (treatment respondents only); whether the applicant applied to at least one nudged program (all respondents); and whether the applicant reported interest in receiving recommendations (control and non-target respondents only). Standard errors are reported in parentheses. These regressions are descriptive because information sources are self-reported and not randomly assigned.
\end{flushleft}
\end{table}

\FloatBarrier
\clearpage
\subsection{Scale-up analyses}
\begin{figure}[ht]
  \centering
  \begin{subfigure}[t]{0.33\linewidth}
    \centering
    \includegraphics[width=\linewidth]{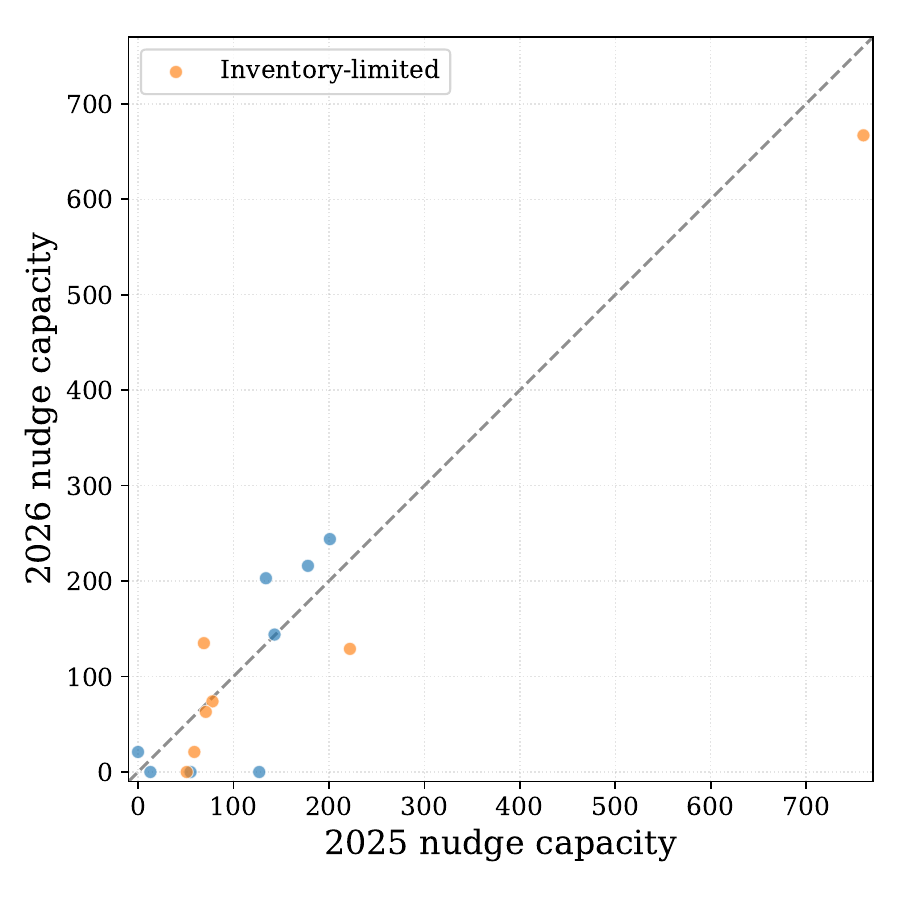}
    \caption{\textbf{40\% nudge adoption}}
    \label{fig:nudgecap_comparison_40}
  \end{subfigure}
    \hspace{0.03\linewidth}
  \begin{subfigure}[t]{0.33\linewidth}
    \centering
    \includegraphics[width=\linewidth]{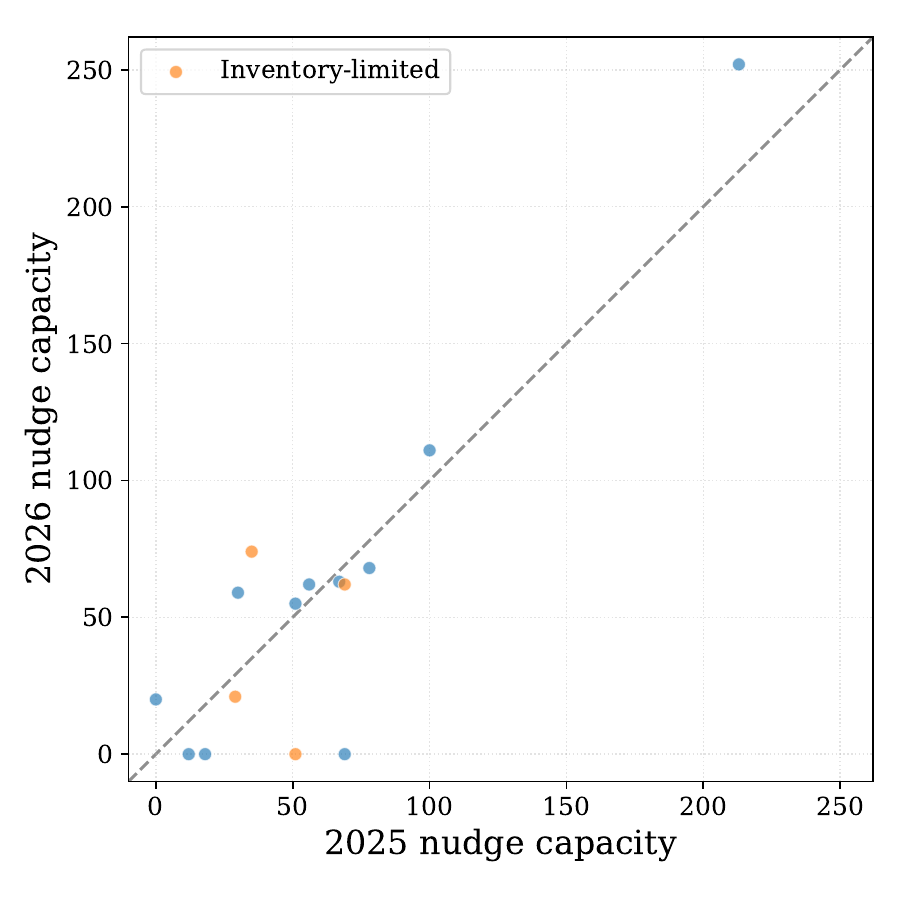}
    \caption{\textbf{100\% nudge adoption}}
    \label{fig:nudgecap_comparison_100}
  \end{subfigure}
  \hfill
  \caption{\textbf{Comparison of nudge capacities between years.} We compute nudge capacities on 2024-25 versus 2025-26 application data, assuming two simulated adoption rates. ``Inventory-limited'' orange dots indicate programs were limited by the number of nudge-eligible students rather than congestion.}
  \label{fig:nudgecap_comparisons}
\end{figure}

\FloatBarrier
\subsection{Other figures}

\begin{figure}[ht]
  \begin{center}
    \centerline{\includegraphics[width=0.6\textwidth]{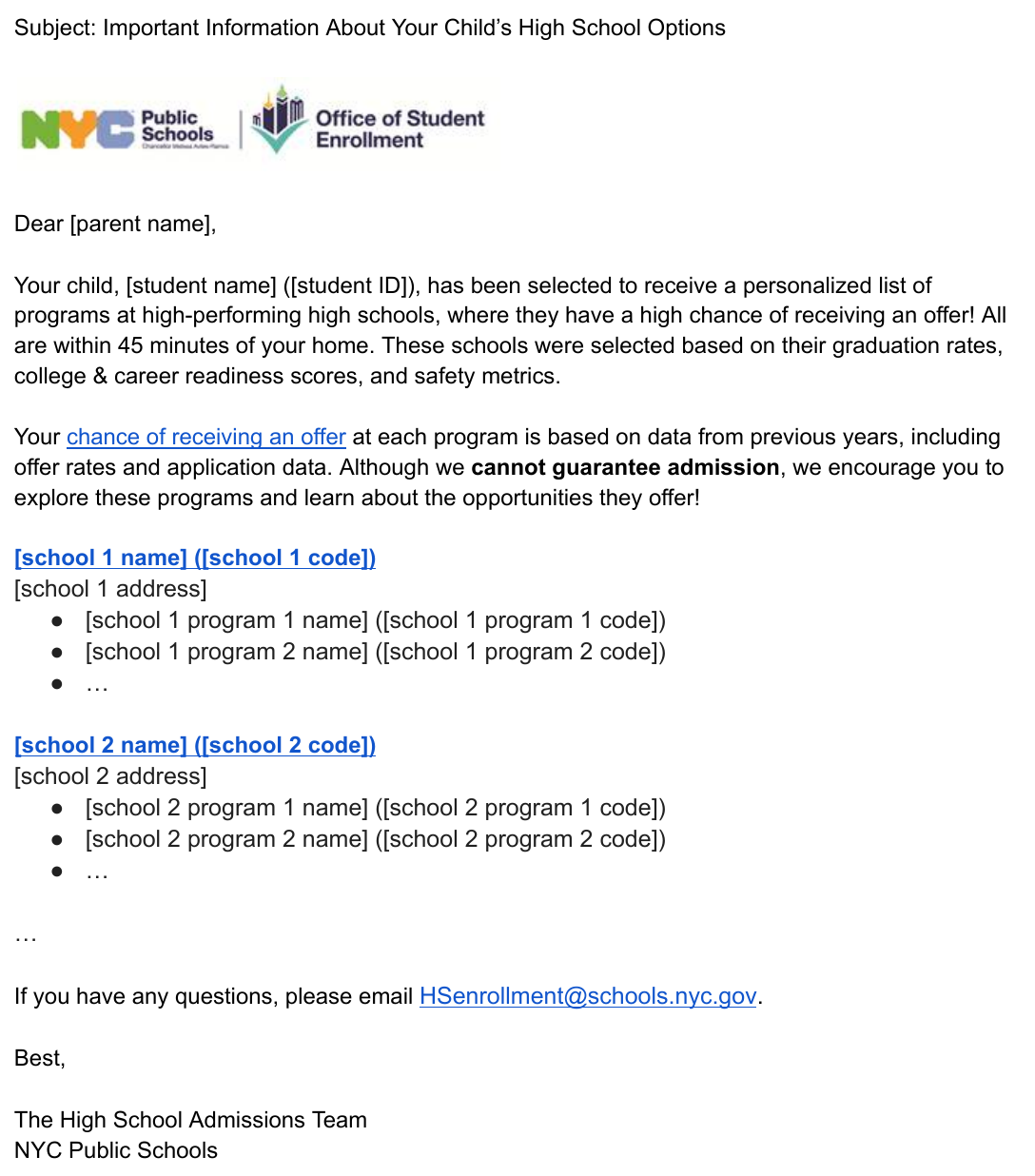}}
    \caption{\textbf{Nudge email text.}}
    \label{fig:email_text}
    \end{center}
\end{figure}\renewcommand{\arraystretch}{1.2}

\end{document}